\documentclass[11pt]{article} 
\usepackage{fullpage}
\usepackage{graphicx}
\usepackage{upref}
\usepackage{enumerate}
\usepackage{latexsym}
\usepackage{pdfpages}
\usepackage[bookmarks]{hyperref}
\usepackage{color,graphics}
\usepackage{comment} 
\usepackage{caption}
\usepackage{subcaption}
\usepackage{wrapfig}
\usepackage{braket}
\newcommand{\ketbra}[2]{| #1 \rangle\!\langle #2|}

\usepackage{amssymb}
\usepackage{amsmath}
\usepackage{bbm}
\usepackage{amsthm}
\usepackage{mathrsfs}
\usepackage{mathtools}

\usepackage{epsfig,graphicx}
\usepackage{algorithmic}
\usepackage[ruled,linesnumbered]{algorithm2e}
\usepackage{amsfonts}
\usepackage{tikz} 
\usepackage{varioref}
\usepackage[ansinew]{inputenc}
\usepackage{qcircuit}
\usepackage{authblk}
\usepackage{multirow}
\usepackage{lineno}

\theoremstyle{plain}
\newtheorem{theorem}{Theorem}
\theoremstyle{plain}
\newtheorem{lemma}[theorem]{Lemma}
\theoremstyle{plain}

\newtheorem{corollary}[theorem]{Corollary}
\theoremstyle{plain}

\theoremstyle{plain}

\theoremstyle{plain}

\theoremstyle{definition}

\newtheorem{definition}[theorem]{Definition}
\theoremstyle{definition}
\newtheorem{fact}[theorem]{Fact}

\theoremstyle{remark}
\newtheorem{remark}{Remark}[section]

\theoremstyle{definition}

\AtBeginDocument{}

\DeclareMathOperator{\real}{\mathbb{R}}
\DeclareMathOperator{\nat}{\mathbb{N}}
\newcommand{\cmplx}{\mathbb{C}}
\newcommand{\intg}{\mathbb{Z}}

\newcommand{\id}{\mathbb{I}}

\newcommand{\Z}{\text{Z}}

\newcommand{\vect}[1]{\mathbf{#1}}

\newcommand{\diag}{\text{diag}}

\newcommand{\ad}{\text{ad}}
\newcommand{\sel}{\text{SELECT}}
\newcommand{\prep}{\text{PREP}}

\newcommand{\nconst}{\mathcal{A}}
\newcommand{\term}{\mathcal{N}}
\newcommand{\g}{\mathcal{G}}
\newcommand{\q}{\mathbf{q}}

\begin{document}

\title{Quantum Simulation of the First-Quantized Pauli-Fierz Hamiltonian }

\author[1]{Priyanka Mukhopadhyay \thanks{mukhopadhyay.priyanka@gmail.com, priyanka.mukhopadhyay@utoronto.ca}}
\author[2,3]{Torin F. Stetina\thanks{torin.stetina@gmail.com, torins@berkeley.edu}}
\author[4,5,6]{Nathan Wiebe \thanks{nawiebe@cs.toronto.edu}}

\affil[1]{Department of Physical \& Environmental Sciences, University of Toronto, ON, Canada}
\affil[2]{Simons Institute for the Theory of Computing, Berkeley, CA, USA}
\affil[3]{Berkeley Quantum Information and Computation Center,
University of California, Berkeley, USA}
\affil[4]{Department of Computer Science, University of Toronto, ON, Canada}
\affil[5]{Pacific Northwest National Laboratory, Richland, WA, USA}
\affil[6]{Canadian Institute for Advanced Research, Toronto, ON, Canada}

\date{}
\maketitle

\begin{abstract}
We provide an explicit recursive divide and conquer approach for simulating quantum dynamics and derive a discrete first quantized non-relativistic QED Hamiltonian based on the many-particle Pauli Fierz Hamiltonian.  We apply this recursive divide and conquer algorithm to this Hamiltonian and compare it to a concrete simulation algorithm that uses qubitization.  Our divide and conquer algorithm, using lowest order Trotterization, scales for fixed grid spacing as $\widetilde{O}(\Lambda N^2\eta^2 t^2 /\epsilon)$ for grid size $N$, $\eta$ particles, simulation time $t$, field cutoff $\Lambda$ and error $\epsilon$.  Our qubitization algorithm scales as $\widetilde{O}(N(\eta+N)(\eta +\Lambda^2) t\log(1/\epsilon)) $.  This shows that even a na\"ive partitioning and low-order splitting formula can yield, through our divide and conquer formalism, superior scaling to qubitization for large $\Lambda$.  We compare the relative costs of these two algorithms on systems that are relevant for applications such as the spontaneous emission of photons, and the photoionization of electrons. We observe that for different parameter regimes, one method can be favored over the other. Finally, we give new algorithmic and circuit level techniques for gate optimization including a new way of implementing a group of multi-controlled-X gates that can be used for better analysis of  circuit cost. 
\end{abstract}



\section{Introduction}
\label{sec:intro}
The prospect of simulating quantum systems is a highly anticipated application for fault-tolerant quantum computers of the future. The inception of this application of quantum computation is typically attributed to Richard Feynman in the 1980s \cite{1982_F}. Since then, there has been a flurry of both theoretical and experimental research on Hamiltonian simulation algorithms \cite{2007_BACS, 2012_CW, 2015_BCCKS, 2017_LC, 2018_HP, 2019_LC, 2019_GSLW, 2021_YSLetal, 2022_DBKetal, 2022_RRW} and specific applications ranging from condensed matter physics \cite{2012_CMLS, 2017_KWBA, 2018_HQ, 2019_KMvBetal}, chemistry \cite{2015_BLKetal, 2015_WHWetal, 2019_BBMN, 2018_BGBetal, 2021_SBWetal}, high energy particle physics \cite{2013_MRRZ, 2015_ZCR, 2020_LX, 2020_SLSW, 2021_NPDB, 2022_TAMetal}, quantum gravity \cite{2009_MCH, 2017_GELetal, 2022_FHHetal, 2022_SSdJetal, 2023_MHAetal}, and much more \cite{2005_DCRetal, 2014_GAN, 2019_YEZetal, 2020_BBMC, 2020_HFCN, 2021_WLTetal, 2022_KRS}. Research in these applications of simulating physics has shown a variety of challenges specific to each regime of interest, and the subtleties of the benefits and limitations of the select Hamiltonian simulation algorithms have become more apparent as progress has been moving forward. 

In this work, we focus on the non-relativistic regime of chemistry, and condensed matter which is a very active field in the development of quantum algorithms. Specifically we use a first-quantized approach to simulating the many electron degrees of freedom, due to their favorable sublinear asymptotic scaling in the number of orbitals or grid points, which is usually much larger than the number of electrons.~\cite{2008_KJLMA, 2021_SBWetal} Typically, quantum simulations of chemistry primarily focus on the Coulomb Hamiltonian for electrons, which includes one and two body interactions and classical clamped nuclei using the Born-Oppenheimer approximation. While this work is important for understanding many chemical properties including chemical reaction rates, with both qualitative and quantitative success, there are many basic and applied problems where the fundamental nature of the quantum electromagnetic (EM) field is important. Thus, we would like to treat electrons and the EM field on even footing, where both have quantum degrees of freedom. One example where this is important is in cavity quantum electrodynamics.~\cite{2002_MD, 2011_B, 2013_RDBE, 2017_FRAR} Here, atomic or molecular systems are placed in a mirrored cavity, increasingly coupling the matter system to the fundamental EM mode defined by the cavity size, to the point where electronic and photonic states combine into so called ``polaritonic'' states. Obviously, the properties of this system cannot be properly modeled with electron only Hamiltonians, requiring explicit quantum degrees of freedom for the EM field. Another active area of work is in attosecond science, where experiment and theory are actively investigating the short time dynamics of electron motion after photoexcitation \cite{2016_RLN, 2010_SFKetal, 2010_KI, 2010_BM, 2011_MLPetal, 2011_NPFetal, 2012_PFNB, 2014_FZNetal, 2018_OM, 2019_VDL, 2021_OM}. Here, there are still many unanswered questions about how the electrons move in the short time after interacting with light, but the complicated light-matter correlations make it difficult to model theoretically.

Overall, the dynamical properties of quantum EM fields interacting with many electron systems is still poorly understood, but there is significant basic and applied scientific motivation to push our understanding further in this field. One of the main goals of understanding this complex interplay of quantum electrodynamics will be to ``actively control and manipulate electrons on the attosecond time and angstrom length scale'' \cite{2015_PNB}.  In order to attempt to simulate this on a fault-tolerant quantum computer, we must add the proper degrees of freedom to account for the quantum EM field. To simulate non-relativistic quantum electrodynamics we utilize the multi-electron Pauli-Fierz Hamiltonian (sometimes referred to as the non-relativistic quantum electrodynamical (NRQED) Hamiltonian), which is described in detail in Section \ref{sec:ham}. In short, the physics of the Pauli-Fierz Hamiltonian modifies the electronic only one-body momentum term from the Coulomb Hamiltonian to include a minimal coupling description of the light-matter interaction, and retains the standard two-body Coulomb electronic interaction, with a free EM dynamical field term as well.



\paragraph{Our results and contributions : } In this paper we describe a couple of approaches for the Hamiltonian simulation of a first quantized full NRQED simulation of light matter interactions using the first quantized Pauli-Fierz Hamiltonian discretized on a lattice. 

(I) In Section \ref{sec:ham} we describe the derivation of the first quantized general spin-1/2 Pauli-Fierz Hamiltonian for $\eta$ particles given in \cite{2004_S}. The real-space is discretized onto a lattice, with a truncation of the electric field Hilbert space. According to our knowledge, this is the first derivation of the many body Pauli-Fierz Hamiltonian in first quantization described in the literature. We consider two approaches to simulate this Hamiltonian $\hat{H}_{PF}$.

(II) First we consider simulation using a recursive divide and conquer approach (Algorithm-I), improving on the technique introduced in \cite{2018_HP}. Here we divide the given Hamiltonian into several fragments using Trotter-Suzuki formulae \cite{1991_S, 2021_CSTetal}, simulate each of them separately, possibly using different algorithms, and then combine the results. In \cite{2018_HP} the authors used Trotterization for each fragment. In this paper we have combined Trotterization with qubitization. Such approaches can be very useful if we want to exploit the best of many worlds. For some Hamiltonians, especially those with commuting terms, Trotterization gives less gate complexity. But it has a super-polynomial scaling of error tolerance. On the other hand, qubitization has a logarithmic dependence on the inverse of tolerable error, but the gate complexity depends on the $\ell_1$ norm of the coefficients when the Hamiltonian is expressed as sum of unitaries. For many complicated Hamiltonians, it may be difficult to simulate with one particular existing technique. In such scenarios, the divide-and-conquer approach can be very helpful. Such divide-and-conquer type of approaches have shown their value in \cite{2021_HHKL, 2023_LSTT}, where the focus has been on simulation of specific local Hamiltonians. Our approach is more general and can be applied to a broader spectrum of Hamiltonians, in order to achieve better complexity of simulation. 

In Section \ref{subsec:sim} we describe the divide-and-conquer algorithm and derive a bound on the gate complexity (Theorem \ref{thm:DC}). In later Appendix \ref{subsec:H12}-\ref{subsec:H32} we describe in detail the simulation of each of the partitions of $\hat{H}_{PF}$. 

(III) The second algorithm (Algorithm-II) that we consider is to use qubitization \cite{2017_LC, 2019_LC, 2019_GSLW}. For this we block encode the entire Hamiltonian $\hat{H}_{PF}$. In Section \ref{subsec:divConqBlock} we describe a divide-and-conquer approach to construct the block encoding of sum and product of different Hamiltonians. We show that for many situations, it is advantageous in terms of number of gates, when we split the Hamiltonian into separate parts, block encode each of them and then combine these. For both our algorithms such recursive block encoding has been useful. We describe Alogrithm-II in Section \ref{subsec:totalQubit} and provide a bound on the gate complexity (Theorem \ref{thm:qub}). 

(IV) Both these algorithms have their own pros and cons and depending on the Hamiltonian under consideration, one can be favored over the other for different parameter regimes. To illustrate more on this, in Section \ref{sec:apps} we have compared the relative costs of these two algorithms compared to some model system of interest. For example, we consider a regime of a small number of electrons in a single atom system, that is relevant for applications like spontaneous emission of photons into the field, photoionization of electrons and photoelectric effect. Roughly, comparing Theorem \ref{thm:DC} and \ref{thm:qub} we find that both these algorithms have a quadratic dependence on the lattice size $N$. While qubitization has a quadratic dependence on the electric cut-off $\Lambda$, Divide and Conquer shows a sub-quadratic dependence. The complexity of the latter depends on the partitioning and in this paper we have tried to prioritize $\Lambda$. This reflects when we compare the cost in Figures \ref{fig:neonNg} and \ref{fig:cost_ratio_lambda}. We observe that qubitization scales better with respect to $N$, while Divide and Conquer performs better with respect to $\Lambda$. Another interesting phenomenon we have observed is the fact that as we increase order of the Trotter splitting in Divide and Conquer, the scaling become closer to qubitization.

We have also discussed a few possible applications for simulating the Hamiltonian considered by us, for example, the determination of photoionization timescales in atomic, molecular and extended systems. Further, we have discussed how the electric cut-off, one of the parameters of interest, scales for certain regimes of applied problems. 

(V) On the circuit synthesis and optimization side, we develop a split-and-merge technique (Theorem \ref{thm:CX}) to implement a group of multi-controlled-X gates in Section \ref{subsec:divConqBlock} (Appendix \ref{app:CX}). Such group of gates occur in many places, for example, Hamiltonian simulation algorithms working with linear combination of unitaries \cite{2015_BCCKS, 2017_LC, 2019_LC, 2007_BACS, 2012_CW}, synthesizing efficient circuits for exponentiated Paulis \cite{2023_MWZ}, Quantum Approximate Optimization Algorithm (QAOA) \cite{2022_TAS}, quantum state preparation \cite{2005_MVBS}, quantum machine learning \cite{2021_SP}, construction of QROM \cite{2018_BGBetal} and QRAM \cite{2008_GLM}. The main intuition is to split and group the control qubits, use extra ancillae to store intermediate information and then implement the requisite logical function using these ancillae. We show that this can lead to an asymptotic improvement in the gate complexity of SELECT operations, by shaving of logarithmic factors. Such circuit optimization technique may be of independent interest and may be useful for other applications. 

(VI) Among other technical contributions, we give improved decompositions of certain matrices as linear combination of unitaries. Specifically, we give general procedures to decompose diagonal integer matrices as sum of exponentially less number of unitaries. This also contributes to the asymptotic improvement in gate complexity. In Appendix \ref{app:lcu} we describe these decompositions and the computation of $\ell_1$ norm of Hamiltonian. It has been shown in \cite{2021_CSTetal} that the Trotter error depends on nested commutators. In Lemma \ref{lem:alpha_comm} we show that these nested commutators depend on pair-wise commutators and sum of the $\ell_1$ norm. The derivations of these terms have also been shown in Appendix \ref{app:comm}. We hope that these technical contributions will be useful in future works for better analysing the complexity of simulating Hamiltonians. 

\section{Results}
\label{sec:result}

Here we review the main results of our paper and provide an extended introduction to the physics of the Pauli-Fierz Hamiltonian.  The  Pauli-Fierz Hamiltonian gives a proper non-relativistic treatment of single particle quantum electrodynamics. 
 This is frequently augmented to the multi-particle case by including artificial Coulomb interactions between the particles resulting in a Hamiltonian that is more general than the standard Hamiltonians studied in quantum chemistry simulation.  While the Pauli-Fierz Hamiltonian is a well studied model, it is typically presented in a second quantized form.  We will first review the derivation of its first quantized form which we will need in order to have a simulation algorithm whose scaling is comparable to the best known simulation results for chemistry in absentia of electrodynamical effects.


\subsection{Pauli-Fierz Hamiltonian}
\label{sec:ham}


In order to simulate the Pauli-Fierz Hamiltonian, we must first discretize the real-space onto a lattice and provide a truncation for the electric field Hilbert space. We will denote this cutoff as $\Lambda$ and discretize the space as a cubic lattice with side length $L$. $N$ is the total number of grid points and so in each Cartesian direction there are $N^{1/3}$ grid points. A single grid point, $\q$, can then be described as $\q=(q_x,q_y,q_z) \in [0,N^{1/3}-1]^3 \subset \mathbb{Z}^3_+$. We write $\q$ varies from $1$ to $N$ for brevity, instead of $q_x,q_y,q_z$ vary from $0$ to $N^{1/3}-1$. $\q+1_{\mu}$ refers to the adjacent point of $\q$ in the $\mu^{th}$ direction, i.e. it is obtained by adding the lattice spacing to $q_{\mu}$. We often write $\q+\vect{1}$ to refer to an adjacent point, when the direction is clear from the context or when we want to refer to all the 3 neighbouring points of $\q$. We write $(q,\mu)$ to refer to the link connecting point $\q$ to its adjacent point in the $\mu^{th}$ direction. We drop the bracket in subscripts  if direction is clear from the context then we drop the second index.

In first quantized representation, the particle number is fixed and each particle has its own ``copy'' of the grid where it lives. Subsequently, each first quantized particle interacts with the background field separately. The discretized Hilbert space for the Pauli-Fierz Hamiltonian is then 
\begin{align}
\mathcal{H}_{\text{PF}} &= \mathcal{H}_p \otimes \mathcal{H}_f \\
&= L^2(\mathbbm{C}^{(2N)^\eta \left(2\Lambda \right)^{3N}}),
\label{eqn:HPF1}
\end{align}
where at each electric link between grid point $\q$ and $\q+\vect{1}$, there are $2\Lambda + 1$ possible electric link values. Recall that we have 3 links per grid point in a periodic basis. However, for practical implementation, the link space will be offset by one, so the total dimension of the Hilbert space at each link is even, $2\Lambda$. Collecting the notation into one place for this manuscript, we will use the following definitions as described in Table \ref{tab:terms}.
\begin{table}[]
\centering
 \begin{tabular}{|c|c|}
 \hline
  \textbf{Term} & \textbf{Definition} \\
  \hline
  $\eta$ & Number of particles in the simulation \\ 
  \hline 
  $e$ & Bare electric charge \\ 
  \hline
  $m_e$ & Electron mass \\ 
  \hline
  $N$ & Number of lattice sites \\ 
  \hline
  $G$ & Set of lattice sites labeled $q$ for a 3D cubic lattice where $q \in [0,N^{1/3}]^3 \subset \mathbb{Z}^3_+$\\ 
  \hline
  $L$ & Length of one side of the simulation box where $\{L \in \mathbb{R} \, |\, L > 0 \}$ \\ 
  \hline
  $\Omega = L^3$ & Volume of box size $L$ \\ 
  \hline
  $\Delta = \frac{\Omega^{1/3}}{N^{1/3}}$ & Lattice spacing size \\ 
  \hline
  $\Lambda$ & Max cutoff for electric link quantum number \\ 
  \hline
  $\mu, \nu$ & Cartesian indices\\ 
  \hline
  $\sigma_{\mu}$ & The $\mu$th Pauli matrix\\ 
  \hline
  $A_{\mu}(q)$ & $\mu$th component of the magnetic vector potential at link site $q$\\ 
  \hline
 \end{tabular}
\caption{List of important variable definitions used throughout this manuscript.}
\label{tab:terms}
\end{table}

The general expression for the Hamiltonian on the $N$-point cubic lattice with $\eta$ electrons is,
\begin{equation}
\hat{H} = H_{\pi} + H_{s} + H_{V_{ee}} + H_{V_{ne}} + H_f. \label{eqn:HPF0} 
\end{equation}
Throughout this paper we often refer to a summand Hamiltonian as a `fragment Hamiltonian', each of which we will describe now. For convenience, we first describe the registers on which the operators act. The state of the qubits in the registers gives the wavefunction. There are two registers - the particle register and link register.  We store the spin and position of each of the $\eta$ particles in the particle register. To be precise, for each particle we allot 1 qubit to store the spin and $3\cdot\log_2N^{1/3}=\log_2N$ qubits to store the Cartesian coordinates of its position in the lattice grid. Thus the particle register is of the form $\bigotimes_{j=1}^{\eta}\ket{s,\vect{q}}_j$, where $\vect{q}=(q_x,q_y,q_z)$ and it has $\eta\left(1+\log_2N\right)$ qubits. Again, we assume a max cutoff for the electric link eigenstates, $\Lambda$. In the link register, we allot $\log_2(2\Lambda)$ qubits for each of the 3 links per grid point. Thus there are $3N\log_2(2\Lambda)$ qubits in the link space. In later sections, when we describe the simulation algorithms, we will mention that in each register we keep extra qubits for selecting a subspace on which an operator acts, but these do not reflect on the state of the wavefunction. 

Now we describe each term of the  fragment Hamiltonian in Equation \ref{eqn:HPF0}. For a more detailed background on the derivation of these Hamiltonian terms, see App.~\ref{app:ham_deriv}  The operators act on 3 disjoint subspaces corresponding to particle spin, particle position and gauge link space. First we describe the fragment Hamiltonians, $H_{V_{ee}}, H_{V_{ne}}$, that involve only the particle degrees of freedom acting on the $\mathcal{H}_p$ Hilbert space.
\begin{align}
H_{V_{ee}} &= \id\otimes\left(\frac{e^2}{4 \pi \epsilon_{0} \Delta} \sum_{i < j}^{\eta} \sum_{\vect{q},\vect{r}}^N \frac{1}{ \,||\vect{q} - \vect{r}||_2} \, | \vect{q} \rangle \langle \vect{q} |_i \,| \vect{r} \rangle \langle \vect{r} |_j\right)\otimes\id \\
H_{V_{ne}}& = \id\otimes\left(-\frac{e^2}{4 \pi \epsilon_0 \Delta} \sum_{i}^{\eta} \sum_{\kappa}^{K} \sum_{\vect{q}}^N \frac{Z_{\kappa}}{ \,||\vect{q} - \vect{R}_{\kappa}||_2} \, | \vect{q} \rangle \langle \vect{q} |_i \right)\otimes\id 
\end{align}
$H_{V_{ee}}$ and $H_{V_{ne}}$ capture the instantaneous Coulomb interaction terms between two electrons and the attractive term between an electron and a classical fixed point charge representing a nucleus, respectively. $| \vect{q} \rangle \langle \vect{q} |_i$ is shorthand notation for the operator acting only on particle $i$ over the $\eta$ particle Hilbert space, $\left(\bigotimes_{k=1}^{i-1}\id_k\right)\otimes | \vect{q} \rangle \langle \vect{q} |_i \otimes \left(\bigotimes_{k=i+1}^{\eta}\id_k\right)$. Additionally, $\kappa$ indexes $K$ classical nuclei at real space coordinate $\vect{R}_{\kappa}$.

The free electromagnetic field Hamiltonian, $H_f$, acting on the EM field Hilbert space $\mathcal{H}_f$ is described as,
\begin{equation}
H_f =\id\otimes\id\left( \sum_{q=1}^{N}\sum_{\mu=1}^3 \frac{e^2}{2} E_{q,\mu}^2 -\sum_{q=1}^N\sum_{\mu\neq\nu=1}^3 \frac{1}{e^2} W_{q,\mu,\nu}^2\right),
\end{equation}
 where each link $q$ connects points adjacent to point $\vect{q}$ in the lattice. If the inner summation index or subscript of an operator is $\mu$, then link $q$ connects point $\q$ to its neighbour in the $\mu^{th}$ direction of the lattice. We will explain shortly what the double subscripts $\mu,\nu$ imply in case of operator $W_{q,\mu,\nu}$.
 
 In the electric link basis the $E_{q,\mu}$ operators are defined as, 
\begin{align}
    E_{q,\mu}    &= \sum_{\epsilon=-\Lambda}^{\Lambda-1} \epsilon |\epsilon \rangle \langle \epsilon |_{q,\mu}\quad; \qquad 
    E_{q,\mu}^2 = \sum_{\epsilon=-\Lambda}^{\Lambda-1} \epsilon^2 |\epsilon \rangle \langle \epsilon |_{q,\mu},  \label{eqn:E2}
\end{align}
where $|\epsilon \rangle$ and $\epsilon$ correspond to eigenvectors and eigenvalues of a particle in a ring respectively, for each link $q$.  Here we note that the cutoffs on the field are asymmetric ($\Lambda-1$ above and $-\Lambda$ below) because for convenience, we want the dimension of the space to be a power of two which facilitates a simpler binary encoding in our quantum simulations.  The magnetic field term can be defined in terms of the ``plaquette'' operator, which is a product of raising and lowering operators on link sites. The latter is denoted by $U_{q,\mu}$. Specifically,
\begin{equation}\label{eq:raising_link_op}
    U_{q,\mu} = \sum_{\epsilon=-\Lambda}^{\Lambda-1} |\epsilon +1 \rangle \langle \epsilon |_{q,\mu} = \exp\left(i\Delta A_{q,\mu}\right)
\end{equation}
and the plaquette operator is
\begin{equation}
    W_{q,\mu,\nu}^2 = \sum_{\mu \neq \nu}^3 U_{q, \mu} U_{q + 1_{\mu}, \nu} U_{q + 1_{\nu}, \mu}^{\dagger} U_{q, \nu}^{\dagger} + \text{h.c.} 
    \label{eqn:W}
\end{equation}
Here we note that $(q,\mu)$ and $(q,\nu)$ are the links connecting point $\vect{q}$ to its adjacent point in the $\mu^{th}$ and $\nu^{th}$ direction, respectively. We denote these adjacent points by $\vect{q}+1_{\mu}$ and $\vect{q}+1_{\nu}$, respectively. $(q+1_{\mu},\nu)$ is the link connecting point $\vect{q}+1_{\mu}$ to its adjacent point in the $\nu^{th}$ direction, while $(q+1_{\nu},\mu)$ is the link connecting point $\vect{q}+1_{\nu}$ to its adjacent point in the $\mu^{th}$ direction. Thus the operators act on a plaquette i.e. 4 edges of a square face in the 3-D cube.

Next, $H_{\pi}$ is the modified electron kinetic term including interaction with the magnetic vector potential, in a familiar form,
\begin{equation}
H_{\pi} = \sum_{j=1}^{\eta} \sum_{q}^N \frac{1}{2 m_e}\left( \id\otimes\boldsymbol{p}_j\otimes\id - \id\otimes\id\otimes\frac{e}{c}\boldsymbol{A}(q)\right)^2,
\end{equation}
where we use the canonical quantization of the standard particle momentum $\boldsymbol{p} \rightarrow -i \boldsymbol{\nabla}$.
\begin{eqnarray}
   H_{\pi} &=& \frac{1}{2m_e} \sum_{j}^{\eta}  \sum_{q}^N  \left( \id\otimes(-i \boldsymbol{\nabla}_{j})\otimes\id  -  \id\otimes\id\otimes\frac{e}{c} \boldsymbol{A}(q)\right)^2 \nonumber \\
   &=&\frac{1}{2m_e} \sum_{j}^{\eta}  \sum_{q}^N\sum_{\mu=1}^3  \left( \id\otimes(-i \nabla_{j,\mu})\otimes\id  -  \id\otimes\id\otimes\frac{e}{c} A_{q,\mu}\right)^2 \nonumber \\
   &=&\frac{1}{2m_e}\sum_{j}^{\eta}\sum_{q}^N\sum_{\mu=1}^3\left(-\id\otimes\nabla_{j,\mu}^2\otimes\id+\id\otimes (i\frac{2e}{c}\nabla_{j,\mu})\otimes A_{q,\mu}+\id\otimes\id\otimes\frac{e^2}{c^2}A_{q,\mu}^2\right),
\end{eqnarray}
where $\boldsymbol{\nabla}_j$ is the position gradient operator over the 3D grid for particle $j$ and $\boldsymbol{A}(q)$ represents the vector potential operator acting on the links connecting $\q$ to its adjacent point in each of the three Cartesian directions. Thus, here the summation over $\mu$ is implicit, which we have later expanded on.

At each link $(q,\mu)$, $A_{q,\mu}$ can be expanded from the definition of $U_{q,\mu}$, as noted in Equation~\eqref{eq:raising_link_op} and the latter forms the `electric field ladder operators' along with its adjoint form. Using this representation, we can determine the form of $A_{q,\mu}$ as follows.
\begin{align}
    A_{q,\mu} &= \frac{1}{i\Delta} \log\left(U_{q,\mu}\right)\\
    A_{q,\mu} &= \frac{1}{i\Delta} \log \left(\sum_{\epsilon=-\Lambda}^{\Lambda-1} |\epsilon +1 \rangle \langle \epsilon |_{q,\mu}\right) 
    \label{eqn:A}
\end{align}
By construction, the matrix log of the operator above turns out to be diagonal in the Fourier transformed basis where $\mathcal{F}$ is the Fourier transform operator.
\begin{align}
    \mathcal{F} A_{q,\mu} \mathcal{F}^{\dagger} &= \frac{1}{i\Delta}  \log(C)_{q,\mu},
\end{align}
where $C$ is Sylvester's ``clock'' matrix 
\begin{equation}
    C = \begin{pmatrix}
            1 & 0 & 0        &\cdots & 0 \\
            0 & \omega & 0   &\cdots & 0 \\
            0 & 0 & \omega^2 &\cdots & 0 \\
            \vdots & \vdots & \vdots & \ddots & \vdots \\
            0 & 0 & 0 &\cdots & \omega^{d-1} \\
        \end{pmatrix}   \label{eqn:clock}
\end{equation}
where $\omega = e^{2 \pi i / d}$, $d$ is the dimension of the matrix. Therefore,
\begin{equation}
    \log(C) = \begin{pmatrix}
            0 & 0 & 0        &\cdots & 0 \\
            0 & \frac{2\pi i}{d} & 0   &\cdots & 0 \\
            0 & 0 & \frac{2\cdot 2 \pi i}{d} &\cdots & 0 \\
            \vdots & \vdots & \vdots & \ddots & \vdots \\
            0 & 0 & 0 &\cdots & \frac{(d-1)\cdot 2\pi i}{d} 
        \end{pmatrix} \,.
\end{equation}
Thus as expected, the $A_{q,\mu}$ operator on an electric field link is diagonal in the Fourier transformed electric field basis and so 
\begin{equation}
    A_{q,\mu} = \frac{1}{i\Delta} \mathcal{F}^{\dagger} \log(C)_{q,\mu} \mathcal{F} .
    \label{eqn:A_2}
\end{equation}

Lastly, the magnetic spin interaction matrix is defined as the following,
\begin{equation}
H_{s} = -\frac{e}{c}\sum_{j=1}^{\eta} \sum_q^N \boldsymbol{\sigma}_j \cdot \boldsymbol{B}(q) = -\frac{e}{c}\sum_{j=1}^{\eta} \sum_q^N \boldsymbol{\sigma}_j \cdot \left( \boldsymbol{\nabla} \times \boldsymbol{A}(q) \right), 
\end{equation}
where $\times$ denotes vector cross product. This term is derived from the initial particle-field interaction term in Equation~\eqref{eq:pfham}, using the Pauli vector identity, as is described in more detail in Appendix \ref{app:spin_term}. Expanding into a sum over Cartesian directions and separating the sub-spaces, the spin Hamiltonian becomes
\begin{equation}
H_{s} = -\frac{e}{c}\sum_{j}^{\eta} \sum_q^N \sum_{\mu \neq \nu \neq \xi} ^3 \sigma_{j,\mu}\otimes\id\otimes \left( \nabla_{\nu} A_{q,\xi} - \nabla_{\xi} A_{q,\nu} \right) .
\end{equation}

Throughout this work, we will assume atomic units,  $\hbar = e = m_e = 4 \pi \epsilon_0 = 1$, where $\epsilon_0$ is the vacuum permittivity constant, unless otherwise noted. Therefore, the final form of the first quantized Pauli-Fierz Hamiltonian in atomic units is
\begin{eqnarray}
&&\hat{H}_{\text{PF}} \nonumber \\
&=&\left(\frac{1}{ \Delta} \sum_{k < j}^{\eta} \sum_{\q,\vect{r}=1}^N \left(\id\otimes\frac{1}{ \,||\q - \vect{r}||_2} \, | \q \rangle \langle \q |_k \,| \vect{r} \rangle \langle \vect{r} |_j\otimes\id\right)  - \frac{1}{\Delta} \sum_{j=1}^{\eta} \sum_{\kappa=1}^{K} \sum_{\q=1}^N \left(\id\otimes \frac{Z_{\kappa}}{ |\q - \vect{R}_{\kappa}\|_2}  | \q \rangle \langle \q |_j \otimes\id\right)\right) \nonumber\\
&& +\left(\frac{1}{2} \sum_{j=1}^{\eta}  \sum_{q=1}^N\sum_{\mu=1}^3  \left( -\id\otimes (i \nabla_{j,\mu}\otimes\id  -  \id\otimes\id\otimes\frac{1}{c} A_{q,\mu}\right)^2 \right)+ \left(\sum_{q=1}^{N}\sum_{\mu=1}^3\id\otimes\id\otimes \frac{1}{2} E_{q,\mu}^2 - \sum_{q=1}^N\sum_{\mu\neq\nu=1}^3W_{q,\mu,\nu}^2 \right)   \nonumber \\
&&-\left(\frac{1}{c}\sum_{j=1}^{\eta} \sum_{q=1}^N \sum_{\mu \neq \nu \neq \xi=1} ^3 \sigma_{j,\mu}\otimes\id\otimes\left( \nabla_{\nu} A_{q,\xi} - \nabla_{\xi} A_{q,\nu} \right)\right) \nonumber \\ 
 &:=&H_V+H_{\pi}+H_f+H_s.    \label{eqn:HPF}
\end{eqnarray}
For convenience of representation in later parts of this paper we define the following.
\begin{eqnarray}
    H_{\pi} &=& H_{1\pi} +H_{2\pi} +H_{3\pi}    \label{eqn:Hpi} \\
    \text{where } H_{1\pi}&=&\frac{1}{2}\sum_{j=1}^{\eta}\sum_{q=1}^N\sum_{\mu=1}^3-\id\otimes\nabla_{j,\mu}^2\otimes\id    \nonumber \\
    H_{2\pi}&=&\frac{1}{c}\sum_{j=1}^{\eta}\sum_{q=1}^N\sum_{\mu=1}^3\id\otimes(i\nabla_{j,\mu})\otimes A_{q,\mu}   \nonumber \\
    H_{3\pi}&=&\frac{1}{2c^2}\sum_{j=1}^{\eta}\sum_{q=1}^N\sum_{\mu=1}^3\id\otimes\id\otimes A_{q,\mu}^2    \nonumber \\ 
    \text{and }H_f&=&H_{f1} + H_{f2}  \label{eqn:Hf}  \\
    \text{where }H_{f1}&=& \sum_{q=1}^{N}\sum_{\mu=1}^3\id\otimes\id\otimes \frac{1}{2} E_{q,\mu}^2 \nonumber \\
    H_{f2}&=& - \sum_{q=1}^N\sum_{\mu\neq\nu=1}^3\id\otimes\id\otimes W_{q,\mu,\nu}^2    \nonumber \\
    \text{and } H_V&=&H_{V_{ee}}+H_{V_{ne}} \label{eqn:HV}  \\
    \text{where }H_{V_{ee}}&=&\frac{1}{ \Delta} \sum_{k < j}^{\eta} \sum_{\q,\vect{r}=1}^N \left(\id\otimes\frac{1}{ \,||\q - \vect{r}||_2} \, | \q \rangle \langle \q |_k \,| \vect{r} \rangle \langle \vect{r} |_j\otimes\id\right)   \nonumber  \\
    H_{V_{ne}}&=&- \frac{1}{\Delta} \sum_{j=1}^{\eta} \sum_{\kappa=1}^{K} \sum_{\q=1}^N \left(\id\otimes \frac{Z_{\kappa}}{ |\q - \vect{R}_{\kappa}\|_2}  | \q \rangle \langle \q |_j \otimes\id\right) \nonumber
\end{eqnarray}
Our aim in the remainder of the work is to provide methods to block encode each of these pieces so that we can simulate them using qubitization as  well as a divide and conquer scheme.  

\subsection{Recursive Block Encoding}
\label{subsec:divConqBlock}

In simulation algorithms like qubitization \cite{2017_LC, 2019_LC} and quantum singular value transformation (QSVT) \cite{2019_GSLW} we need to block encode a matrix into a unitary in a higher-dimensional Hilbert space.  In this section we briefly describe this approach and discuss how block encodings can be recursed through an approach reminiscent of classical divide and conquer algorithms \cite{2022_CLRS}.
\begin{definition}[\textbf{Block encoding} \cite{2019_GSLW}]
 Suppose $A$ is an $n$-qubit operator, $\mu,\epsilon\in\real_{+}$ and $m\in\nat$. We then say that the $(m+n)$-qubit unitary $U_A$ is a $(\lambda,m,\epsilon)$-block-encoding of $A$ if 
 \begin{eqnarray}
  \|A-\lambda\left(\bra{S}\otimes\id_n\right)U_A\left(\ket{S}\otimes\id_n\right)\|_{\infty}\leq\epsilon,
 \end{eqnarray}
where $\ket{S}$ is an $m$-qubit state, also referred to as the `signal state'.
\end{definition}
 We will often drop the second argument and write `$(\lambda,-,\epsilon)$-block-encoding of $A$', because we focus on the gate complexity and the second argument only captures the extra ancilla needed in the block encoding. Often, even for more brevity, if $\epsilon=0$, then we write 'block-encoding of $\frac{A}{\lambda}$'.

Suppose without loss of generality, we have a Hamiltonian $H_i$ expressed as a linear combination of unitaries (LCU), i.e. $H_i=\sum_{j=1}^{M_i}h_{ij}U_{ij}$, such that $\lambda_i=\sum_j|h_{ij}|$. In this case, we can have a $\left(\lambda_i,\log M_i,0\right)$-block encoding of $H_i$ using an ancilla preparation subroutine and a unitary selection subroutine, which we denote by $\prep_i$ and $\sel_i$ respectively.
\begin{eqnarray}
    \prep_i\ket{0}^{\log M_i}&=&\sum_{j=1}^{M_i}\sqrt{\frac{h_{ij}}{\lambda_i}}\ket{j}   \label{eqn:prepi} \\
 \sel_i&=&\sum_{j=1}^{M_i}\ket{j}\bra{j}\otimes U_{ij}   \label{eqn:seli} 
\end{eqnarray}
It can be shown that~\cite{2012_CW}
\begin{eqnarray}
    \bra{0}\prep_i^{\dagger}\cdot \sel_i\cdot\prep_i\ket{0}&=&\frac{H_i}{\lambda_i}.    \label{eqn:prepiSeli}
\end{eqnarray}
Suppose we have $M$ Hamiltonians - $H_1,\ldots,H_M$, each of which has an LCU decomposition and for each one of them we define the subroutines as in Equations \ref{eqn:prepi} and \ref{eqn:seli}. Now we use these subroutines to define the following,
\begin{eqnarray}
 \prep\ket{0}^{\log M+\sum_i\log M_i}&=&\left(\sum_{i=1}^M\sqrt{\frac{w_i\lambda_i}{\mathcal{\nconst}}}\ket{i}\right)\otimes\bigotimes_{i=1}^M\prep_i    \label{eqn:divPrep} \\
 \sel&=&\sum_{i=1}^M\left(\ket{i}\bra{i}\otimes\bigotimes_{k=1}^{i-1}\id\otimes\sel_i\otimes\bigotimes_{k=i+1}^M\id\right)     \label{eqn:divSel}
\end{eqnarray}
where $w_i>0$ and $\nconst=\sum_{i=1}^Mw_i\lambda_i$. We can use the above two subroutines to block encode a linear combination of Hamiltonians i.e. we can show that
\begin{eqnarray}
    &&(\bra{0}\otimes 1)\prep^{\dagger}\cdot\sel\cdot\prep(\ket{0}\otimes 1)=\frac{1}{\nconst}\sum_{i=1}^Mw_iH_i. \nonumber
\end{eqnarray}
Similar approaches have been used in previous works like \cite{2019_GSLW, 2019_BBMN, 2021_SBWetal} but we provide a general and rigorous statement of this recursive block encoding result in the following theorem, where we provide a formal statement. 

\begin{theorem}
Let $H=\sum_{i=1}^Mw_iH_i$ be the sum of $M$ Hamiltonians and each of them is expressed as sum of unitaries as : $H_i=\sum_{j=1}^{M_i}h_{ij}U_{ij}$ such that $\lambda_i=\sum_j|h_{ij}|$, $w_i>0$. Each of the summand Hamiltonian is block-encoded using the subroutines defined in Equations \ref{eqn:prepi} and \ref{eqn:seli}. Then, we can have an $(\mathcal{A},\lceil \log_2(M) \rceil,0)$-block encoding of $H$, where $\nconst=\sum_{i=1}^Mw_i\lambda_i$, using the ancilla preparation subroutine ($\prep$) defined in Equation \ref{eqn:divPrep} and the unitary selection subroutine ($\sel$) defined in Equation \ref{eqn:divSel}.
\begin{enumerate}
    \item The PREP subroutine has an implementation cost of $\mathcal{C}_{\prep}=\sum_{i=1}^M\mathcal{C}_{\prep_i}+\mathcal{C}_{w}$, where $\mathcal{C}_{\prep_i}$ is the number of gates to implement $\prep_i$ and $\mathcal{C}_w$ is the cost of preparing the state $\sum_{i=1}^M\sqrt{\frac{w_i\lambda_i}{\nconst}}\ket{i}$.

    \item The $\sel$ subroutine can be implemented with a set of multi-controlled-X gates - \\
    $\{M_i\text{ pairs of }C^{\log_2M_i+1}X\text{ gates }:i=1,\ldots,M\}$,$M$ pairs of $C^{\log M}X$ gates and $\sum_{i=1}^MM_i$ single-controlled unitaries - $\{cU_{ij}: j=1,\ldots,M_i; i=1,\ldots,M\}$. 
\end{enumerate}
 \label{thm:blockEncodeDivConq}
\end{theorem}
The proof has been given in Appendix \ref{app:divConqBlock}, where we have argued that we require less number of gates if we divide and block encode, instead of block encoding $H$ as sum of $M'=\sum_{i=1}^MM_i$ unitaries. As an example, let us compare the number of CNOT and T gates required to implement the $\sel$ sub-routine as follows. 
\begin{eqnarray}
 \sel:\ket{i}\ket{0,k_1}_1\ldots\ket{1,j}_i\ldots\ket{0,k_M}_M\ket{\psi}\mapsto\ket{i}\ket{0,k_1}_1\ldots\ket{1,j}_i\ldots\ket{0,k_M}_MU_{ij}\ket{\psi}  \nonumber
\end{eqnarray}
In the above we have represented each set of ancillae qubits in the $M+1$ subspaces of $\prep$ as a separate register. We allot one ancilla qubit, initialized to 0, for each $\prep_i$ register. If state of the first register containing $\log M$ qubits is $\ket{i}$ then the $i^{th}$ register corresponding to $\prep_i$ is selected by flipping this ancilla to $1$. We require $M$ (compute-uncompute) pairs of $C^{\log_2M}X$ gates and $M$ ancillae to make this selection. Now if the state of $\prep_i$-regitser is $\ket{j}$, then we select the $j^{th}$ unitary in the LCU decomposition of $H_i$ i.e. $U_{ij}$. To select unitaries of the $i^{th}$ Hamiltonian $H_i$, we require $M_i$ pairs of $C^{\log_2M_i+1}X$. Decomposing these multi-controlled-NOT gates\cite{2017_HLZetal, 2018_G}, we require $\sum_iM_i(4\log (M_i+1)-4)+M(4\log M-4)$ T, $\sum_iM_i(4\log (M_i+1)-3)+M(4\log M-3)$ CNOT. The use of logical AND gadgets reduces the gate complexity in the uncomputation part.

Now suppose we block encode $H$ as sum of $M'=\sum_{i=1}^MM_i$ unitaries. In the $\sel'$ sub-routine we have $M'$ unitaries, each controlled on $\log_2M'$ qubits. Each of them, in turn can be implemented with a (compute-uncompute) pair of $C^{\log_2M'}X$ and one controlled unitary. Decomposing the multi-controlled-NOT in terms of Clifford+T \cite{2017_HLZetal, 2018_G}, we see that we require at most $M'(4\log_2M'-4)$ T, $M'(4\log_2M'-3)$ CNOT. 

Thus the difference in T-gate count estimate is
\begin{eqnarray}
&&\sum_iM_i(4\log (M_i+1)-4)+M(4\log M-4)-(4\log(\sum_iM_i)-4)(\sum_iM_i)  \nonumber \\
&=&4\sum_iM_i\log\left(\frac{M_i+1}{\sum_jM_j}\right)+4M\log M-4M   \nonumber
\end{eqnarray}
which is less than 0 in most cases. Similarly we can show that the difference in CNOT count estimate is
\begin{eqnarray}
 4\sum_iM_i\log\left(\frac{M_i+1}{\sum_jM_j}\right)+4M\log M-3M \nonumber
\end{eqnarray}
which is again less than 0 in most cases. We use same number of controlled unitaries in both the approaches. Thus, using the divide-and-conquer technique (Theorem \ref{thm:blockEncodeDivConq}) it is possible to reduce the implementation cost in terms of gate count, especially the T-gate and CNOT gate. 

More details can be found in Appendix \ref{app:divConqBlock}, where we
have also explained that we can follow such an approach to block encode product of Hamiltonians using less number of gates.

\begin{remark}[\textbf{Sum of same Hamiltonian, but acting on disjoint subspaces }]
Suppose, in Theorem \ref{thm:blockEncodeDivConq}, all the $H_i$ are the same but they act on disjoint subspaces. In this case, each $\prep_i$ is the same and so it is sufficient to keep only one copy of $\prep_i$ in the $\prep$ subroutine of Equation \ref{eqn:divPrep}. We can absorb $w_i$ in the weights of the unitaries obtained in the LCU decomposition of $H_i$. Thus, in this case we have
 \begin{eqnarray}
  \prep\ket{0}^{\log M+\log M_i}=\left(\sqrt{\frac{1}{M}}\sum_{i=1}^M\ket{i}\right)\otimes \prep_i  \label{eqn:divPrepEq}
 \end{eqnarray}
We require only $\lceil\log M \rceil$ H gates to prepare the superposition in the first register by padding out the number of such subspaces to be a power of $2$.  This step can be avoided, although standard approaches require amplitude amplification~\cite{2021_SBWetal}. With this modification in mind, we also need to make slight modifications in the $\sel$ procedure.
This time, we keep an extra ancilla qubit, initialized to 0, in each subspace. Given a particular state of the first register, we select a subspace by flipping the qubit in the corresponding subspace. The unitaries in each subspace are now additionally controlled on this qubit (of its own subspace). In Appendix \ref{app:divConqBlock} we have discussed the more general situation when each $\prep_i$ are same but the Hamiltonians $H_i$ are different. 
\label{remark:divConqBlock}
\end{remark}

We can further optimize the number of gates by implementing the group of multi-controlled-unitaries in the SELECT subroutines, using the following theorem. Here we partition the control qubits into different groups, store intermediate information in some ancillae and then implement the required logic using these intermediate results. 
\begin{theorem}
Consider the unitary $U = \sum_{j=0}^{M-1} \ketbra{j}{j} \otimes U_j$ for unitary operators $U_j$ that can be implemented controllably. We assume $M$ is a power of 2 for simplicity.
Suppose we have $\log_2M$ qubits and $M$ (compute-uncompute) pairs of $C^{\log_2M}X$ gates for selecting the $M$ basis states. Let $r_1,\ldots,r_n\geq 1$ be positive fractions such that $\sum_{i=1}^n\frac{1}{r_i}=1$ and $\frac{\log_2M}{r_i}$ are integers. Then, $U$ can be implemented with a circuit with $$\sum_{i=1}^nM^{\frac{1}{r_i}}C^{\frac{\log_2M}{r_i}}X + MC^nX$$ 
(compute-uncompute) pairs of gates, $M$ applications of controlled $U_j$ and at most $\sum_{i=1}^nM^{\frac{1}{r_i}}$ ancillae. 
\label{thm:CX}
\end{theorem}
The proof has been provided in Appendix \ref{app:CX}.
Following the construction in \cite{2017_HLZetal, 2018_G}, the number of T-gates required to implement such multiply controlled gates is
\begin{eqnarray}
   \mathcal{T}_n= \sum_{i=1}^nM^{\frac{1}{r_i}}\left(\frac{4\log_2M}{r_i}-4\right)+M(4n-4).
    \label{eqn:CXT}
\end{eqnarray}
and the difference between the T-count estimates with and without splitting is
\begin{eqnarray}
    \mathcal{T}_1-\mathcal{T}_n&=&M(4\log_2M-4)-\sum_{i=1}^nM^{\frac{1}{r_i}}\left(\frac{4\log_2M}{r_i}-4\right)-M(4n-4)    \nonumber \\
    &=&4M(\log_2M-n)-4\sum_{i=1}^nM^{\frac{1}{r_i}}\left(\frac{\log_2M}{r_i}-1\right)
\end{eqnarray}
which can be positive for many values of $n, r_1, r_2,\ldots,r_n$. For example, if each $\frac{1}{r_i}=\frac{1}{n}$, i.e. we divide the control qubits into equal sized groups then
\begin{eqnarray}
    \left(\mathcal{T}_1-\mathcal{T}_n\right)'&=&4M(\log_2M-n)-4nM^{\frac{1}{n}}\left(\frac{\log_2M}{n}-1\right)=4\left(M-M^{\frac{1}{n}}\right)\left(\log_2M-n\right),
    \label{eqn:(T1-Tn)'}
\end{eqnarray}
which is $0$ if $n=1$ and $n=\log_2M$ and is greater than $0$ for every $1<n<\log_2M$. More illustrations have been given in Figure \ref{fig:cx}, where we have shown the variation of this difference (Equation \ref{eqn:(T1-Tn)'}) for different values of $n$ and $M$ ad we find that the maximum difference occurs when we divide into two equal parts. It can be shown that when $0<\frac{1}{r_i}\leq \frac{1}{2}$ then $M^{\frac{1}{r_i}}\frac{\log_2M}{r_i}\leq K'M$, for a large enough constant $K'$. So, we can say that the number of T and CNOT gates is in  $O(M)$, saving logarithmic factors in the asymptotic complexity. This bound also holds for many $\frac{1}{r_i}>\frac{1}{2}$, but breaks down at $r_i=1$. In Figure \ref{fig:cx_t} we compare the number of T-gates for different values of $n$, when the size of each partition is the same (Equation \ref{eqn:CXT}). The linear growth is evident from the curves. Similarly we can show that we can have a reduction in the number of CNOT gates. With the help of logical AND gadgets we do not require to use any T-gate for the uncomputation part.

\begin{figure}
    \centering
    \begin{subfigure}[b]{0.48\textwidth}
        \includegraphics[width=\textwidth]{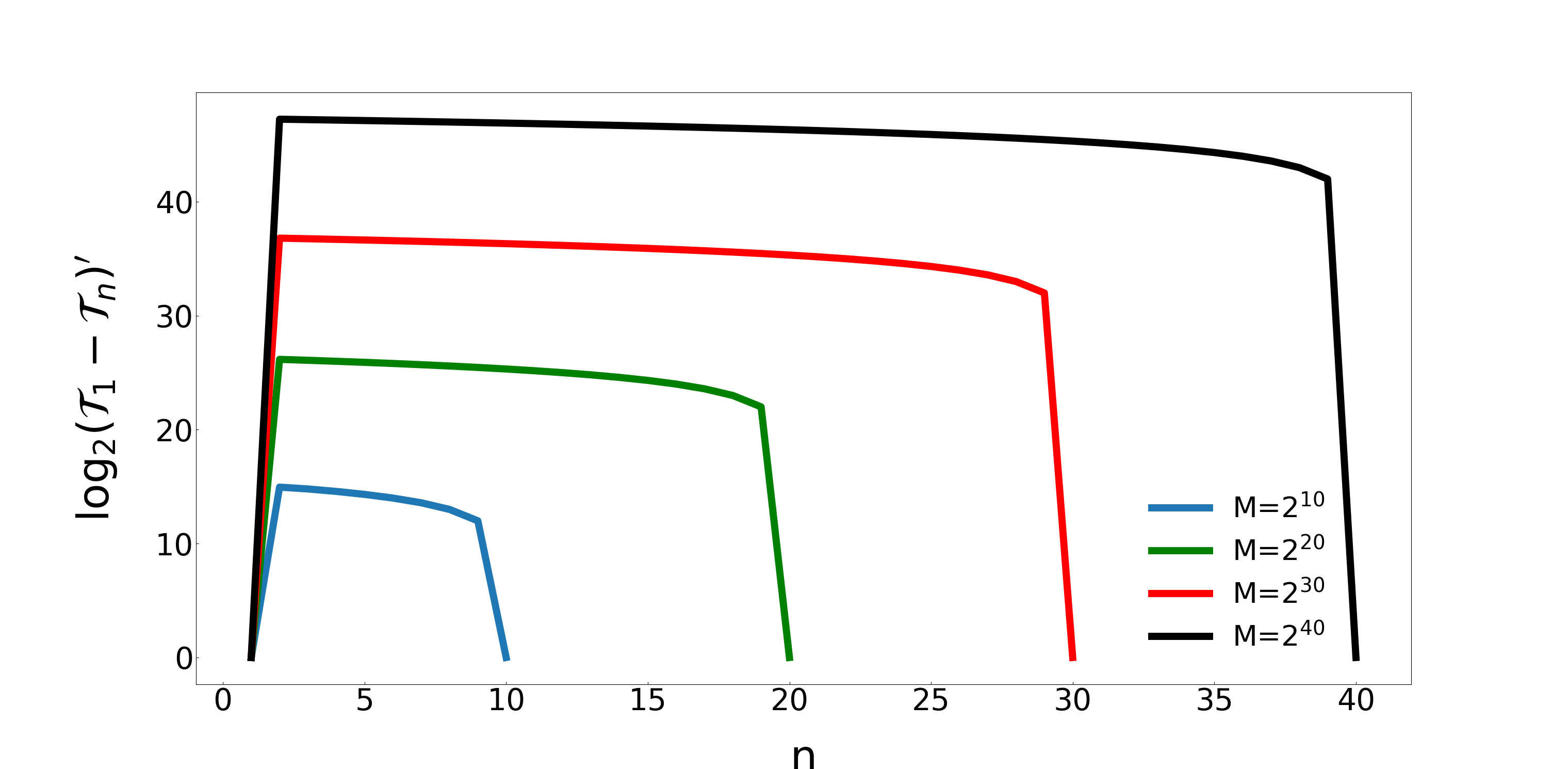}
        \caption{}
        \label{fig:cx}
    \end{subfigure}
    \hfill 
    \begin{subfigure}[b]{0.48\textwidth}
        \includegraphics[width=\textwidth]{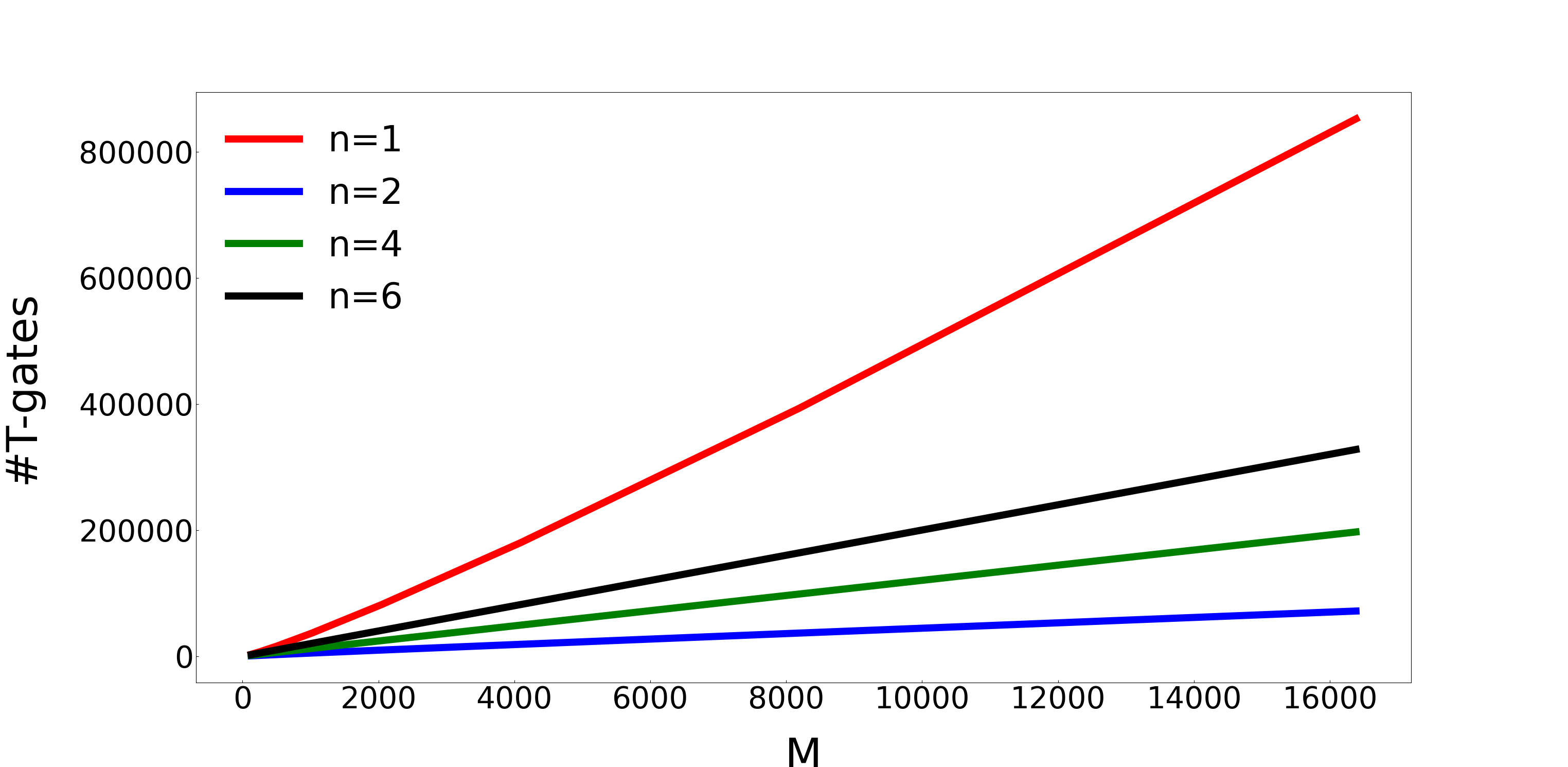}
        \caption{}
        \label{fig:cx_t}
    \end{subfigure}
    \caption{(a) Logarithm of the difference in the T-gate count estimates for a SELECT operation with $M$ controlled unitaries denoted $(\mathcal{T}_1-\mathcal{T}_n)'$ for various number of partitions $n$, for the case when each partition has equal number of control qubits i.e. $\log_2M/n$.  In these cases, we see that the optimal division occurs when $n=2$ which corresponds to the set being split in half. (b) Number of pairs of T-gates for different values of $M$ and number of partitions $n$, when we make equal partitioning.  In these cases we observe that the case where the SELECT circuit's controls are split into two groups outperforms the other splittings considered. }
    \label{fig:splitNmerge}
\end{figure}

\subsection{Algorithm-I : Divide and Conquer - Recursive Trotter splitting}
\label{subsec:sim}

The notion of the divide and conquer approach to simulation is simple.  The core idea behind it is that a Trotter splitting can be used as a means of dividing the simulation into smaller parts, each of which can be directly simulated or further subdivided into smaller parts. The recursive division of the Hamiltonian naturally forms a tree, as depicted in Figure \ref{fig:divNconq}.
The partitions of the Hamiltonian are found according to a heuristic based on different criteria like norm, commutativity, etc; and then simulation of each fragment is performed using different simulation algorithms with sufficient accuracy. We can repeatedly divide each fragment and use the Trotter-Suzuki formula \cite{1991_S, 2021_CSTetal} to bound the error in the exponentials. The resulting number of operations is bounded by the result of the following theorem.
\begin{theorem}
Let $p_1\geq 1$ be a constant. Assuming that $\eta, K\leq N$, $1/c\Delta^2 \in o(1)$ and $K, Z_{sum}\in O(\eta)$ it is possible to simulate $e^{-i\hat{H}_{PF}t}$ with error $\epsilon$, using a Divide and Conquer algorithm, with gate complexity in 
\begin{eqnarray}
    \widetilde{O}\left(N^2t\log^2\Lambda\left(\frac{\eta}{\Delta^2}+\Lambda^2\right)\left(\frac{t\eta}{\epsilon\Delta^2\Lambda}\right)^{\frac{1}{p_1}}\right).  \nonumber
\end{eqnarray}
\label{thm:DC}
\end{theorem}
We take two factors into account for grouping the Hamiltonian terms. First, we consider the pair-wise commutators. This is because the error introduced due to splitting is determined by the expansion of the Trotter-Suzuki formula, given in \cite{2021_CSTetal}, depends on the norm of the nested commutators \cite{2021_CSTetal}. From Lemma \ref{lem:alpha_comm} stated later, we find that the pair-wise commutators play a significant role in bounding the nested commutators, especially for lower order formulae. We must keep in mind that as we increase the order, the number of exponentials and hence complexity increases. Our algorithm mitigates these errors by grouping together terms with larger commutator bound, so that this does not reflect on the overall error. The second factor that we consider is the $\ell_1$ norm of the fragment Hamiltonians.  Specifically, we consider the $\ell_1$ norm of the coefficients in an LCU decomposition of the Hamiltonian and this also serves as an upper bound on the spectral norm of the Hamiltonian.
In simulation algorithms like \cite{2015_BCCKS, 2017_LC, 2019_LC, 2019_GSLW} the block encoding of the Hamiltonian is repeated a
number of times proportional to its $\ell_1$ norm. So, if we block encode terms with small norm together with terms with larger norm then we end up repeating the smaller norm terms more frequently than is necessary. Instead, we group these terms separately and adjust the error accordingly. We summarize the different $\ell_1$ norm and pairwise commutators in Tables \ref{tab:norm} and \ref{tab:comm} and in Appendix \ref{app:lcu}, \ref{app:comm} we have given a more detail description of our calculations. 

Let $U=e^{-i\hat{H}_{PF}t}$, where $t$ is the total time of evolution and $\widetilde{U}$ is the final unitary we implement. In the first level of split we divide $\hat{H}_{PF}$ into two parts, i.e. $\hat{H}_{PF}=H_{11}+H_{12}$, where $H_{11}=H_{f1}+H_s+H_V+H_{\pi}$ and $H_{12}=H_{f2}$. This is because the innermost commutator between $H_s, H_V, H_{\pi}$ are significantly higher (Table \ref{tab:comm}) and we have tried to avoid terms with $\Lambda$. By grouping them together we have tried to keep the error small and independent of $\Lambda$. We divide $t$ into $r_1$ intervals, each of length $\tau_1=t/{r_1}$, such that we can approximate the Trotter-Suzuki formula of order  $p_1$ well within each time segment. If $\widehat{U}$ is the unitary we obtain by approximating $U$, then invoking Box 4.1 from~\cite{2010_NC}
\begin{eqnarray}
 \|U-\widehat{U}\|&=&\|\left(e^{-i\hat{H}_{PF}t/{r_1}}\right)^{r_1}-\left(\mathscr{S}_{p_1}(H_{11},H_{12};t/{r_1})\right)^{r_1}\|   \nonumber \\
 &\leq& r_1\|e^{-i\hat{H}_{PF}\tau_1}-\mathscr{S}_{p_1}(H_{11},H_{12},\tau_1)\|:=r_1\epsilon_1.
 \label{eqn:U-Uhat}
\end{eqnarray}
Suppose, after approximating $e^{-i\hat{H}_{PF}\tau_1}$ using the Trotter-Suzuki formula, we obtain at most $\term_1$ terms of the form $e^{-iH_{11}\tau_1}=U_{1\tau_1}^{(1)}$ and $e^{-iH_{12}\tau_1}=U_{2\tau_1}^{(1)}$. We denote the unitary implementation of $U_{1\tau}$ and $U_{2\tau}$ by $\widetilde{U}_{1\tau_1}^{(1)}$ and $\widetilde{U}_{2\tau_1}^{(1)}$ respectively. Thus
\begin{eqnarray}
 \|\widehat{U}-\widetilde{U}\|\leq r_1\term_1\|U_{1\tau_1}^{(1)}-\widetilde{U}_{1\tau_1}^{(1)}\|+r_1\term_1\|U_{2\tau_1}^{(1)}-\widetilde{U}_{2\tau_1}^{(1)}\|:=r_1\term_1\|U_{1\tau_1}^{(1)}-\widetilde{U}_{1\tau_1}^{(1)}\|+r_1\term_1\delta_1.  \label{eqn:Uhat-Utilde}
\end{eqnarray}
Thus, after the first level of splitting in the figure, using Equations \ref{eqn:U-Uhat} and \ref{eqn:Uhat-Utilde} we have,
\begin{eqnarray}
 \|U-\widetilde{U}\|\leq \|U-\widehat{U}\|+\|\widehat{U}-\widetilde{U}\|\leq r_1\epsilon_1+r_1\term_1\delta_1+r_1\term_1\|U_{1\tau_1}^{(1)}-\widetilde{U}_{1\tau_1}^{(1)}\| \label{eqn:U-Utilde1}
\end{eqnarray}
In the second level of splitting depicted in the figure, we divide $H_{11}$ into two groups $H_{21}=H_{f1}$ and $H_{22}=H_s+H_V+H_{\pi}$. We further divide $\tau_1$ into $r_2$ intervals, each of length $\tau_2=\tau_1/r_2$. Let $\widehat{U}_{1\tau_1}^{(1)}$ be the unitary we obtain by approximating $U_{1\tau_1}^{(1)}$ with a Trotter-Suzuki formula of order $p_2$. Then,
\begin{eqnarray}
 \|U_{1\tau_1}^{(1)}-\widehat{U}_{1\tau_1}^{(1)}\|&=&\|\left(e^{-iH_{11}\tau_1/r_2}\right)^{r_2}-\left(\mathscr{S}_{p_2}(H_{21},H_{22};\tau_1/r_2)\right)^{r_2}\|    \nonumber \\
 &\leq&r_2\|e^{-iH_{11}\tau_2}-\mathscr{S}_{p_2}(H_{21},H_{22};\tau_2)\|:=r_2\epsilon_2,   \label{eqn:U1-U1hat}
\end{eqnarray}
After approximating $e^{-iH_{11}\tau_2}$, suppose we obtain at most $\term_1$ number of $e^{-iH_{21}\tau_2}=U_{1\tau_2}^{(2)}$ and $e^{-iH_{22}\tau_2}=U_{2\tau_2}^{(2)}$. The unitary implementations of $e^{-iH_{21}\tau_2}$ and $e^{-iH_{22}\tau_2}$ are denoted by $\widetilde{U}_{1\tau_2}^{(2)}$ and $\widetilde{U}_{2\tau_2}^{(2)}$ respectively. Thus,
\begin{eqnarray}
 \|\widehat{U}_{1\tau_1}^{(1)}-\widetilde{U}_{1\tau_1}^{(1)}\|\leq r_2\term_2\|U_{1\tau_2}^{(2)}-\widetilde{U}_{1\tau_2}^{(2)}\|+r_2\term_2\|U_{2\tau_2}^{(2)}-\widetilde{U}_{2\tau_2}^{(2)}\|:=r_2\term_2\delta_2+r_2\term_2\|U_{2\tau_2}^{(2)}-\widetilde{U}_{2\tau_2}^{(2)}\|    \label{eqn:U1hat-U1tilde}
\end{eqnarray}
and so plugging in Equations \ref{eqn:U1-U1hat} and \ref{eqn:U1hat-U1tilde} in Equation \ref{eqn:U-Utilde1} we get the following after the second level of splitting.
\begin{eqnarray}
 \|U-\widetilde{U}\|\leq r_1\epsilon_1+r_1\term_1\delta_1+r_1\term_1\left(r_2\epsilon_2+r_2\term_2\delta_2+r_2\term_2\|U_{2\tau_2}^{(2)}-\widetilde{U}_{2\tau_2}^{(2)}\|\right) \label{eqn:U-Utilde2}
\end{eqnarray}
In the third level of splitting, we divide $H_{22}$ into two groups $H_{31}=H_s+H_{3\pi}$ and $H_{32}=H_V+H_{1\pi}+H_{2\pi}$ and . We also divide $\tau_2$ into $r_3$ itnervals, each of length $\tau_3=\tau_2/r_3$. Let $\widehat{U}_{2\tau}^{(2)}$ be the unitary we obtain by approximating $U_{2\tau_2}^{(2)}$ with a Trotter-Suzuki formula of order $p_3$. Then,
\begin{eqnarray}
 \|U_{2\tau_2}^{(2)}-\widehat{U}_{2\tau_2}^{(2)}\|&=&\|\left(e^{-iH_{22}\tau_2/r_3}\right)^{r_3}-\left(\mathscr{S}_{p_3}(H_{31},H_{32};\tau_2/r_3)\right)^{r_3}\|   \nonumber \\
 &\leq& r_3\|e^{-iH_{22}\tau_3}-\mathscr{S}_{p_3}(H_{31},H_{32};\tau_3)\|:=r_3\epsilon_3    \label{eqn:U2-U2hat}
\end{eqnarray}
After approximation, suppose we obtain at most $\term_3$ number of $e^{-iH_{31}\tau_3}:=U_{1\tau_3}^{(3)}$ and $e^{-iH_{32}\tau_3}:=U_{2\tau_3}^{(3)}$. The unitary implementations of $e^{-iH_{31}\tau_3}$ and $e^{-iH_{32}\tau_3}$ are denoted by $\widetilde{U}_{1\tau_3}^{(3)}$ and $\widetilde{U}_{2\tau_3}^{(3)}$, respectively. So,
\begin{eqnarray}
 \|\widehat{U}_{2\tau_2}^{(2)}-\widetilde{U}_{2\tau_2}^{(2)}\|\leq r_3\term_3\|U_{1\tau_3}^{(3)}-\widetilde{U}_{1\tau_3}^{(3)}\|+r_3\term_3\|U_{2\tau_3}^{(3)}-\widetilde{U}_{2\tau_3}^{(3)}\|:=r_3\term_3\delta_{31}+r_3\term_3\delta_{32} \label{eqn:U2hat-U2tilde}
\end{eqnarray}
and hence plugging Equations \ref{eqn:U2-U2hat} and \ref{eqn:U2hat-U2tilde} in Equation \ref{eqn:U-Utilde2} we get the following bound on the simulation error after the third and final level of splitting.
\begin{eqnarray}
 \|U-\widetilde{U}\|&\leq&r_1\epsilon_1+r_1\term_1\delta_1+r_1r_2\term_1\epsilon_2+r_1r_2\term_1\term_2\delta_2+r_1r_2\term_1\term_2\left(r_3\epsilon_3+r_3\term_3\delta_{31}+r_3\term_3\delta_{32}\right)  \nonumber \\
 &=&r_1\epsilon_1+r_1\term_1\delta_1+r_1r_2\term_1\epsilon_2+r_1r_2\term_1\term_2\delta_2+r_1r_2r_3\term_1\term_2\epsilon_3 \nonumber \\
 &&+r_1r_2r_3\term_1\term_2\term_3(\delta_{31}+\delta_{32}) \label{eqn:U-Utilde3}
\end{eqnarray}
Let the number of gates required to implement the unitaries $U_{2\tau_1}^{(1)}, U_{1\tau_2}^{(2)}, U_{1\tau_3}^{(3)}$ and $U_{2\tau_3}^{(3)}$ be $\g_1, \g_2, \g_{31}$ and $\g_{32}$ respectively. Thus the total number of gates for implementing $U=e^{-i\hat{H}_{PF}t}$ is
\begin{eqnarray}
    \g\leq r_1\term_1\g_1+r_1r_2\term_1\term_2\g_2+r_1r_2r_3\term_1\term_2\term_3\left(\g_{31}+\g_{32}\right)    \label{eqn:gate_net}
\end{eqnarray}
We summarize the above results in the following lemma. For simplicity, we assume that the time lengths are always exactly divisible by the number of segments. This Lemma can be generalized for arbitrary divisions, for which we can draw a tree similar to Figure \ref{fig:divNconq} that may be useful in deriving bounds on the error and gate complexity.
\begin{lemma}
 
 Let $\term_1, \term_2$ and $\term_3$ be the number of operator exponentials that appear in the divide and conquer simulation method given in Figure~\ref{fig:divNconq} where the $\tau_i$ refer to the timestep, $\epsilon_i$ the error tolerance at the level of the division, $\delta_i$ is the synthesis error tolerable and $\mathcal{G}_i$ represents the gate count required for implementing the resulting exponentials. If $U=e^{-i\hat{H}_{PF}t}$ and $\widetilde{U}$ is the final unitary implementation of $U$, then the total simulation error is
 \begin{eqnarray}
  \|U-\widetilde{U}\|
 &=&r_1\left(\epsilon_1+\term_1\left(\delta_1+r_2\left(\epsilon_2+\term_2\left(\delta_2+r_3\epsilon_3 
 +r_3\term_3(\delta_{31}+\delta_{32})\right)\right)\right)\right) \nonumber
 \end{eqnarray}
and the total number of gates required is
\begin{eqnarray}
  \g\leq r_1\term_1\left(\g_1+r_2\term_2\left(\g_2+r_3\term_3\left(\g_{31}+\g_{32}\right)\right)\right).    \nonumber
\end{eqnarray}
\label{lem:divNconq}
\end{lemma}

\begin{figure}
 \centering
 \small
 \begin{tikzpicture}
  \draw (0,0.5) rectangle (-1.2,-0.8);
  \node at (-0.6,0.2) {$\hat{H}_{PF};t$};
  \draw (0,-0.2)--(-1.2,-0.2);
  \node at (-0.6,-0.5) {$\mathscr{S}_{p_1}$; $\epsilon_1$};
  
  \draw (-0.6,-0.8)--(-4,-2);
  \draw (-0.6,-0.8)--(1.6,-2);
  \node at (-3,-1.2) {$r_1\term_1$};
  \node at (1,-1.2) {$r_1\term_1$};
  
  \draw (-6.5,-2) rectangle (-1.5,-3.8);
  \node at (-4,-2.3) {$H_{11}=H_{f1}+H_s+H_V+H_{\pi}$};
  \node at (-4,-2.8) {$\tau_1=\frac{t}{r_1}$};
  \draw (-6.5,-3.2)--(-1.5,-3.2);
  \node at (-4,-3.5) {$\mathscr{S}_{p_2}$; $\epsilon_2$};
  
  \draw[ultra thick] (0.2,-2) rectangle (3,-3.2);
  \node at (1.5,-2.3) {$H_{12}=H_{f2}$};
  \node at (1.5,-2.8) {$\tau_1=\frac{t}{r_1}$; $\g_1$; $\delta_1$};
  
  \draw (-4,-3.8)--(-6,-5);
  \draw (-4,-3.8)--(-1,-5);
  \node at (-5.8,-4.4) {$r_2\term_2$};
  \node at (-1.5,-4.4) {$r_2\term_2$};
  
  \draw[ultra thick] (-7.3,-5) rectangle (-4.7,-6.2);
  \node at (-6,-5.3) {$H_{21}=H_{f1}$};
  \node at (-6,-5.8) {$\tau_2=\frac{\tau_1}{r_2}$; $\g_2$; $\delta_2$};
  
  \draw (-3,-5) rectangle (1,-6.8);
  \node at (-1,-5.3) {$H_{22}=H_s+H_{\pi}+H_V$};
  \node at (-1,-5.8) {$\tau_2=\frac{\tau_1}{r_2}$};
  \draw (-3,-6.2)--(1,-6.2);
  \node at (-1,-6.5) {$\mathscr{S}_{p_3}$; $\epsilon_3$};

  \draw (-1,-6.8)--(-5,-8);
  \draw (-1,-6.8)--(1,-8);
  \node at (-3.6,-7.2) {$r_3\term_3$};
  \node at (0.5,-7.2) {$r_3\term_3$};
  
  \draw[ultra thick] (-7,-8) rectangle (-3,-9.2);
  \node at (-5,-8.3) {$H_{31}=H_s+H_{3\pi}$};
  \node at (-5,-8.8) {$\tau_3=\frac{\tau_2}{r_3}$; $\g_{31}$; $\delta_{31}$};
  
  \draw[ultra thick] (-1,-8) rectangle (3,-9.2);
  \node at (1,-8.3) {$H_{32}=H_V+H_{1\pi}+H_{2\pi}$};
  \node at (1,-8.8) {$\tau_3=\frac{\tau_2}{r_3}$; $\g_{32}$; $\delta_{32}$};
  
 \end{tikzpicture}
\caption{A tree depicting the partition of the Hamiltonian at different stages. In each rectangular node we mention the Hamiltonian and the time interval. If the Hamiltonian is partitioned i.e. it is a parent node, then we divide the rectangle and in the lower half we mention the order of the Trotter-Suzuki formula and error incurred due to splitting. The leaf nodes (thick rectangles) store the Hamiltonians that are simulated. They store information about the time, gate and simulation error. The edges are labeled by the number of segments of the time interval of parent and the number of exponentials (copies of each child node) obtained after applying Trotter-Suzuki formula.}
\label{fig:divNconq}
\end{figure}
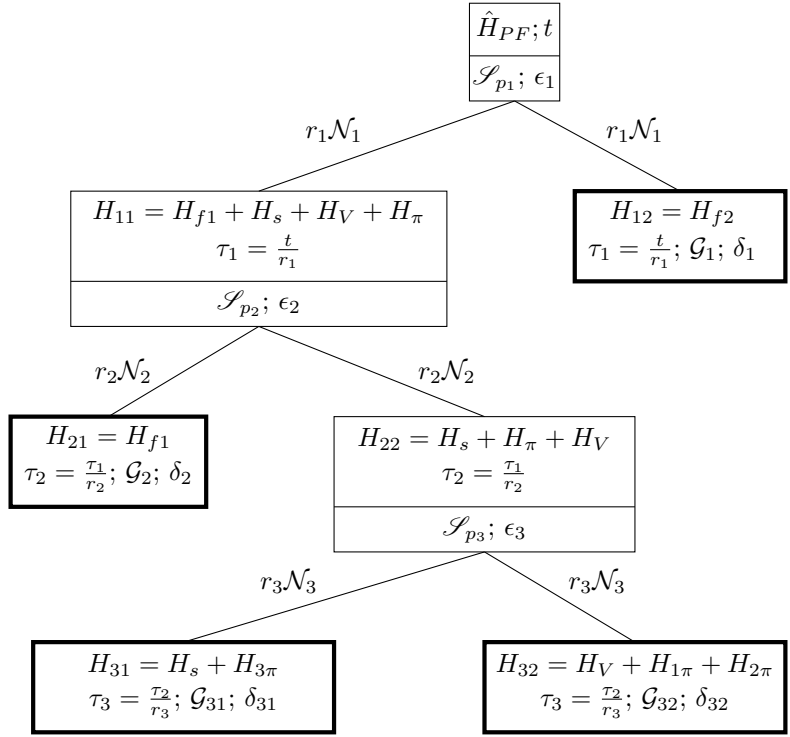

\begin{table}[t]
\centering
\small
 \begin{tabular}{|p{2.1cm}|c|c|c|}
 \hline
  \textbf{Hamiltonian or Operator} & $\#$\textbf{Unitaries} & $\ell_1$ \textbf{norm} & \textbf{Types of unitaries}      \\
  \hline
  $A$ & $\lceil\log_2d\rceil+1$ & $\frac{2\pi}{\Delta}$ & $\Z$ (Cor. \ref{cor:A_lcu}) \\ 
  \hline 
  $A^2$ & $2\lceil\log_2d\rceil$ (or $\frac{\log_2^2d+\log_2d}{2}$) & $\frac{4\pi^2}{\Delta^2}$  & $\Z$ (Cor. \ref{cor:A2_lcu}) \\
  \hline 
  $\nabla$ & $2a$ & $\frac{\ln (2a^2)}{h}$ & Adder (Lem. \ref{lem:lcuNabla})   \\
  \hline 
  $\nabla^2$ & $2a+1$ & $\frac{4\pi^2}{3h^2}$ & Adder (Lem. \ref{lem:lcuNabla2})   \\
  \hline 
  $E^2$ & $2\log_2\Lambda$ (or $\frac{\log_2^2\Lambda+\log_2\Lambda+2}{2}$) & $\Lambda^2$ & $\Z$ (Eq. \ref{eqn:E2_lcu_0})    \\
  \hline 
  $U$ & $1$ & $1$ & Rotation, QFT (Cor. \ref{cor:U_lcu})  \\
  \hline
  $H_{1\pi}$ & $6a\eta N$ & $\frac{8\pi^2\eta N}{h^2}$ & . \\
  \hline
  $H_{2\pi}$ & $6a\eta N\log_2d$ & $\frac{12\pi\eta N\ln 2a^2}{ch\Delta}$ & .   \\
  \hline 
  $H_{3\pi}$ & $6\eta N\log_2d$ & $\frac{12\pi^2\eta N}{c^2\Delta^2}$ & .   \\
  \hline
  $H_{V_{ee}}$ & $\frac{\eta(\eta-1)N}{2}$ & $\frac{\eta(\eta-1)}{2\Delta^2}$ & .\\
  \hline
  $H_{V_{ne}}$ & $\eta LN$ & $\frac{\eta}{\Delta^2}\left(\sum_{\kappa=1}^K|Z_{\kappa}|\right):=\frac{\eta Z_{sum}}{\Delta^2}$ & .\\
  \hline
  $H_{f1}$ & $6N\log_2\Lambda$ & $\frac{3N\Lambda^2}{2}$ & .\\
  \hline 
  $H_{f2}$ & $6N$ & $6N$ & . \\
  \hline 
  $H_s$ & $12 \eta Na\log_2d$ & $\frac{12\pi \eta N\ln 2a^2}{ch\Delta}$ & . \\
  \hline
 \end{tabular}
\caption{Summary of the number of unitaries in the decomposition of different Hamiltonians and operators and an upper bound on the $\ell_1$ norm of the coefficients in an LCU decomposition of the Hamiltonian. The latter is also an upper bound on the spectral norm of the Hamiltonian. In the last column we have mentioned the unitaries occurring in the decompositions of the operators. For $A^2$ and $E^2$ we have given the unitaries for the decomposition provided in this paper, which yields the number of unitaries shown in the bracket.  Here $2a+1$ is the number of points used in the finite difference approximation and $h$ is the spatial scaling of the grid used in the finite difference formula.}
\label{tab:norm}
\end{table}

\begin{table}
\centering
\small
 \begin{tabular}{|c|c|c|c|c|c|c|}
  \hline 
   & $H_{\pi}$ & $H_{V_{ee}}$ & $H_{V_{ne}}$ & $H_{f1}$ & $H_{f2}$ & $H_s$ \\
  \hline 
  $H_{\pi}$ & 0 & - & - & . & . & .   \\
  \hline
  $H_{V_{ee}}$ & $\frac{4\pi\eta(\eta-1)N^{8/3}}{h^2\Delta^2}\left(\pi+\frac{6h\ln (2a^2)}{c\Delta}\right)$ & 0 & - & . & . & .   \\
  \hline
  $H_{V_{ne}}$ & $\frac{4\pi\eta N^{5/3}KZ_{max}}{h^2\Delta^2}\left(\pi+\frac{6h\ln (2a^2)}{h\Delta}\right)$ & - & 0 & . & . & .   \\
  \hline
  $H_{f1}$ & $\frac{6\pi\eta N\Lambda^2}{c\Delta}\left(\frac{\ln(2a^2)}{h}+\frac{2\pi}{c\Delta}\right)$  & 0 & 0 & 0 & . & .   \\
  \hline
  $H_{f2}$ & $\frac{198\pi\eta N}{c\Delta}\left(\frac{\ln (2a^2)}{h}+\frac{\pi}{c\Delta}\right)$ & 0 & 0 & $12N\Lambda$ & 0 & .   \\
  \hline
  $H_s$ & $\frac{96\pi^2\eta^2 N\ln (2a^2)}{hc^2\Delta^2}\left(\frac{\ln (2a^2)}{h}+\frac{\pi}{c\Delta}\right)$ & 0 & 0 & $\frac{24\pi\eta N\Lambda^2\ln (2a^2)}{ch\Delta}$ & $\frac{288\pi\eta N\ln (2a^2)}{ch\Delta}$ & 0   \\
  \hline
 \end{tabular}
\caption{Summary of bounds on the norm of the pair-wise commutators between types of terms that appear in the Pauli-Fierz Hamiltonian.  Components that can be computed using the anti-symmetry of the commutator have been dropped from the table for clarity.}
\label{tab:comm}
\end{table}

We describe the algorithms to simulate the exponentials of the four fragments  - $H_{12}$, $H_{21}$, $H_{31}$ and $H_{32}$ in Appendix \ref{sec:method}. Above, we state the complexity of the circuits in terms of Clifford+T and (controlled)-rotation gates. Clifford+T is one of the most popular fault-tolerant universal gate set. But not all unitaries can be exactly implementable by it. So we have used the rotation gates, which are the only approximately implementable gates we use. The T-count of (controlled)-rotation gate is proportional to the logarithm of the synthesis error \cite{2015_KMM, 2015_BRS, 2016_RS, 2022_GMM} and thus low T-counts are often observed for reasonable error budgets. In all the cases we have separately reported the complexity of (controlled)-rotation gates, without further decomposing it with the Clifford+T gate set. This is because in this paper we do not account for the gate synthesis error. 

The decomposition of the operators is described explicitly in Appendix \ref{app:lcu}. In Table \ref{tab:norm} we summarize these decompositions. We give the number of unitaries, $\ell_1$ norm and for the operators in the last column we have mentioned the types of unitaries in these decompositions. For our analysis we also require bounds on the pair-wise commutators of the different Hamiltonian partitions. We have summarized these bounds in Table \ref{tab:comm} and given a detail derivation in Appendix \ref{app:comm}.

Here, in the following Lemmas, we summarize the bounds on the total number of gates required for simulating each of the fragment Hamiltonian. These information will be useful for deriving the complexity of simulating $e^{-i\hat{H}_{PF}t}$, using both the algorithms considered by us. The proofs can be found in Appendix \ref{subsec:H12}-\ref{subsec:H32}. Here, we have given some brief explanations of mainly the PREP and SELECT sub-routines while block-encoding.  
\begin{lemma}
Let $H_{12}=-\sum_{q=1}^N\sum_{\mu\neq\nu=1}^3\id\otimes\id\otimes W_{q,\mu,\nu}^2$, where $W_{q,\mu,\nu}^2$ is the plaquette operator described in Equation \ref{eqn:W}. Then we require
\begin{eqnarray}
    \g_1'\in O\left(N\log d\right) \nonumber
\end{eqnarray}
gates to have a $(6N,-,0)$-block-encoding of $H_{12}$ and hence the number of gates required for simulating $e^{iH_{12}\tau_1}$ with error $\delta_{12}>0$, using qubitization is,
\begin{eqnarray}
    \g_{1}\in O\left(N^2\tau_1\log d+\frac{\log(1/\delta_{12})}{\log\log(1/\delta_{12})}N\log d\right).  \nonumber
\end{eqnarray}
    \label{lem:g1}
\end{lemma}
We know that $H_{12}=H_{f2}$ (Equation \ref{eqn:Hf}) which corresponds to the plaquette terms in the dynamics of the system. In Corollary \ref{cor:U_lcu} we show that the raising operator $U_{q,\mu}$ (defined in Equation \ref{eq:raising_link_op}).
\begin{eqnarray}
 U_{q,\mu}=\mathcal{F}_{q,\mu}\left(\bigotimes_{k=1}^{\log_2d-1}R_z(\theta_k)\right)_{q,\mu}\mathcal{F}_{q,\mu}^{\dagger}\qquad\text{where }\theta_k=\frac{2\pi}{d}2^k\text{ and }\mathcal{F}\text{ is Fourier Transform}.  \nonumber
\end{eqnarray}
This shows that an individual $U_{q,\mu}$ can be implemented using $\log_2(d)$ single qubit rotations.  The plaquette operator $W_{q,\mu,\nu}$ can be implemented using $4$ such terms and thus $W_{q,\mu,\nu}^2$ (Equation \ref{eqn:W}) can be implemented by a layer of at most $4\log_2d$ parallel rotations, conjugated by Fourier transformation. The ancilla preparation sub-routine does the following. 
\begin{eqnarray}
 \prep_{f2}\ket{0}^*=\left(\frac{1}{\sqrt{N}}\sum_{q=1}^N\ket{q}\right)\otimes\left(\frac{1}{\sqrt{6}}\sum_{\mu\neq\nu=1}^3\sum_{k=0}^1\ket{\mu}\ket{\nu}\ket{k}\right) \nonumber
\end{eqnarray}
First we have the $\log_2N$-qubit electric link index register that stores the $N$ electric link indices in equal superposition. Next we have the 4+1-qubit spin index register. The first 4 qubits stores the value of $\mu, \nu$. The last 1 qubit indicates if we apply h.c. If $\mu=\nu$ or $\mu,\nu>3$, then we discard. Throughout this paper, by 'discarding' we mean unfollowing a computation path. This is indicated by an ancilla qubit, which when set to $\ket{1}$, we only apply $\id$. 

The unitary selection operator can be expressed as
\begin{eqnarray}
    \sel_{f2}&=&\sum_{q=1}^N\sum_{\mu\neq\nu=1}^3\sum_{k=0}^1\ket{q,\mu,\nu,0}\bra{q,\mu,\nu,0}\otimes U_{q,\mu}U_{q+1_{\mu},\nu}U_{q+1_{\nu},\mu}U_{q,\nu} \nonumber \\
    &&+\sum_{q=1}^N\sum_{\mu\neq\nu=1}^3\sum_{k=0}^1\ket{q,\mu,\nu,1}\bra{q,\mu,\nu,1}\otimes U_{q,\nu}^{\dagger}U_{q+1_{\nu},\mu}^{\dagger}U_{q+1_{\mu},\nu}^{\dagger}U_{q,\mu}^{\dagger},  \nonumber
\end{eqnarray}
by which we can conveniently prove that
\begin{eqnarray}
    \bra{0}\prep_{f2}^{\dagger}\cdot\sel_{f2}\cdot\prep_{f2}\ket{0}=\frac{H_{f2}}{6N},  \nonumber
\end{eqnarray}
providing a $(6N,.,0)$-block encoding if $H_{f2}$. More detail analysis on the gate complexity can be found in Appendix \ref{subsec:H12}.

\begin{lemma}
Let $H_{21}=\sum_{q=1}^N\sum_{\mu=1}^3\id\otimes\id\otimes \frac{1}{2}E_{q,\mu}^2$, where $E_{q,\mu}^2$ is the operator described in Equation \ref{eqn:E2}. Then with Trotterization we can implement $e^{-iH_{21}\tau_2}$ exactly using the following number of gates.
\begin{eqnarray}
    \g_2\in O\left(N\log^2\Lambda\right)  \nonumber
\end{eqnarray}

Alternatively, we can have a $(\frac{3N\Lambda^2}{2},-,0)$-block-encoding of $H_{21}$ with the following number of gates.
\begin{eqnarray}
    \g_2'\in O\left(N\log^2\Lambda\right) \nonumber
\end{eqnarray}
    \label{lem:g2}
\end{lemma}
We know that $H_{21}=H_{f1}$ (Equation \ref{eqn:Hf}) and since $[E_{\ell,\mu}^2,E_{q,\nu}^2]=0$ if $\ell\neq q$, so $
 e^{-iH_{f1}\tau_2}=\prod_{q=1}^N\prod_{\mu=1}^3\id\otimes\id\otimes e^{-i\frac{1}{2}E_{q,\mu}^2\tau_2}$. If $\zeta=1+\log_2\Lambda$, then $E^2$ can be written as sum of Z-operators, as shown below, and we simulate $e^{-iH_{f1}\tau_2}$ by Trotterization, as done in \cite{2020_SLSW}.
\begin{eqnarray}
 E^2=\frac{1}{6}\left(2^{2\zeta-1}+1\right)\id+\sum_{j=0}^{\zeta-1}2^{j-1}\Z_j+\sum_{j=0}^{\zeta-2}\sum_{k>j}^{\zeta-1}2^{j+k-1}\Z_j\Z_k. \label{eqn:main:E2_lcu}
\end{eqnarray}
\begin{lemma}[Lemma 2 in \cite{2020_SLSW}]
 There exists a circuit that implements $e^{-iE^2\tau_2}$ on $\zeta$ qubits exactly, up to an (efficiently computable) global phase, using $\frac{(\zeta+2)(\zeta-1)}{2}$ CNOT operations and $\frac{\zeta(\zeta+1)}{2}$ single-qubit rotations. Here $\zeta=1+\log_2\Lambda$.
 \label{lem:E2020_SLSW}
\end{lemma}
Since $E^2$ is expressed as sum of Pauli operators so we can use the algorithm in \cite{2023_MWZ} to optimize the rotation gates, possibly at the cost of increasing a few Toffoli gates. More detail on the gate complexity can be found in Appendix \ref{subsec:H21}.

Alternatively, we can use qubitization to simulate $e^{-iH_{21}\tau_2}$. The Algorithm-II described in Section \ref{subsec:totalQubit} applies qubitization on the entire Hamiltonian $\hat{H}_{PF}$ and for this we require to block encode $H_{21}$, which we briefly describe here.  The ancilla preparation sub-routine is defined as follows.
\begin{eqnarray}
    \prep_{f1}\ket{0}^*=\left(\frac{1}{\sqrt{N}}\sum_{q=1}^N\ket{q}\right)\otimes\left(\frac{1}{\sqrt{3}}\sum_{\mu=1}^3\ket{\mu}\right)\otimes\left(\sum_{k=1}^{\frac{\log^2(2\Lambda)+\log(2\Lambda)+2}{2}}\sqrt{\frac{w_k}{\sum_kw_k}}\ket{k}\right). \nonumber
\end{eqnarray}
Here $w_k$ are the weight of the unitaries in the LCU decomposition of $E^2$ (Equation \ref{eqn:main:E2_lcu}).
In the first $\log_2N$-qubit electric link index register we store the $N$ electric link indices in equal superposition. In the next 2-qubit spin index register we store the values of $\mu$. Since $E^2$ is a sum of $\frac{\log_2^2(2\Lambda)+\log_2(2\Lambda)+2}{2}=\frac{\log_2^2\Lambda+3\log_2\Lambda+4}{2}$ unitaries (Equation \ref{eqn:main:E2_lcu}), so the last register of $\log_2\left(\log_2^2\Lambda+3\log_2\Lambda+4\right)-1$ qubits stores the indices of the unitaries in a superposition, weighted according to Equation \ref{eqn:main:E2_lcu}.  

The unitary selection sub-routine does the following.
\begin{eqnarray}
    \sel_{f1}:\ket{q}\ket{\mu}\ket{k}\left(\bigotimes_{q=1}^N\bigotimes_{\mu'=1}^3\ket{f_e=0}_{q,\mu'}\right)\ket{\phi}\mapsto\ket{q}\ket{\mu}\ket{k}\left(\ket{1} \right)_{q,\mu}\left(E_k^2\right)_{q,\mu}\ket{\phi}  \nonumber
\end{eqnarray}
and it follows in a straightforward manner that
\begin{eqnarray}
    \bra{0}\prep_{f1}^{\dagger}\cdot\sel_{f1}\cdot\prep_{f1}\ket{0}=\frac{H_{f1}}{3N\Lambda^2/2},  \nonumber
\end{eqnarray}
and here too we keep in mind that $\|H_{f1}\|\leq\frac{3N\Lambda^2}{2}$, from Table \ref{tab:norm}.
More details on the gate complexity for block-encoding can be found in Appendix \ref{subsec:H21}.
We omit discussion of the number of gates required for simulating $e^{-iH_{21}\tau_2}$ using qubitization because it is not required in this paper and also it is quite straightforward to derive.

\begin{lemma}
Let 
$$H_{31}=-\frac{1}{c}\sum_{j=1}^{\eta}\sum_{q=1}^N\sum_{\mu\neq\nu\neq\xi=1}^3\left(\sigma_{j,\mu}\otimes\id\right)\otimes\left(\nabla_{\nu}A_{q,\xi}-\nabla_{\xi}A_{q,\nu}\right)+\frac{1}{2c^2}\sum_{j=1}^{\eta}\sum_{q=1}^N\sum_{\mu=1}^3\id\otimes\id\otimes A_{q,\mu}^2,
$$ where $A_{q,\mu}$ is the operator described in Equation \ref{eqn:A_2}. Then we can have a $(\nconst_{31},-,0)$-block-encoding $H_{31}$, where $\nconst_{31}=\frac{12\pi\eta N\ln 2a^2}{ch\Delta}+\frac{12\pi^2\eta N}{c^2\Delta^2}$, with
\begin{eqnarray}
    \g_{31}'\in O\left(\eta+N(a+\log d)\log d \right) \nonumber
\end{eqnarray}
gates and hence the number of gates required for simulating $e^{iH_{31}\tau_3}$ with error $\delta_{31}>0$, using qubitization is as follows.
\begin{eqnarray}
    \g_{31}&\in&O\left(\frac{\eta^2N\ln (2a^2)}{\Delta^2}\tau_3+\frac{\eta N^2\ln (2a^2)\log d}{\Delta^2}(a+\log d)\tau_3\right. \nonumber \\
    &&\left.+\frac{\log (1/\delta_{31})}{\log\log (1/\delta_{31})}\left(\eta+N(a+\log d)\log d\right)\right)    \nonumber
\end{eqnarray}
    \label{lem:g31}
\end{lemma}
We block encode $H_{31}=H_s+H_{3\pi}$ in a recursive manner, using Theorem \ref{thm:blockEncodeDivConq} repeatedly. 
\begin{eqnarray}
 \text{Let }H_s^{j,q}&=&-\sum_{\mu\neq\nu\neq\xi=1}^3\left(\sigma_{j,\mu}\otimes\id\right)\otimes\left(\nabla_{\nu}A_{q,\xi}-\nabla_{\xi}A_{q,\nu}\right) \nonumber \\
 &=&\sum_{\mu\neq\nu\neq\xi=1}^3\left(\sigma_{j,\mu}\otimes\id\right)\otimes\left(\nabla_{\xi}A_{q,\nu}-\nabla_{\nu}A_{q,\xi}\right)   \nonumber \\
 \text{and } H_{3\pi}^{j,q}&=&\sum_{\mu=1}^3\id\otimes\id\otimes A_{q,\mu}^2, \nonumber \\
 \text{such that }\quad H_{31}^{j,q}&=&\frac{1}{c}H_s^{j,q}+\frac{1}{2c^2}H_{3\pi}^{j,q},\qquad H_{31}^j=\sum_{q=1}^NH_{31}^{j,q},\qquad H_{31}=\sum_{j=1}^{\eta}H_{31}^j.  \nonumber
\end{eqnarray}

\paragraph{Block encoding of $H_s^{j,q}$ : }
The ancillae preparation sub-routine, denoted by $\prep_s^{j,q}$ does the following.
\begin{eqnarray}
\prep_s^{j,q}\ket{0}^*
&=&\left(\frac{1}{\sqrt{6}}\sum_{\mu\neq\nu\neq\xi}^3\sum_{b=0}^1\ket{\mu}\ket{\nu}\ket{\xi}\ket{b}\right) \otimes\left(\sum_{k=-a}^{a}\sqrt{\frac{|d_{2a+1,k}'|}{\sum_k|d_{2a+1,k}'|}}\ket{k+a}\right)  \nonumber \\
&\otimes&\left(\sum_{k'=1}^{\log_2d}\sqrt{\frac{w_{k'}'}{\sum_{k'}w_{k'}'}}\ket{k'}\right)  \nonumber
\end{eqnarray}
 The first $(2\times 3+1)=7$-qubit spin index register stores directions or spins in equal superposition and the last qubit selects between $\nabla_{\nu}A_{q,\xi}$ and $\nabla_{\xi}A_{q,\nu}$. The second and third registers with $\log_2(2a)$ and $\log_2\log_2d$ qubits, respectively, indicate which adder to apply or on which qubit Z-gate should be applied. These are unitaries obtained in the LCU decomposition of $\nabla$ (Lemma \ref{lem:lcuNabla}) and $A$ (Corollary \ref{cor:A_lcu}) in Appendix \ref{app:lcu}. We denote the next sub-routine by $\sel_s^{j,q}$, which is described as follows.
\begin{eqnarray}
&& \sel_s^{j,q}:\ket{\mu,\nu,\xi,0}\ket{k''}\ket{k'}\ket{\phi} \nonumber \\
&\mapsto&\ket{\mu,\nu,\xi,0}\ket{k''}\ket{k'}\left(\sigma_{\mu}\otimes\id\right)_j\left(\nabla_{k''}\right)_{q,\nu}\left(A_{k'}\right)_{q,\xi}\ket{\phi}    \nonumber
\end{eqnarray}
Controlled on $\ket{\mu}$, we apply $\sigma_{\mu}$ on the spin subspace of the $j^{th}$ particle. Controlled on $\ket{\nu,\xi}$ we select spin-subspaces of the $q^{th}$ link register. Controlled on $\ket{k''}$ and $\ket{k'}$ we apply the $k''^{th}$ and $k'^{th}$ unitary in the LCU decompositions of $\nabla$ and $A$, respectively. If the third qubit in the spin register is $\ket{1}$ then we apply $\nabla_{\xi}$ and $A_{\nu}$. It follows that
\begin{eqnarray}
    \bra{0}\prep_{s}^{j,q\dagger}\cdot\sel_{s}^{j,q}\cdot\prep_{s}^{j,q}\ket{0}=\frac{H_{s}^{j,q}}{12\pi\ln 2a^2/h\Delta},  \nonumber
\end{eqnarray}
and thus we have a $(12\pi\ln 2a^2/h\Delta,.,0)$-block encoding of $H_s^{j,q}$.
  
\paragraph{Block encoding of $H_{3\pi}^{j,q}$ :} The first ancillae preparation sub-routine is described as follows.
\begin{eqnarray}
&& \prep_{3\pi}^{j,q}\ket{0}^{*}=
\left(\frac{1}{\sqrt{3}}\sum_{\mu'=1}^3\ket{\mu'}\right)\otimes\left(\sum_{k=1}^{\frac{\log^2d+\log d}{2}}\sqrt{\frac{w_k'}{\sum_{k}w_k'}}\ket{k}\right)    \nonumber
\end{eqnarray}
The first 2-qubit register is the spin index register. Since $A^2$ is a sum of $\frac{\log_2^2d+\log_2d}{2}$ unitaries (Table \ref{tab:norm}), so we prepare a $\log_2\left(\frac{\log_2^2d+\log_2d}{2}\right)$-qubit register in a superposition weighted according to the LCU decomposition of $A^2$ (Corollary \ref{cor:A2_lcu} in Appendix \ref{app:lcu}).  
The next sub-routine is described as follows.
\begin{eqnarray}
&& \sel_{3\pi}^{j,q}:\ket{\mu'}\ket{k}\ket{\phi}
\mapsto\ket{k}\left(A_k^2\right)_{q,\mu'}\ket{\phi} \nonumber
\end{eqnarray}
It follows that
\begin{eqnarray}
    \bra{0}\prep_{3\pi}^{j,q\dagger}\cdot\sel_{3\pi}^{j,q}\cdot\prep_{3\pi}^{j,q}\ket{0}=\frac{H_{3\pi}^{j,q}}{24\pi^2/\Delta^2},  \nonumber
\end{eqnarray}
and thus we have a $(24\pi^2/\Delta^2,.,0)$-block encoding of $H_{3\pi}^{j,q}$.

\paragraph{Block encoding of $H_{31}$ : }  We use Theorem \ref{thm:blockEncodeDivConq} repeatedly. First we block encode $H_{31}^{j,q}=\frac{1}{c}H_s^{j,q}+\frac{1}{2c^2}H_{3\pi}^{j,q}$ with $O(1)$ extra gate cost. Next we consider $H_{31}^{j}=\sum_{q=1}^NH_{31}^{j,q}$, where each of the summand Hamiltonians act on separate link registers and finally we consider $H_{31}=\sum_{j=1}^{\eta}H_{31}^j$, where again the summands act on separate subspaces. Thus the overall ancilla preparation sub-routine is,
\begin{eqnarray}
    \prep_{31}\ket{0}^* &=&
    \left(\frac{1}{\sqrt{\eta}}\sum_{j=1}^{\eta}\ket{j}\right)\otimes\left(\frac{1}{\sqrt{N}}\sum_{q=1}^{N}\ket{q}\right) 
    \otimes\left(\sqrt{\frac{\lambda_s}{c\nconst}}\ket{0}+\sqrt{\frac{\lambda_{3\pi}}{2c^2\nconst}}\ket{1}\right)   \nonumber \\
    &&\otimes\prep_{s}^{j,q}\otimes\prep_{3\pi}^{j,q},  \nonumber
\end{eqnarray}
where $\lambda_s=\|H_s^{j,q}\|=\frac{12\pi\ln 2a^2}{h\Delta}$, $\lambda_{3\pi}=\|H_{3\pi}^{j,q}\|=\frac{24\pi^2}{\Delta^2}$ and $\nconst=\frac{\lambda_s}{c}+\frac{\lambda_{3\pi}}{2c^2}$. The overall unitary selection sub-routine is as follows. 
\begin{eqnarray}
    \sel_{31}:\ket{j,q,0}\ket{\mu,\nu,\xi,b,k'',k'}\ket{\mu',k}\ket{\phi}&\mapsto&\ket{j,q,0}\ket{\mu',k}\sel_s^{j,q}\left(\ket{\mu,\nu,\xi,b,k'',k'}\ket{\phi}\right)   \nonumber \\
    \sel_{31}:\ket{j,q,1}\ket{\mu,\nu,\xi,b,k'',k'}\ket{\mu',k}\ket{\phi}&\mapsto&\ket{j,q,1}\ket{\mu,\nu,\xi,b,k'',k'}\sel_{3\pi}^{j,q}\left(\ket{\mu',k}\ket{\phi}\right)   \nonumber
\end{eqnarray}
It is straightforward to check that
\begin{eqnarray}
    \bra{0}\prep_{31}^{\dagger}\cdot\sel_{31}\cdot\prep_{31}\ket{0}=\frac{H_{31}}{\eta N\nconst},  \nonumber
\end{eqnarray}
where $\eta N\nconst=\frac{12\pi\eta N\ln 2a^2}{ch\Delta}+\frac{12\pi^2\eta N}{c^2\Delta^2}$, which is also the sum of the norms of the Hamiltonians $H_s$ and $H_{3\pi}$ in Table \ref{tab:norm}. Thus we have a $(\eta N\nconst,.,0)$-block encoding of $H_{31}$. More details about the procedures and gate complexity analysis can be found in Appendix \ref{subsec:H31}.

\begin{lemma}
Let 
\begin{eqnarray}
H_{32}&=&\frac{1}{ \Delta} \sum_{k < j}^{\eta} \sum_{\q,\vect{r}=1}^N \left(\id\otimes\frac{1}{ \,||\q - \vect{r}||_2} \, | \q \rangle \langle \q |_k \,| \vect{r} \rangle \langle \vect{r} |_j\otimes\id\right)  - \frac{1}{\Delta} \sum_{j=1}^{\eta} \sum_{\kappa=1}^{K} \sum_{\q=1}^N \left(\id\otimes \frac{Z_{\kappa}}{ |\q - \vect{R}_{\kappa}\|_2}  | \q \rangle \langle \q |_j \otimes\id\right) \nonumber \\
&&\sum_{j=1}^{\eta}\sum_{q=1}^N\sum_{\mu=1}^3\left(\id\otimes \left(-\frac{1}{2}\nabla_{j,\mu}^2\right)\otimes\id+\frac{1}{c}\id\otimes \left(\frac{i}{c}\nabla_{j,\mu}\right) \otimes A_{q,\mu}\right) \nonumber
\end{eqnarray}
where $A_{q,\mu}$ is the operator described in Equation \ref{eqn:A_2}. Then we can have a $(\nconst_{32},-,0)$-block-encoding of $H_{32}$, where $\frac{\eta(\eta-1)}{2\Delta^2}+\frac{\eta Z_{sum}}{\Delta^2}+\frac{8\pi^2\eta N}{h^2}+\frac{12\pi\eta N\ln 2a^2}{ch\Delta}$, with
\begin{eqnarray}
    \g_{32}'\in O\left(\eta a\log_2N+N\log_2d+\log_2N\log_2\frac{N}{\delta'}+K\log_2\frac{1}{\delta''}\right) \nonumber
\end{eqnarray}
gates, where $\delta',\delta''>0$. Let $R_{32}\in O\left(\frac{\eta N}{\Delta^2}\left(1+\frac{\eta_s}{N}\right)\tau_3+\frac{\log(1/\delta_{32})}{\log\log(1/\delta_{32})}\right)$. 
Then the number of gates required for simulating $e^{iH_{32}\tau_3}$ with error $\delta_{32}>0$, using qubitization is
$
    \g_{32}\in R_{32}\cdot \g_{32}'. 
$
    \label{lem:g32}
\end{lemma}
Again we use Theorem \ref{thm:blockEncodeDivConq} to block encode $H_{32}=H_V+H_{1\pi}+H_{2\pi}$ in a recursive manner. We define the following.
\begin{eqnarray}
   && H_{1\pi}^{j,q,\mu}=-\id\otimes\nabla_{j,\mu}^2\otimes\id,\qquad H_{2\pi}^{j,q,\mu}=\id\otimes\left(i\nabla_{j,\mu}\right)\otimes A_{q,\mu}  \nonumber \\
   &&H_{12\pi}^{j,q,\mu}=\frac{1}{2}H_{1\pi}^{j,q,\mu}+\frac{1}{c}H_{2\pi}^{j,q,\mu},\qquad H_{12\pi}=\sum_{j=1}^{\eta}\sum_{q=1}^N\sum_{\mu=1}^3H_{12\pi}^{j,q,\mu}   \nonumber
\end{eqnarray}

\paragraph{Block encoding of $H_{12\pi}$ : } As in the case of $H_{31}$, we first block encode $H_{1\pi}^{j,q,\mu}$ and $H_{2\pi}^{j,q,\mu}$ separately using the ancillae preparation sub-routines $\prep_{1\pi}^{j,q,\mu}$ and $\prep_{2\pi}^{j,q,\mu}$ respectively, followed by the unitary selection sub-routines $\sel_{1\pi}^{j,q,\mu}$ and $\sel_{2\pi}^{j,q,\mu}$ respectively. Then we block encode $H_{12\pi}^{j,q,\mu}$ and $H_{12\pi}$, as discussed in Theorem \ref{thm:blockEncodeDivConq}. Thus our overall ancillae preparation sub-routine is as follows.
\begin{eqnarray}
    \prep_{12\pi}\ket{0}^* 
 &=&\left(\frac{1}{\sqrt{\eta}}\sum_{j=1}^{\eta}\ket{j}\right)\otimes\left(\frac{1}{\sqrt{N}}\sum_{q=1}^N\ket{q}\right)\otimes\left(\frac{1}{\sqrt{3}}\sum_{\mu=1}^3\ket{\mu}\right) \nonumber \\
  &&\otimes\left(\sqrt{\frac{\lambda_1}{2\nconst'}}\ket{0}+\sqrt{\frac{\lambda_2}{c\nconst'}}\ket{1}\right)  
   \otimes\prep_{1\pi}^{j,q,\mu}\otimes\prep_{2\pi}^{j,q,\mu}, \nonumber
\end{eqnarray}
  where $\lambda_1=\|2H_{1\pi}\|,\lambda_2=\|cH_{2\pi}\|,\nconst'=\frac{\lambda_1}{2}+\frac{\lambda_2}{c}=\frac{8\pi^2\eta N}{h^2}+\frac{12\pi\eta N\ln 2a^2}{ch\Delta}$ and
  \begin{eqnarray}
   \prep_{1\pi}^{j,q,\mu}\ket{0}^* &=& \left(\sum_{k=-a}^a\sqrt{\frac{|d_{2a+1,k}|}{\sum_{k}|d_{2a+1,k}|}}\ket{k+a}\right); \nonumber \\
\prep_{2\pi}^{j,q,\mu}\ket{0}^* &=& \left(\sum_{k_1=-a}^a\sqrt{\frac{|d_{2a+1,k_1}''|}{\sum_{k_1}|d_{2a+1,k_1}''|}}\ket{k_1+a}\right)\otimes\left(\sum_{k_2=1}^{\log_2d}\sqrt{\frac{w_{k_2}}{\sum_{k_2}w_{k_2}}}\ket{k_2}\right) \nonumber
\end{eqnarray}
The overall unitary selection sub-routine is as follows.
\begin{eqnarray}
&&\sel_{1\pi}^{j,q,\mu}:\ket{k'}\ket{\phi}
\mapsto\ket{k'}\left(\id\otimes\nabla_{k'}^2\right)_{j,\mu}\ket{\phi}    \label{sel:1pijqmu} \nonumber\\
&&\sel_{2\pi}^{j,q,\mu}:\ket{k_1'}\ket{k_2'}\ket{\phi}\mapsto\ket{k_1'}\ket{k_2'}\left(\nabla_{k_1'}\right)_{j,\mu}\left(A_{k_2'}\right)_{q,\mu}\ket{\phi} \label{sel:2pijqmu} \nonumber\\
   && \sel_{12\pi}:\ket{j,q,\mu,0}\ket{k'}\ket{k_1',k_2'}\ket{\phi}   
   \mapsto\ket{j,q,\mu,0}\ket{k_1',k_2'}\sel_{1\pi}^{j,q,\mu}\left(\ket{k'}\ket{\phi}\right)    \nonumber \\
  && \sel_{12\pi}:\ket{j,q,\mu,1}\ket{k'}\ket{k_1',k_2'}\ket{\phi} 
  \mapsto\ket{j,q,\mu,1}\ket{k'}\sel_{2\pi}^{j,q,\mu}\left(\ket{k_1',k_2'}\ket{\phi}\right)  \nonumber 
\end{eqnarray}
It can be verified in a straightforward manner that
\begin{eqnarray}
    \bra{0}\prep_{12\pi}^{\dagger}\cdot\sel_{12\pi}\cdot\prep_{12\pi}\ket{0}=\frac{H_{12\pi}}{\nconst'},  \nonumber
\end{eqnarray}
where $\nconst'=\frac{8\pi^2\eta N}{h^2}+\frac{12\pi\eta N\ln 2a^2}{ch\Delta}$, which is also the sum of the norms of $H_{1\pi}$ and $H_{2\pi}$ (Table \ref{tab:norm}).

\paragraph{Block encoding of $H_V$ : } We know that $H_V=H_{V_{ee}}+H_{V_{ne}}$ and we block encode it, following the approach taken in \cite{2019_BBMN, 2021_SBWetal}, with some modifications and incorporating the optimizations in Theorem \ref{thm:CX}. The ancilla preparation sub-routine is as follows.
\begin{eqnarray}
    \prep_V\ket{1}\ket{0}^*&\propto&\ket{0}\sum_{i<j}^{\eta}\sum_{v_x,v_y,v_z=-N^{1/3}}^{N^{1/3}}\frac{1}{\|\vect{v}\|_2}\ket{i}\ket{j}\ket{v_x,v_y,v_z}    \nonumber \\
 &&-\ket{1}\sum_{i=1}^{\eta}\sum_{\kappa=1}^K\sum_{v_x,v_y,v_z=-N^{1/3}}^{N^{1/3}}\frac{\sqrt{Z_{\kappa}}}{\|\vect{v}\|_2}\ket{i}\ket{\kappa}\ket{v_x,v_y,v_z}  \nonumber
\end{eqnarray}
The first ancilla is used to select between $H_{V_{ee}}$ and $H_{V_{ne}}$. Next the $\log_2\eta$-qubit register stores the particle indices in equal superposition. We follow the state preparation procedure, described in \cite{2019_BBMN}, to prepare $\sum_{v_x,v_y,v_z=-N^{1/3}}^{N^{1/3}}\frac{1}{\|v\|_2}\ket{\vect{v}}$. This has been described in Appendix \ref{app:statePrep}. 
The unitary selection sub-routine is described as follows.
\begin{eqnarray}
&& \sel_V:\ket{0}\ket{i}\ket{j}\ket{\vect{v}}\ket{\vect{q}_1,\ldots \vect{q}_i,\ldots \vect{q}_j,\ldots \vect{q}_{\eta}}\ket{0}  \nonumber\\
&&\mapsto\ket{0}\ket{i}\ket{j}\ket{\vect{v}}\ket{\vect{q}_1,\ldots \vect{q}_i,\ldots \vect{q}_j,\ldots \vect{q}_{\eta}}\ket{\vect{q}_i-\vect{q}_j} \nonumber \\
&& \sel_V:\ket{1}\ket{i}\ket{\kappa}\ket{\vect{v}}\ket{\vect{q}_1,\ldots \vect{q}_i,\ldots \vect{q}_{\eta}}\ket{0}  \nonumber\\
&&\mapsto\ket{1}\ket{i}\ket{\kappa}{\vect{v}}\ket{\vect{q}_1,\ldots \vect{q}_i,\ldots  \vect{q}_{\eta}}\ket{\vect{R}_{\kappa}-\vect{q}_i}   \nonumber
\end{eqnarray}
If the first qubit is $\ket{0}$ we discard if $\vect{q}_i-\vect{q}_j\neq \vect{v}$. Since for each pair of $\vect{q}_i, \vect{q}_j$, only one value of $\vect{v}$ survives, so the probability distribution is unaffected. 
If the first register is $\ket{1}$ then we use a classical database to access $\vect{R}_{\kappa}$. Controlled on the  particle index register we take the difference $\vect{R}_{\kappa}-\vect{q}_i$ and discard the computational path if it is not equal to $\vect{v}$. 

\paragraph{Block encoding of $H_{32}$ : } Since $H_V$ has a probabilistic ancilla preparation sub-routine, we can block encode $H_{32}=H_V+H_{12\pi}$ using the procedure described in \cite{2021_SBWetal}, by repeating the $\prep_V$ sub-routine constant number of times. This does not change the asymptotic gate complexity. In Appendix \ref{subsec:H32} we have given more detail description about these block encoding procedures and gate complexity.

\paragraph{Trotter error :}
Now we derive bounds on the Trotter errors $\epsilon_1$, $\epsilon_2$ and $\epsilon_3$, thus bounding the simulation error described in Lemma \ref{lem:divNconq}.
If a Hamiltonian $H=\sum_{\gamma=1}^{\Gamma}H_{\gamma}$ is a sum of $\Gamma$ fragment Hamiltonians, then $e^{-itH}$ can be approximated by product of exponentials, using the $p^{th}$ order Trotter-Suzuki formula \cite{1991_S}, $\mathscr{S}_p(t)=e^{-itH}+\mathscr{A}(t)$, where $\|\mathscr{A}(t)\|\in O\left(\widetilde{\alpha}_{comm} t^{p+1}\right)$ if each $H_{\gamma}$ are Hermitian \cite{2021_CSTetal}. Here 
$\widetilde{\alpha}_{comm}=\sum_{\gamma_1,\gamma_2,\ldots,\gamma_{p+1}=1}^{\Gamma}\|[H_{\gamma_{p+1}},\ldots[H_{\gamma_2},H_{\gamma_1}]]\|$. 
The following result provides an upper bound on $\widetilde{\alpha}_{comm}$.  The proof is given in Appendix \ref{app:comm}, where we have also explained some variations that can lead to tighter bounds.
\begin{lemma}
Let $H=\sum_{\gamma=1}^{\Gamma}H_{\gamma}$ and $\widetilde{\alpha}_{comm}=\sum_{\gamma_1,\gamma_2,\ldots,\gamma_{p+1}=1}^{\Gamma}\|[H_{\gamma_{p+1}},\ldots[H_{\gamma_2},H_{\gamma_1}]]\|$. Then for any integer $1\leq p'\leq p$,
\begin{eqnarray}
    \widetilde{\alpha}_{comm}\leq 2^{p-(p'+1)}\sum_{\gamma_{i_1},\gamma_{i_2},\ldots,\gamma_{i_{p'+1}}} \|[H_{\gamma_{p'+1}},[\ldots[H_{\gamma_3},[H_{\gamma_2},H_{\gamma_1}]]\ldots]]\|  \left(\sum_{\gamma=1}^{\Gamma}\|H_{\gamma}\|\right)^{p-p'}.    \nonumber
\end{eqnarray}
    \label{lem:alpha_comm}
\end{lemma}
In this paper we take $p'=1$ and thus the need to compute all the first order commutators in Table \ref{tab:comm}, as well as the norm in Table \ref{tab:norm}.
Now we have the results needed to prove our main theorem about divide and conquer simulations in Theorem~\ref{thm:DC}
\begin{proof}[Proof of Theorem~\ref{thm:DC}]{
 It is clear that we can bound the Trotter errors due to repeated splitting of the Hamiltonian $\hat{H}_{PF}$, using the bounds in Table \ref{tab:norm} and \ref{tab:comm}. In the rest of the paper we assume $h\leq K_h\Delta$, for some constant $K_h$, $a$ is a constant and let \begin{equation}
     \eta_s:=\eta+Z_{sum},
 \end{equation} where $Z_{sum}=\sum_{\kappa=1}^K|Z_{\kappa}|$. In the first level (Figure \ref{fig:divNconq}) we have two partitions - $H_{11}=H_{f1}+H_s+H_V+H_{\pi}$ and $H_{12}=H_{f2}$ and the error introduced due to this split (Equation \ref{eqn:U-Uhat}) is
\begin{eqnarray}
 \epsilon_1&\in& O\left(\widetilde{\alpha}_{1comm}(t/r_1)^{p_1+1}\right),    
\end{eqnarray}
where
\begin{equation}
    \widetilde{\alpha}_{1comm}\leq 2^{p_1-2}\|[H_{f2},H_{f1}+H_s+H_V+H_{\pi}]\|\cdot\left(\|H_{\pi}\|+\|H_V\|+\|H_f\|+\|H_s\|\right)^{p_1-1}. 
\end{equation}
Using the bounds in Table \ref{tab:norm} we get
\begin{eqnarray}
     && \|H_{\pi}\|+\|H_V\|+\|H_f\|+\|H_s\| \nonumber \\
 &\leq& \frac{12\pi^2\eta N}{c^2\Delta^2}+\frac{8\pi^2\eta N}{h^2}+\frac{24\pi\eta N\ln (2a^2)}{ch\Delta}+\frac{\eta(\eta-1)}{2\Delta^2}+\frac{\eta Z_{sum}}{\Delta^2}+\frac{3N\Lambda^2}{2}+6N    \nonumber \\
 &=&\frac{\eta N}{\Delta^2}\left(8\pi^2\frac{\Delta^2}{h^2}+\frac{12\pi^2}{c^2}+\frac{24\pi\ln (2a^2)}{c}\cdot\frac{\Delta}{h}+\frac{\eta+2Z_{sum}}{2N}+\frac{6\Delta^2}{\eta}+\frac{3\Lambda^2\Delta^2}{2\eta}\right)  \nonumber \\
 &\lessapprox& K_1\frac{\eta N}{\Delta^2}\left(1+\frac{\eta_s}{N}+\frac{\Delta^2\Lambda^2}{\eta}\right) \label{eqn:Hpf_norm} \qquad [\text{for some constant }K_1].
\end{eqnarray}
From Table \ref{tab:comm} we get
\begin{eqnarray}
    &&\|[H_{f2},H_{f1}+H_s+H_{\pi}+H_V]\| \nonumber \\
&\leq& \|[H_{f2},H_{f1}]\|+\|[H_{f2},H_s]\|+\|[H_{f2},H_{\pi}]\|+\|[H_{f2},H_V]\|   \nonumber   \\
&\leq&12N\Lambda+\frac{288\pi\eta N\ln(2a^2)}{ch\Delta}+\frac{198\pi\eta N}{c\Delta}\left(\frac{\ln (2a^2)}{h}+\frac{\pi}{c\Delta}\right) \nonumber   \\
&=&\frac{\eta N}{\Delta^2}\left(12\frac{\Delta^2\Lambda}{\eta}+\frac{486\pi\ln (2a^2)}{c}\cdot\frac{\Delta}{h}+\frac{198\pi^2}{c^2}\right)    \nonumber \\
&\leq& K_2\frac{\eta N}{\Delta^2}\left(1+\frac{\Delta^2\Lambda}{\eta}\right)\qquad [\text{for some constant }K_2] \nonumber
\end{eqnarray}
and so,
\begin{eqnarray}
 \widetilde{\alpha}_{1comm} 
\leq 2^{p_1-2}K_1K_2^{p_1-1}\left(\frac{\eta N}{\Delta^2}\right)^{p_1}\left(1+\frac{\Delta^2\Lambda}{\eta}\right)\left(1+\frac{\eta_s}{N}+\frac{\Delta^2\Lambda^2}{\eta}\right)^{p_1-1}  \nonumber 
\end{eqnarray}
and hence,
\begin{eqnarray}
\epsilon_1 \in O\left(\left(\frac{\eta N}{\Delta^2}\right)^{p_1}\left(1+\frac{\Delta^2\Lambda}{\eta}\right)\left(1+\frac{\eta_s}{N}+\frac{\Delta^2\Lambda^2}{\eta}\right)^{p_1-1}\left(\frac{t}{r_1}\right)^{p_1+1} \right) 
\label{eqn:e1}
\end{eqnarray}
In the second level (Figure \ref{fig:divNconq}) we have partitioned $H_{11}$ into $H_{21}=H_{f1}$ and $H_{22}=H_s+H_V+H_{\pi}$ and so the error (Equation \ref{eqn:U1-U1hat}) introduced is
\begin{eqnarray}
    \epsilon_2 &\in& O\left(\widetilde{\alpha}_{2comm}(t/r_1r_2)^{p_2+1}\right), \nonumber \\
    \text{where}\quad\widetilde{\alpha}_{2comm}&\leq&2^{p_2-2}\|[H_{f1},H_s+H_V+H_{\pi}]\|\cdot\left(\|H_{f1}\|+\|H_s\|+\|H_V\|+\|H_{\pi}\|\right)^{p_2-1}. \nonumber
\end{eqnarray}
From Table \ref{tab:comm} we have
\begin{eqnarray}
    \|[H_{f1},H_s+H_V+H_{\pi}]\| 
    &\leq& \frac{6\pi\eta N\Lambda^2}{c\Delta}\left(\frac{\ln (2a^2)}{h}+\frac{2\pi}{c\Delta}\right)+\frac{24\pi\eta N\Lambda^2\ln (2a^2)}{ch\Delta}\leq K_3\frac{\eta N\Lambda^2}{\Delta^2},    \nonumber
\end{eqnarray}
for some constant $K_3$ and from Table \ref{tab:norm}, somewhat similar to Equation \ref{eqn:Hpf_norm}, we have
\begin{eqnarray}
    \|H_{f1}\|+\|H_s\|+\|H_V\|+\|H_{\pi}\| 
   \leq K_4\frac{\eta N}{\Delta^2}\left(1+\frac{\eta_s}{N}+\frac{\Delta^2\Lambda^2}{\eta}\right).  \nonumber
\end{eqnarray}
So,
$
    \widetilde{\alpha}_{2comm}\leq K_3K_4^{p_2-1}\Lambda^22^{p_2-2}\left(\frac{\eta N}{\Delta^2}\right)^{p_2}\left(1+\frac{\eta_s}{N}+\frac{\Delta^2\Lambda^2}{\eta}\right)^{p_2-1}
$
and hence,
\begin{eqnarray}
    \epsilon_2\in O\left(\Lambda^2\left(\frac{\eta N}{\Delta^2}\right)^{p_2}\left(1+\frac{\eta_s}{N}+\frac{\Delta^2\Lambda^2}{\eta}\right)^{p_2-1}\left(\frac{t}{r_1r_2}\right)^{p_2+1}  \right).    \label{eqn:e2}
\end{eqnarray}
In the third level of the divide and conquer algorithm (Figure \ref{fig:divNconq})  $H_{22}$ is divided into $H_{31}=H_s+H_{3\pi}$ and $H_{32}=H_V+H_{1\pi}+H_{2\pi}$ and so the error (Equation \ref{eqn:U2-U2hat}) is
\begin{eqnarray}
    \epsilon_3 &\in& O\left(\widetilde{\alpha}_{3comm}(t/r_1r_2r_3)^{p_3+1}\right), \nonumber \\
    \text{where }\quad\widetilde{\alpha}_{3comm}&\leq&2^{p_3-2}\|[H_s+H_{3\pi},H_V+H_{1\pi}+H_{2\pi}]\|\cdot\left(\|H_s\|+\|H_V\|+\|H_{\pi}\|\right)^{p_3-1}.   \nonumber
\end{eqnarray}
From Table \ref{tab:comm} we have,
\begin{eqnarray}
 \|[H_{31},H_{32}]\|&=&\|[H_V,H_s]+[H_{1\pi},H_s]+[H_{2\pi},H_s]+[H_V,H_{3\pi}]+[H_{1\pi},H_{3\pi}]+[H_{2\pi},H_{3\pi}]\|   \nonumber \\
 &=&\|[H_{2\pi},H_s]\|\leq K_5\frac{\eta^2 N}{\Delta^4} \qquad [\text{for some constant } K_5]\nonumber
\end{eqnarray}
and from Table \ref{tab:norm} we have in the non-relativistic limit where $\log(a)/c\Delta^2 \in O(1)$ and since $N\ge \eta$
\begin{eqnarray}
 &&\|H_s\|+\|H_V\|+\|H_{\pi}\|    \nonumber \\
 &\leq& \frac{8\pi^2\eta N}{h^2}+\frac{12\pi^2\eta N}{c^2\Delta^2}+\frac{24\pi\eta N\ln (2a^2)}{ch\Delta}+\frac{\eta(\eta-1)}{2\Delta^2}+\frac{\eta Z_{sum}}{\Delta^2}  \nonumber \\
 &\leq& K_6\frac{\eta N}{\Delta^2}\left(1+\frac{\eta_s}{N}\right)\qquad [\text{for some constant }K_6].   \nonumber
\end{eqnarray}
So, 
$
  \widetilde{\alpha}_{3comm}\leq K_5K_6^{p_3-1}\frac{\eta}{\Delta^2}2^{p_3-2}\left(\frac{\eta N}{\Delta^2}\right)^{p_3}\left(1+\frac{\eta_s}{N}\right)^{p_3-1}
$,
and hence
\begin{eqnarray}
    \epsilon_3\in O\left(\frac{\eta}{\Delta^2}\left(\frac{\eta N}{\Delta^2}\right)^{p_3}\left(1+\frac{\eta_s}{N}\right)^{p_3-1}\left(\frac{t}{r_1r_2r_3}\right)^{p_3+1}\right). \label{eqn:e3}
\end{eqnarray}
We can bound the overall simulation error $\epsilon$ using Lemma \ref{lem:divNconq}, where we take $\term_i\leq 2\cdot 5^{p_i-1}$. Using Equations \ref{eqn:e1}, \ref{eqn:e2} and \ref{eqn:e3} we get,
\begin{eqnarray}
    \epsilon &\leq& r_1\epsilon_1+r_1r_2\term_1\epsilon_2+r_1r_2r_3\term_1\term_2\epsilon_3+r_1\term_1\delta_1+r_1r_2\term_1\term_2\delta_2+r_1r_2r_3\term_1\term_2\term_3(\delta_{31}+\delta_{32})     \nonumber   \\
    &\in& O\left(t^{p_1+1}\left(\frac{\eta N}{r_1\Delta^2}\right)^{p_1}\left(1+\frac{\Delta^2\Lambda}{\eta}\right)\left(1+\frac{\eta_s}{N}+\frac{\Delta^2\Lambda^2}{\eta}\right)^{p_1-1} \right.\nonumber \\
    &&\left.+5^{p_1}\Lambda^2t^{p_2+1}\left(\frac{\eta N}{\Delta^2r_1r_2}\right)^{p_2}\left(1+\frac{\eta_s}{N}+\frac{\Delta^2\Lambda^2}{\eta}\right)^{p_2-1}  \right.  \nonumber \\
    &&\left.+5^{p_1+p_2}\frac{\eta}{\Delta^2}t^{p_3+1}\left(\frac{\eta N}{\Delta^2r_1r_2r_3}\right)^{p_3}\left(1+\frac{\eta_s}{N}\right)^{p_3-1}+r_15^{p_1}\delta_1+r_1r_25^{p_1+p_2}\delta_2   \right. \nonumber \\
    &&\left.+r_1r_2r_35^{p_1+p_2+p_3}(\delta_{31}+\delta_{32})  \right) \label{eqn:e}
\end{eqnarray}
Also, from Lemma \ref{lem:divNconq} we have the following bound on total gate complexity,
\begin{eqnarray}    
    \g &\in& O\left(5^{p_1}r_1\g_1+5^{p_1+p_2}r_1r_2\g_2+5^{p_1+p_2+p_3}r_1r_2r_3(\g_{31}+\g_{32}) \right),  \label{eqn:g}
\end{eqnarray}
where $\g_1, \g_2, \g_{31}, \g_{32}$ are the gate complexities given in Lemma \ref{lem:g1}-\ref{lem:g32}.  We make the additional assumption that $\delta'=\delta''$ and observe that because of our choice of asymetric cutoffs on the field $2\Lambda=d$, which follows from the dimension of operator $A$ that acts on the link space. Then substituting these values we find that
\begin{eqnarray}
    \g_1&\in& O\left(N^2\frac{t}{r_1}\log \Lambda+\frac{\log (1/\delta_1)}{\log\log (1/\delta_1)}N\log\Lambda\right)    \nonumber \\
    \g_2&\in& O\left(N\log^2\Lambda\right) \nonumber \\
    \g_{31}&\in& O\left(\frac{\eta^2N}{\Delta^2}\frac{t}{r_1r_2r_3}+\frac{\eta N^2\log^2 \Lambda}{\Delta^2}\frac{t}{r_1r_2r_3}+\frac{\log (1/\delta_{31})}{\log\log (1/\delta_{31})}(\eta+N\log^2\Lambda)\right)   \nonumber   \\
    \g_{32}&\in& R_{32}\cdot\g_{32}'    \nonumber \\
    R_{32}&\in& O\left(\frac{\eta N}{\Delta^2}\left(1+\frac{\eta_s}{N}\right)\frac{t}{r_1r_2r_3}+\frac{\log(1/\delta_{32})}{\log\log (1/\delta_{32})}\right) \nonumber  \\
    \g_{32}'&\in& O\left(\eta \log N+N\log\Lambda+\log N\log\frac{N}{\delta'}+K\log\frac{1}{\delta'}\right).    \nonumber 
\end{eqnarray}
In principle, the least upper bound on the cost of the simulation can be found by optimizing over $r_1, r_2, r_3,$ $p_1, p_2, p_3, \delta_1, \delta_2, \delta_{31}$, $\delta_{32}$, while ensuring the constraint on the overall error in Equation \ref{eqn:e}. However, this is a difficult non-linear optimization problem and the true optima is difficult to find. Nonetheless, any choice of values will yield an upper bound on the complexity and for simplicity we take the simplest choice that satisfies the bound on the error in Equation \ref{eqn:e}. We take the orders of the splitting formulas $p_1=p_2=p_3$ to be the same,
$r_1\in O\left(\left(\frac{t}{\epsilon}\right)^{\frac{1}{p_1}}\frac{t\eta N}{\Delta^2}\left(1+\frac{\eta_s}{N}+\frac{\Delta^2\Lambda^2}{\eta}\right)^{1-\frac{1}{p_1}}\left(1+\frac{\Delta^2\Lambda}{\eta}\right)^{\frac{1}{p_1}}\right)$, $r_2\in O\left(\Lambda^{\frac{2}{p_1}}\left(1+\frac{\Delta^2\Lambda}{\eta}\right)^{-\frac{1}{p_1}}\right)$, 
$r_3\in O\left(\left(\frac{\eta}{\Delta^2\Lambda^2}\right)^{\frac{1}{p_1}}\right)$. We take $\delta_1\in O\left(\frac{\epsilon}{r_1}\right)\in O\left(\left(\frac{\epsilon}{t}\right)^{1+\frac{1}{p_1}}\left(\frac{\eta N}{\Delta^2}\right)^{-1}\left(1+\frac{\eta_s}{N}+\frac{\Delta^2\Lambda^2}{\eta}\right)^{\frac{1}{p_1}-1}\right)$, $\delta_2\in O\left(\frac{\epsilon}{r_1r_2}\right)$, $\delta_{31},\delta_{32}\in O\left(\frac{\epsilon}{r_1r_2r_3}\right)$. Since $r_2,r_3\geq 1$, so for simplicity we choose $\delta_2,\delta_{31},\delta_{32}\in O(\delta_1)$. $\delta'$ is included in the error of block encoding $H_{32}$. So we can assume that $\delta'\in O\left(\delta_1\right)$.

For brevity, let 
\begin{equation}
    L_1:=1+\frac{\eta_s}{N}+\frac{\Delta^2\Lambda^2}{\eta}.
\end{equation} Then we have,
\begin{eqnarray}
   && r_1\g_1+r_1r_2\g_2    \nonumber \\
   &\in& O\left(N^2t\log\Lambda+\frac{\log(1/\delta_1)}{\log\log (1/\delta_1)}\left(\frac{t}{\epsilon}\left(1+\frac{\Delta^2\Lambda}{\eta}\right)\right)^{\frac{1}{p_1}}L_1^{1-\frac{1}{p_1}}\frac{t\eta N^2\log\Lambda}{\Delta^2}+\left(\frac{t}{\epsilon}\right)^{\frac{1}{p_1}}L_1^{1-\frac{1}{p_1}}\frac{t\eta N^2\log^2\Lambda}{\Delta^2}\Lambda^{\frac{2}{p_1}}  \right)    \nonumber   \\
   &\in& O\left(\frac{\log(1/\delta_1)}{\log\log (1/\delta_1)}\left(\frac{t}{\epsilon}\left(1+\frac{\Delta^2\Lambda}{\eta}\right)\right)^{\frac{1}{p_1}}L_1^{1-\frac{1}{p_1}}\frac{t\eta N^2\log\Lambda}{\Delta^2}+\left(\frac{t}{\epsilon}\right)^{\frac{1}{p_1}}L_1^{1-\frac{1}{p_1}}\frac{t\eta N^2\log^2\Lambda}{\Delta^2}\Lambda^{\frac{2}{p_1}} \right) \label{eqn:r1g1+r1r2g2}  
\end{eqnarray}
and
\begin{eqnarray}
    r_1r_2r_3\g_{31}&\in& O\left(\frac{\eta^2Nt}{\Delta^2}+\frac{\eta N^2t\log^2\Lambda}{\Delta^2}+\frac{\log (1/\delta_1)}{\log\log (1/\delta_1)}\left(\frac{t}{\epsilon}\right)^{\frac{1}{p_1}}L_1^{1-\frac{1}{p_1}}\left(\frac{\eta^2Nt}{\Delta^2}+\frac{t\eta N^2\log^2\Lambda}{\Delta^2} \right)\left(\frac{\eta}{\Delta^2}\right)^{\frac{1}{p_1}} \right)   \nonumber \\
    &\in&O\left(\frac{\log (1/\delta_1)}{\log\log (1/\delta_1)}\left(\frac{t\eta}{\epsilon\Delta^2}\right)^{\frac{1}{p_1}}L_1^{1-\frac{1}{p_1}}\left(\frac{\eta^2Nt}{\Delta^2}+\frac{t\eta N^2\log^2\Lambda}{\Delta^2} \right)  \right). \label{eqn:r1r2r3g31} 
\end{eqnarray}
Now, let $\eta_s+N:=N_s$. Then,
\begin{eqnarray}
    r_1r_2r_3\g_{32}&\in& O\left(\frac{\eta N_st}{\Delta^2}\g_{32}'+\frac{\log(1/\delta_1)}{\log\log(1/\delta_1)}\left(\frac{t}{\epsilon}\right)^{\frac{1}{p_1}}\frac{\eta Nt}{\Delta^2}L_1^{1-\frac{1}{p_1}}\g_{32}'\left(\frac{\eta}{\Delta^2}\right)^{\frac{1}{p_1}} \right)   \nonumber \\
    &\in& O\left(\frac{\eta N_s't}{\Delta^2}\g_{32}' \right)\qquad \left[N_s+N\frac{\log(1/\delta_1)}{\log\log(1/\delta_1)}\left(\frac{t\eta}{\epsilon\Delta^2}\right)^{\frac{1}{p_1}}L_1^{1-\frac{1}{p_1}}:=N_s' \right]   \nonumber \\
    &\in&O\left(\frac{\eta^2N_s't\log N}{\Delta^2}+\frac{\eta NN_s't\log\Lambda}{\Delta^2}+\frac{\eta N_s't}{\Delta^2}\log N\log\frac{N}{\delta'}+\frac{\eta N_s'Kt}{\Delta^2}\log\frac{1}{\delta'} \right) \label{eqn:r1r2r3g32}
\end{eqnarray}
and so,
\begin{eqnarray}
    \g&\in&r_1\g_1+r_1r_2\g_2+r_1r_2r_3\left(\g_{31}+\g_{32}\right) \nonumber \\
    &\in&O\left(\frac{\eta^2N_s't\log N}{\Delta^2}+\frac{\eta NN_s't\log^2\Lambda}{\Delta^2}\left(1+\frac{\Delta^2\Lambda}{\eta}\right)^{\frac{1}{p_1}}+\frac{\eta N_s't}{\Delta^2}\log N\log\frac{N}{\delta'}+\frac{\eta N_s'Kt}{\Delta^2}\log\frac{1}{\delta'} \right),   \label{eqn:g_div}
\end{eqnarray}
where $N_s=\eta+Z_{sum}+N$, $N_s'=N_s+N\frac{\log(1/\delta_1)}{\log\log(1/\delta_1)}\left(\frac{t\eta}{\epsilon\Delta^2}\right)^{\frac{1}{p_1}}L_1^{1-\frac{1}{p_1}}$ and $L_1=1+\frac{\eta+Z_{sum}}{N}+\frac{\Delta^2\Lambda^2}{\eta}$. Also, $\delta'\in O\left(\left(\frac{\epsilon}{t}\right)^{1+\frac{1}{p_1}}\left(\frac{\eta N}{\Delta^2}\right)^{-1}L_1^{\frac{1}{p_1}-1}\right)$.
The result of Theorem \ref{thm:DC}, then follows by making these substitutions and using $\widetilde{O}$ notation to drop sub-dominant logarithmic factors from asymptotic bound on the gate complexity in Equation \ref{eqn:g_div}. 
}
\end{proof}

\subsection{Algorithm-II : Qubitization without divide and conquer}
\label{subsec:totalQubit}

In this section, we describe another algorithm for simulating $e^{-i\hat{H}_{PF}t}$ where we apply qubitization on the entire exponential i.e. we do not divide the Hamiltonian and apply different algorithms to simulate each fragment. So now, we block encode the entire $\hat{H}_{PF}$.  

From Lemma \ref{lem:g1}-\ref{lem:g32} we know the gate complexities $\g_1', \g_2', \g_{31}', \g_{32}'$ for block encoding $\frac{H_{12}}{\lambda_{12}}$, $\frac{H_{21}}{\lambda_{21}}$, $\frac{H_{31}}{\lambda_{31}}$, $\frac{H_{32}}{\lambda_{32}}$, where $\lambda_{12}=\|H_{12}\|$, $\lambda_{21}=\|H_{21}\|$, $\lambda_{31}=\|H_{31}\|$, $\lambda_{32}=\|H_{32}\|$, respectively. We remind the readers that in this paper $\|H_i\|$ for any Hamiltonian $H_i$ is equal to the sum of the coefficients in an LCU decomposition of $H_i$, which can also be used as a bound on the $\ell_1$ norm of $H_i$. Since $\hat{H}_{PF}=H_{12}+H_{21}+H_{31}+H_{32}$, so using Theorem \ref{thm:blockEncodeDivConq} we can say that the cost of having a $(\lambda_{12}+\lambda_{21}+\lambda_{31}+\lambda{32},-,0)$-block-encoding of $\hat{H}_{PF}$ is 
\begin{eqnarray}
    \g'\in O\left(\g_2'+\g_1'+\g_{31}'+\g_{32}'\right),
\end{eqnarray}
where $\g_1', \g_2', \g_{31}', \g_{32}'$ are the gate complexities described in Lemma \ref{lem:g1}-\ref{lem:g32}. With the assumptions made in the previous section we have
\begin{eqnarray}
    \g_{1}' &\in& O\left(N\log\Lambda\right)   \nonumber \\
    \g_{2}' &\in& O\left(N\log^2\Lambda\right)    \nonumber \\
    \g_{31}' &\in& O\left(\eta+N\log^2\Lambda\right)   \nonumber \\
    \g_{32}' &\in& O\left(\eta \log N+N\log\Lambda+\log N\log\frac{N}{\delta'}+K\log\frac{1}{\delta'}\right)     \nonumber
\end{eqnarray}
and so
\begin{eqnarray}
    \g'&\in& O\left(\eta\log N+N\log^2\Lambda+\log N\log\frac{N}{\delta'}+K\log\frac{1}{\delta'}\right) \nonumber   \\
    &\in&O\left((\eta+\log N)\log N+(K+\log N)\log\frac{1}{\delta'}+N\log^2\Lambda\right). \nonumber
\end{eqnarray}
Also, from Equation \ref{eqn:Hpf_norm},
\begin{eqnarray}
    \|\hat{H}_{PF}\|&\leq& \|H_{\pi}\|+\|H_V\|+\|H_s\|+\|H_f\|  \nonumber \\
    &\in& O\left(\frac{\eta N}{\Delta^2}L_1\right)\qquad \left[\text{where }L_1=1+\frac{\eta+Z_{sum}}{N}+\frac{\Delta^2\Lambda^2}{\eta}\right] \nonumber
\end{eqnarray}
 and so we require 
$
    O\left(\frac{\eta Nt}{\Delta^2}L_1+\frac{\log(1/\epsilon)}{\log\log(1/\epsilon)}\right)
$
calls to the block encoding of $\hat{H}_{PF}$ in order to implement an $\epsilon$-precise block encoding of $e^{-i\hat{H}_{PF}t}$ \cite{2019_GSLW}. We can assume that $\delta'\in O\left(\frac{\epsilon\Delta^2}{\eta NtL_1}\right)$. Thus the number of gates required is as follows.
\begin{eqnarray}
    \g'' 
    \in& O\left(\left(\frac{\eta Nt}{\Delta^2}L_1+\frac{\log(1/\epsilon)}{\log\log(1/\epsilon)}\right)\left((\eta+\log N)\log N+(K+\log N)\log\frac{1}{\delta'}+N\log^2\Lambda\right)\right)   \label{eqn:g_totalQubit}
\end{eqnarray}
Hence we get the following theorem.
\begin{theorem}
Assuming that $\eta, K\leq N$, $a$ is constant, $1/\Delta^2c \in O(1)$ and $K, Z_{sum}\in O(\eta)$ then there exists an algorithm that simulates $e^{-i\hat{H}_{PF}t}$ with error $\epsilon$, using qubitization, with gate complexity in
\begin{eqnarray}
    \widetilde{O}\left(Nt\left(\frac{\eta}{\Delta^2}+\Lambda^2\right)\left(\eta\log\frac{1}{\epsilon}+N\log^2\Lambda\right)\right).  \nonumber 
\end{eqnarray}
    \label{thm:qub}
\end{theorem}
This shows that we can achieve similar scaling to that attainable with Trotter-Suzuki simulations.  One important difference, however, is that the scaling with the cutoff is superior for the divide and conquer approach for the case where $p_1=1$.  The scaling with respect to $\epsilon$ and $t$ is superior however for qubitization.  This shows that we expect both simulation algorithms to offer advantages in appropriate regimes.  We observe this for chemical applications in the following section.

\section{Applications}
\label{sec:apps}
As we saw in the previous section, there are different asymptotic advantages and disadvantages to simulating the Pauli-Fierz using either qubitization or the divide and conquer algorithm.  
One advantage of divide and conquer is the fact that we do not have to repeat the simulations of all the fragments - $H_{12}$, $H_{21}$, $H_{31}$ and $H_{32}$, number of times proportional to the $\ell_1$ norm of the complete $\hat{H}_{PF}$. Rather, each fragment is repeated number of times proportional to a smaller Hamiltonian with a lower $\ell_1$ norm (refer to Figure \ref{fig:divNconq} for convenience) - Equations \ref{eqn:r1g1+r1r2g2}-\ref{eqn:r1r2r3g32}. Due to Trotter splitting, we do have to repeat number of times more than the $\ell_1$ norm, but by a clever choice of grouping such that the commutators are less and by an appropriate selection of the order of splitting, it is possible to reduce the gate complexities. Further, since we have the liberty to apply different simulation algorithms, it has been possible to simulate part of the Hamiltonian (i.e. $H_{21}$) by Trotterization, which has a much lower gate complexity. 

However, directly comparing these costs are obfuscated by the high number of different variables and terms. In this section, we will compare the relative costs of these algorithms compared to some model system of interest, while scaling a single important system variable, in order to gain an intuition for which algorithm one would choose depending on explicit physical regimes of select systems of interest. The main two regimes we want to explore here are small atomic systems with many degrees of freedom on the electromagnetic links, and large extended material systems in the thermodynamic limit with many electrons. Since we only have expressions for asymptotic costs of these algorithms, we will fix a starting instance of the problem with set values, and then take the ratio of the algorithm with itself, as a single variable is changed. In this way, the missing constant factors become irrelevant, and we can fairly compare the two algorithms on the same footing, at the expense that the ``Cost Ratio" is dependent on the initial problem instance, and has no clear meaning in terms of actual gate complexity. 

\subsection{Atomic/Molecular Regime}
First, we will compare costs for the regime of a small number of electrons in a single atom system, to investigate regimes of applications such as spontaneous emission of photons into the field, as well as photoionization of electrons, also known as the photoelectric effect. For example on the latter case, state of the art attosecond laser pulse experiments have attempted to probe electronic dynamics after photoexcitation~\cite{2016_RLN}. In many theoretical models, and previous experimental limitations this excitation is typically treated as instantaneous, but in the attosecond timescale, complicated electron-photon dynamics occurs that is still poorly understood. Here theoretical predictions do not match experimental values. Specifically, in the Neon atom, an experiment concluded that there was a (21 $\pm$ 5) attosecond delay between photoemmission of the 2$p$ orbital with respect to the 2$s$ orbital from the same $\sim100$ eV photon source~\cite{2010_SFKetal}. The origin of this effect is still not fully understood, and various different explanations have been explored by theoretical investigations \cite{2010_KI, 2010_BM, 2011_MLPetal, 2011_NPFetal, 2012_PFNB, 2014_FZNetal, 2018_OM, 2019_VDL, 2021_OM}. In the bigger picture, we can see a possible benefit of full quantum simulation on fault-tolerant quantum computers to settle these mysteries in ways that cannot be done without computing correlated electrons and quantum EM fields, especially for even more complicated molecular and material systems.

 Using the neon attosecond photoemission experiment as an example reference system, we can compare the gate complexity cost ratio of the different Hamiltonian simulation algorithms with respect to one variable at a time. This allows us to compare the asymptotic gate complexity of the algorithms indirectly on the same footing without having to worry about the constant factors that were dropped for the ease of analysis. The asymptotic gate complexity for the divide and conquer algorithm (DC), and the qubitization algorithm were reported previously in Eq.~\eqref{eqn:g_div} and Eq.~\eqref{eqn:g_totalQubit} respectively. To create a reference instance inspired by the attosecond experiment on neon above, we refer to the following unless otherwise noted: $\eta = Z_{sum} = 10$, the simulation box size is $\Omega^{1/3} = 30$ (Bohr) which is roughly ten times the atomic radius of neon, $N = 10^6$ lattice sites, $\Lambda = 100$, and the simulation time is $t = 83$ where 83 $\approx$ 2000 attoseconds. Additionally the error $\epsilon = 10^{-3}$ in all cases.

First, we examine the cost ratio, as a function of $N$, of the simulation algorithms with respect to the reference neon calculation at $N=10^2$. This is shown in Fig. \ref{fig:neonNg}. We have also varied the order of the outermost Trotter splitting (Figure \ref{fig:divNconq}) i.e. variable $p_1$ in Eq. \ref{eqn:g_div}, and plotted the cost ratio as a function of $N$. We see qubitization scales better for all $N$ up to $10^{15}$. However, by choosing a higher value of $p_1$ we begin to approach the qubitization cost ratio scaling, emphasizing that the divide and conquer technique can be ``tuned'' to the problem instance at hand depending on the most important variable(s) of interest. Again, note that the meaning of the ``cost ratio'' on the y-axis is ambiguous for actual gate costs, but is a useful tool for comparing which algorithms scale better when choosing a variable and picking a specific problem instance.

\begin{figure}
    \centering
    \includegraphics[scale=1.15]{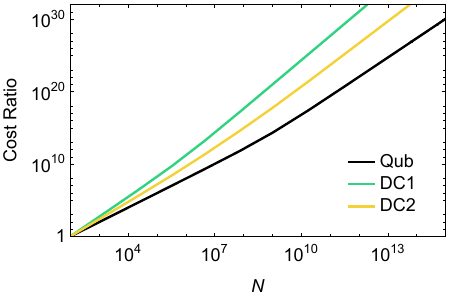}
    \caption{Cost ratio of both the qubitization simulation algorithm (Qub) and the divide and conquer algorithm when the order of the outermost Trotter splitting, i.e. $p_1=1$ (DC1) and $p_1=2$ (DC2), with respect to the reference single neon atom system as a function of $N$ grid points. For the reference we have $N=100$.}
    \label{fig:neonNg}
\end{figure}

Another variable of interest is $\Lambda$, the cutoff on the value for the electric field link space. To compare the scalings in terms of $\Lambda$, we maintain the same neon atom reference calculation, but set the minimum cutoff to $\Lambda=2$, and compare up to $\Lambda = 10^{10}$. 
\begin{figure}
    \centering
    \begin{subfigure}{0.48\textwidth}
        \centering
        \includegraphics[width=\textwidth]{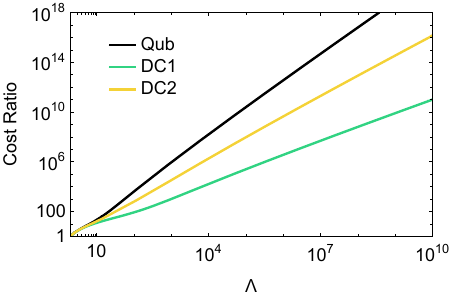}
    \end{subfigure}
    \begin{subfigure}{0.48\textwidth}
        \centering
        \includegraphics[width=\textwidth]{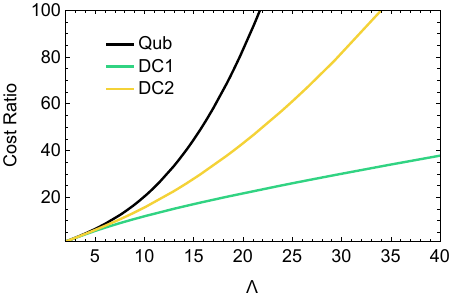}
    \end{subfigure}
    \caption{Cost ratio of both the qubitization simulation algorithm (Qub) and the divide and conquer algorithm when the order of the outermost Trotter split i.e. $p_1=1$ (DC1) and $p_1=2$ (DC2), with respect to the reference single neon atom system as a function of $\Lambda$. For the reference we have $\Lambda=2$. On the left we have a Log-Log plot where we compare up to higher values of $\Lambda$. On the right, we display the same data on a standard linear plot and focus on a smaller range for the value of $\Lambda$.}
    \label{fig:cost_ratio_lambda}
\end{figure}
The cost ratio compared to the $\Lambda=2$ instance is shown in Fig.~\ref{fig:cost_ratio_lambda}. On the left plot, we can see a sizable difference in the cost ratios on the log-log scale, where the DC algorithm outperforms qubitization. Again, selecting a higher order value for the outermost Trotter splitting i.e. variable $p_1$ in Eq. \ref{eqn:g_div} in the the DC algorithm, changes the cost ratio to approach the qubitization result. On the right hand plot, we look at small values of $\Lambda$ (on a standard linear plot) and see that for even small values of $\Lambda$, the cost ratio can be dramatically different between the two algorithms. Intuitively, this is expected because we have partitioned in such a way that the gate complexity has less dependence on $\Lambda$. For example, as mentioned before, we have grouped $H_s, H_V, H_{\pi}$ together within $H_{11}$ and ensured that the commutator error between $H_{11}$ and $H_{12}$ (Figure \ref{fig:divNconq}) is independent of $\Lambda$. Again, we emphasize that different partitionings of $\hat{H}_{PF}$ and tuning of parameters like the order of splitting in the divide and conquer algorithm, can yield different results.

\subsection{Bounds on $\Lambda$}

While we have discussed the cost of simulation with respect to the chosen cutoff, $\Lambda$, with respect to each electric field link, we now want to discuss how $\Lambda$ scales for certain regimes of applied problems. In order to quantify this, we need a more formal method to discuss the quality of a given choice of $\Lambda$. A succinct way to quantify this was presented in \cite{2022_TAMetal}, with a quantity denoted as ``leakage''. Specifically, leakage quantifies the probability that for some initial state $\lambda$ on the links between the bounds $\pm \Lambda_0$, the state grows beyond $\Lambda$ at time $t$, defined as,
\begin{equation}
    \left\|\left( 1 - \Pi_{[-\Lambda,\Lambda]} \right)e^{-itH}\Pi_{[-\Lambda_0,\Lambda_0]}\right\|
\end{equation}
where $H$ is the Hamiltonian, and the projectors $\Pi$ are defined as
\begin{equation}
    \Pi_{[-\Lambda,\Lambda]} = \sum_{|\lambda| \leq \Lambda} \Pi_{\lambda}
\end{equation}
and 
\begin{equation}
    \Pi_{\lambda} = |\lambda \rangle \langle \lambda |
\end{equation}
for a single link site. Using this definition, the long-time leakage bound was defined in Theorem 3 of \cite{2022_TAMetal} which is a quantity that quantifies a bound on $\Lambda$ given a specific time $t$, using the $t=0$ starting bounds, $\Lambda_0$. The long-time leakage bound can be computed as 
\begin{equation}
    \Lambda(t) = \Lambda_0 + \Big\lceil \frac{1}{\delta-1}\left((\Lambda_0^{1-r}+2\chi |t|(1-r)(\delta-1))^{\frac{1}{1-r}}-\Lambda_0\right) \Big\rceil(\delta-1)
\end{equation}
where $r=0$ for lattice gauge theories, $\delta$ is an integer where $\delta > 0$, and $\chi$ is constant dependent on the definition of the Hamiltonian. Specifically, we can choose $\chi$ to upper bound the spectral norm of the Hamiltonian terms that modify the value on the electric link spaces. Using similar notation in \cite{2022_TAMetal}, we denote this Hamiltonian as $H_W^{\ell}$, acting on a single choice of link $\ell$:
\begin{equation}
    H_W^{\ell} = \sum_j^{\eta} \left( \frac{2i}{c} \nabla_j \otimes A_{\ell} + \frac{1}{c^2} \id \otimes A_{\ell}^2\right) + \sum_{\mu\neq\nu=1}^3 \left(U_{\ell,\mu}U_{\ell+1,\nu},U_{\ell+1,\mu}^{\dagger}U_{\ell,\nu}^{\dagger} + h.c.\right)
\end{equation}
Therefore upper bound of the spectral norm of $H_W^{\ell}$ is
\begin{equation}
    || H_W^{\ell} || \leq \frac{4 \pi \eta \ln(2)}{\Delta^2 c} + \frac{4 \pi^2}{\Delta^2 c} + 6
\end{equation}
using (Fact 5.1). By setting $\chi$ equal to this value, this implies that
\begin{equation}
    \chi = || H_W^{\ell} || \in O \left(\frac{\eta N^{2/3}}{\Omega^{2/3}} \right)
\end{equation}
Therefore,
\begin{equation}
    \Lambda(t) \in O \left(\frac{\eta N^{2/3} t}{\Omega^{2/3}} \right)
\end{equation}
Note that this bound increases linearly with time.  This is potentially problematic for long-time evolutions as it in effect causes the time-dependence of the simulation to scale polynomially with $t$.

\subsection{Heuristic $\Lambda$ estimate for typical light-matter interaction energies}

 While the leakage bound provides a formal guarantee on the $\Lambda$ bounds in the worst case scenario, we do not expect it to be tight in most physically reasonable situations.  In particular, a major assumption made by the above analysis is that the input state has maximally bad scaling and in turn leads to linear scaling of the error bound with time.  In practice, however, if we are interested in low-energy physics of a system then these worst case scenarios are unlikely to occur.  Here we provide a proposed way to address this by giving physically informed estimates to heuristically bound $\Lambda$ for the systems of interest in non-relativistic light-matter interactions. First we can assume that for a single particle like excitation between an incoming photon, will be upper bounded by the deepest potential well on the heaviest atom. Specifically, for a hydrogenic atom, the single electron energy levels in Hartree units are
\begin{equation}
    E_n = -\frac{Z^2}{2n^2}
\end{equation}
where $n$ is the principle quantum number, the zero-point energy is set at $n = \infty$, or when the electron is unbound, and $Z$ is the charge of the nucleus modeled as a point charge. Therefore the highest bound state energy needed naturally corresponds to the deepest (1s) orbital where $n = 1$. Therefore, we will fix $n=1$, and take the absolute value of the energy expression with the maximum $Z$ value in the system to correspond to the highest effective $\Lambda$ as
\begin{equation}
    \tilde{\Lambda}^{1e} = \frac{Z^2_{\text{max}}}{2}.
\end{equation}
for a single electron excitation. Now, for a system containing $\eta$ electrons, we can assume that they are all non-interacting and occupy the lowest energy energy state of the $Z_{\text{max}}$ hydrogenic ion, and the effective cutoff is then upper bounded as 
\begin{equation}
    \tilde{\Lambda} \leq \frac{\eta Z^2_{\text{max}}}{2}.
\end{equation}
This upper bound roughly corresponds to assuming that all $\eta$ electrons are occupying the 1s orbital of the deepest potential well, and $\eta$ individual photons interact and excite the system into the continuum.

\section{Conclusion}
\label{sec:conclude}

We derive the first quantized representation of the many body Pauli-Fierz Hamiltonian (also referred to as the non-relativistic quantum electrodynamics Hamiltonian) and subsequently design two algorithms to simulate its dynamics. First, we develop a divide and conquer algorithm that partitions the Hamiltonian terms and simulates each using different simulation algorithms like Trotterization and qubitization. Next, we derive the complexity of simulating this Hamiltonian using complete qubitization. Additionally, we discuss some potential applications, such as simulating the attosecond dynamics of photoionization in atoms and molecules. We also discuss the relative merits of using these two algorithms for different parameter regimes. We observe that depending on the partitioning scheme, the divide and conquer approach has the potential to yield smaller gate costs. For example, one particular parameter of interest is the electric cutoff $\Lambda$. Roughly, the complexity of qubitization varies quadratically with $\Lambda$, while divide and conquer shows a sub-quadratic dependence. While both of these algorithms scale quadratically with the lattice size $N$, it appears that the cost of qubitization scales more favorably with this parameter overall. Another interesting observation is the fact that as we increase the order of the Trotter splitting in the divide and conquer method, the scaling approaches that of qubitization. Finally, we also develop efficient techniques to implement group of multi-controlled-X gates, that shaves off log factors in the asymptotic complexity and thus can yield a significant improvement in the cost of implementing the SELECT operations. 


Overall, we have found that the quantum simulation of the first quantized many body Pauli-Fierz Hamiltonian is efficient, but there are many avenues for future work. First, we expect that there are many opportunities for optimizing this simulation in general, and especially when tailored to specific applications of interest. Second, since the Pauli-Fierz Hamiltonian captures so much of the phenomena in the low energy regime of molecular and condensed matter interacting with light, we expect that many new applications of this model can be employed that were not discussed here. In fact, identifying key applications of this model, and computing exact gate counts for said applications is an exciting avenue to probe for evidence of practical advantages of this simulation routine on future fault-tolerant quantum computers. 

Looking forward, there are a number of ways in which this work can be built on.  An open question involves whether these ideas can be generalized to a broader class of field theories including non-Abelian gauge theories.  Further, while the Pauli-Fierz model is appropriate for a strong electromagnetic field coupled with the system, it is not capable of capturing all of quantum electrodynamics because of its inability to generate electromagnetic fields directly through particle motion and instead relies on the adhoc introduction of the Coulomb potential.  Providing ways to go beyond the limitations of this model would be an important step towards completing our understanding of how to simulate quantum electrodynamics on a quantum computer.


\section*{Acknowledgements}
The authors thank the anonymous reviewers for helpful comments, that have helped improve the manuscript.
T.F.S. is a Quantum Postdoctoral Fellow at the Simons Institute for the Theory of Computing, supported by the U.S. Department of Energy, Office of Science, National Quantum Information Science Research Centers, Quantum Systems Accelerator.  N.W. and P.M. acknowledge funding from 
 ``Embedding QC into Many-body Frameworks for Strongly Correlated Molecular and Materials Systems'' project, which is funded by the U.S. Department of Energy, Office of Science, Office of Basic Energy Sciences, the Division of Chemical Sciences, Geosciences, and Biosciences. The Pacific Northwest National Laboratory is operated by Battelle for the U.S. Department of Energy under Contract DE-AC05-76RL01830.  P.M. and N.W. further acknowledge support from Google inc.

\newcommand{\etalchar}[1]{$^{#1}$}

\appendix

\section{Derivation of the first-quantized Pauli-Fierz Hamiltonian}\label{app:ham_deriv}
In this section we derive the first-quantized Pauli-Fierz Hamiltonian.
The full (rigged) Hilbert space of the Pauli-Fierz model, $\mathscr{H}_{\text{PF}}$, in Euclidean 3-dimensional space, describing spin-1/2 electrons as fermions and the bosonic gauge field has the form
\begin{equation}
	\mathscr{H}_{\text{PF}} = \mathscr{H}_p \otimes \mathscr{H}_f, \nonumber
\end{equation}
where $\mathcal{H}_p$ is the Hilbert space of particles and $\mathcal{H}_f$ is the Hilbert space of the electromagnetic (EM) field. The  Hilbert space that describes the  particles is
\begin{equation}
    \mathscr{H}_p = P_{\text{a}}\left(\bigotimes^{\eta} L^2(\mathbbm{R}^3, \mathbbm{C}^2) \right), \nonumber
\end{equation}
where $P_{\text{a}}$ is the projection onto the anti-symmetric subspace of the $\eta$ particle system. The Hilbert space for the EM field $\mathscr{H}_f$ is then
\begin{equation}
    \mathscr{H}_f = L^2 (\mathbbm{R}^3 \times \{-\infty, \infty\}), \nonumber
    \end{equation}
where the spectrum of the field is unbounded. Naturally, for a finite simulation the maximum allowed values on the EM field need to be related to a cutoff $\Lambda$. This needs to be quantitatively estimated for the energy scales in the problem of interest, and is discussed in Section \ref{sec:apps}. The general spin-$1/2$ Pauli-Fierz Hamiltonian for $\eta$ particles is the following.
\begin{equation}\label{eq:pfham}
    \hat{H} = \sum_j^{\eta} \left[\boldsymbol{\sigma}_j \cdot \left(\boldsymbol{p}_j - \frac{e}{c} \boldsymbol{A}(\boldsymbol{x})\right) \right]^2 + \hat{H}_{f} + \hat{H}_{V}.
\end{equation}
This follows the form of Equation 20.2 in \cite{2004_S}, where bolded notation corresponds to a vector. $\boldsymbol{\sigma}_j$ is the vector of Pauli matrices $\{\sigma_1, \sigma_2, \sigma_3\}$ acting on the spin-1/2 degree of freedom for particle $j$, $\boldsymbol{p}_j$ is the 3-vector of momentum $\{p_x, p_y, p_z\}$ for the $j$th particle in 3D space, $e$ is the electric charge constant, $c$ is the speed of light, $\boldsymbol{A}(\boldsymbol{x})$ is the magnetic vector potential $\{A_x(\boldsymbol{x}), A_y(\boldsymbol{x}), A_z(\boldsymbol{x}) \}$ where $\boldsymbol{x} \in \mathbbm{R}^3$ is the position space coordinate. This Hamiltonian is represented in the Coulomb gauge, $\nabla \cdot \boldsymbol{A} = 0$, meaning that the divergence of the magnetic vector potential is chosen to be $0$.

The $\hat{H}_f$ term in Equation \eqref{eq:pfham} is the free photon space Hamiltonian defined as the following
\begin{equation}
\hat{H}_f = \frac{1}{2} \int d^3 x \,\, \boldsymbol{E}(\boldsymbol{x})^2 + \boldsymbol{B}(\boldsymbol{x})^2.
\end{equation}
Where $\boldsymbol{E}(\boldsymbol{x})$ is the electric field component and $\boldsymbol{B}(\boldsymbol{x})$ is the magnetic field component, defined as the following in terms of $\boldsymbol{A}(\boldsymbol{x})$.
\begin{align}
	\boldsymbol{E}(\boldsymbol{x}) &= -\frac{1}{c} \frac{\partial}{\partial t} \boldsymbol{A}(\boldsymbol{x}) \\
	\boldsymbol{B}(\boldsymbol{x}) &= \nabla \times \boldsymbol{A}(\boldsymbol{x})
\end{align}
The $\hat{H}_{V}$ term in Equation \eqref{eq:pfham} is the instantaneous two particle Coulomb repulsion interaction defined as 
\begin{equation}
\hat{H}_{V} = \sum_{i \neq j}^{\eta} \frac{e_i e_j}{2 \, || \boldsymbol{r}_i - \boldsymbol{r}_j ||_2},
\end{equation}
where $\boldsymbol{r}_j$ is the position vector of particle $j$.  This gives a continuous Hamiltonian that describes the dynamics of a fermionic system that is coupled to an external electromagnetic field.  In order for the relativistic limit to hold, we need to assume that $1/c \ll 1$.  This limit also removes any need to incorporate Ampere's law in the calculation because such corrections only contribute at higher-order in $1/c$. But it is worth noting that if the magnetic fields generated by the fermions are substantial then models such as these may not be applicable. 

\section{Non-relativistic spin term in the standard Pauli-Fierz Hamiltonian}\label{app:spin_term}


Following the derivation from \cite{2004_S}, the general spin-$1/2$ Pauli-Fierz Hamiltonian for $\eta$ particles is the following
\begin{equation}\label{app:eq:pfham}
    \hat{H} = \sum_j^{\eta} \left[\boldsymbol{\sigma}_j \cdot \left(\boldsymbol{p}_j - \frac{e}{c} \boldsymbol{A}(x)\right) \right]^2 + \hat{H}_{f} + \hat{H}_{V}
\end{equation}
where we are only focused on the first term which includes the spin-1/2 particles coupled to the field.  With this definition, we can then expand the first term in Equation~\eqref{eq:pfham} to isolate the spin-dependent terms.

The first term in Equation~\eqref{eq:pfham} can be expanded as follows for a single particle $j$ using the Pauli vector identity $(\boldsymbol{\sigma}\cdot \mathbf{a})(\boldsymbol{\sigma}\cdot \mathbf{b}) =  \mathbf{a}\cdot\mathbf{b} + i\boldsymbol{\sigma}\cdot \left(\mathbf{a} \times \mathbf{b}\right)$:
\begin{align}
     \left[\boldsymbol{\sigma}_j \cdot \left(\boldsymbol{p}_j - \frac{e}{c} \boldsymbol{A}(x)\right) \right]^2 &= \left[\boldsymbol{\sigma}_j \cdot \left(\boldsymbol{p}_j - \frac{e}{c} \boldsymbol{A}(x)\right) \right] \left[\boldsymbol{\sigma}_j \cdot \left(\boldsymbol{p}_j - \frac{e}{c} \boldsymbol{A}(x)\right) \right] \\ 
     &= \left(\boldsymbol{p}_j - \frac{e}{c} \boldsymbol{A}(x)\right) \cdot \left(\boldsymbol{p}_j - \frac{e}{c} \boldsymbol{A}(x)\right) + i \boldsymbol{\sigma} \cdot \left(\left(\boldsymbol{p}_j - \frac{e}{c} \boldsymbol{A}(x)\right) \times \left(\boldsymbol{p}_j - \frac{e}{c} \boldsymbol{A}(x)\right) \right)
\end{align}
The first term just reduces back to the original kinetic momentum term without spin. Therefore, using the fact that the cross products $\boldsymbol{p} \times \boldsymbol{p}$ and $\boldsymbol{A} \times \boldsymbol{A}$ vanish,
\begin{align}
    \left[\boldsymbol{\sigma}_j \cdot \left(\boldsymbol{p}_j - \frac{e}{c} \boldsymbol{A}(x)\right) \right]^2 &= \left(\boldsymbol{p}_j - \frac{e}{c} \boldsymbol{A}(x)\right)^2  + i \boldsymbol{\sigma} \cdot \left(\left(\boldsymbol{p}_j - \frac{e}{c} \boldsymbol{A}(x)\right) \times \left(\boldsymbol{p}_j - \frac{e}{c} \boldsymbol{A}(x)\right) \right) \\
    &= \left(\boldsymbol{p}_j - \frac{e}{c} \boldsymbol{A}(x)\right)^2  + i \boldsymbol{\sigma} \cdot \left( \boldsymbol{p} \times \boldsymbol{p} - \frac{e}{c}\boldsymbol{p} \times \boldsymbol{A}(x) - \frac{e}{c}\boldsymbol{A}(x) \times \boldsymbol{p} + \frac{e^2}{c^2}\boldsymbol{A}(x) \times \boldsymbol{A}(x) \right) \\
    &= \left(\boldsymbol{p}_j - \frac{e}{c} \boldsymbol{A}(x)\right)^2  + i \boldsymbol{\sigma} \cdot \left( - \frac{e}{c}\boldsymbol{p} \times \boldsymbol{A}(x) - \frac{e}{c}\boldsymbol{A}(x) \times \boldsymbol{p}  \right) \\
    &= \left(\boldsymbol{p}_j - \frac{e}{c} \boldsymbol{A}(x)\right)^2  - i\frac{e}{c} \boldsymbol{\sigma} \cdot \left( \boldsymbol{p} \times \boldsymbol{A}(x) + \boldsymbol{A}(x) \times \boldsymbol{p}  \right) 
\end{align}
Now, by substituting in $\boldsymbol{p}\rightarrow -i\nabla$, assuming that this operates on a scalar function $\psi$, and subsequently using the vector identity $\nabla \times (\psi \boldsymbol{A}) = \psi (\nabla \times \boldsymbol{A}) - (\boldsymbol{A} \times \nabla) \psi $ we get
\begin{align}
    \left[\boldsymbol{\sigma}_j \cdot \left(\boldsymbol{p}_j - \frac{e}{c} \boldsymbol{A}(x)\right) \right]^2 \psi &= \left(\nabla_j - \frac{e}{c} \boldsymbol{A}(x)\right)^2 \psi  - \frac{e}{c} \boldsymbol{\sigma} \psi \cdot \left( \nabla_j \times \boldsymbol{A}(x) + \boldsymbol{A}(x) \times \nabla_j  \right) \psi \\
    &= \left(\nabla_j - \frac{e}{c} \boldsymbol{A}(x)\right)^2 \psi  - \frac{e}{c} \boldsymbol{\sigma} \cdot \left( \nabla_j \times (\boldsymbol{A}(x) \psi) + \boldsymbol{A}(x) \times (\nabla_j \psi)  \right)\\
    &= \left(\nabla_j - \frac{e}{c} \boldsymbol{A}(x)\right)^2 \psi  - \frac{e}{c} \boldsymbol{\sigma} \cdot \left( \psi(\nabla_j \times (\boldsymbol{A}(x) ) - (\boldsymbol{A}(x) \times \nabla_j ) \psi + \boldsymbol{A}(x) \times (\nabla_j \psi)  \right) \\
    &= \left(\nabla_j - \frac{e}{c} \boldsymbol{A}(x)\right)^2 \psi  - \frac{e}{c} \boldsymbol{\sigma} \cdot \left( \nabla_j \times \boldsymbol{A}(x) \psi  \right) \\
    &= \left(\nabla_j - \frac{e}{c} \boldsymbol{A}(x)\right)^2 \psi  - \frac{e}{c} \boldsymbol{\sigma} \cdot  \boldsymbol{B}(x) \psi
\end{align}
Therefore, the Pauli-Fierz Hamiltonian including spin simplifies to
\begin{equation}
    \hat{H} = \sum_j^{\eta} \left[ \left(\boldsymbol{p}_j - \frac{e}{c} \boldsymbol{A}(x)\right)^2 - \frac{e}{c} \boldsymbol{\sigma}_j \cdot  \boldsymbol{B}(x) \right] + \hat{H}_{f} + \hat{V}_{\text{coul}}
\end{equation}
where the only spin-dependent term at the one body interaction level is the $\boldsymbol{\sigma}_j \cdot  \boldsymbol{B}(x)$ term.

\section{Divide and conquer for block encoding } 
\label{app:divConqBlock}

In this section we describe a divide-and-conquer approach for block encoding of weighted linear combination or product of Hamiltonians. Suppose we have $M$ Hamiltonians - $H_1,\ldots,H_M$, such that each has an LCU decomposition as : $H_i=\sum_{j=1}^{M_i}h_{ij}U_{ij}$ and $\lambda_i=|h_{ij}|$. 
We can have a $(\lambda_i,\log M_i,0)$-block encoding of $H_i$ using an ancilla preparation sub-routine and a unitary selection sub-routine, which we denote by $\prep_i$ and $\sel_i$ respectively.
\begin{eqnarray}
 \prep_i\ket{0}^{\log M_i}&=&\sum_{j=1}^{M_i}\sqrt{\frac{h_{ij}}{\lambda_i}}\ket{j}   \label{app:eqn:prepi} \\
 \sel_i&=&\sum_{j=1}^{M_i}\ket{j}\bra{j}\otimes U_{ij}   \label{app:eqn:seli} \\
 \bra{0}\prep_i^{\dagger}\cdot \sel_i\cdot\prep_i\ket{0}&=&\frac{H_i}{\lambda_i}    \label{app:eqn:prepiSeli}
\end{eqnarray}
Now we use these sub-routines to define the following.
\begin{eqnarray}
 \prep\ket{0}^{\log M+\sum_i\log M_i}&=&\left(\sum_{i=1}^M\sqrt{\frac{w_i\lambda_i}{\mathcal{\nconst}}}\ket{i}\right)\otimes\bigotimes_{i=1}^M\prep_i    \label{app:eqn:divPrep} \\
 \sel&=&\sum_{i=1}^M\left(\ket{i}\bra{i}\otimes\bigotimes_{k=1}^{i-1}\id\otimes\sel_i\otimes\bigotimes_{k=i+1}^M\id\right)     \label{app:eqn:divSel}
\end{eqnarray}
where $w_i>0$ and $\nconst=\sum_{i=1}^Mw_i\lambda_i$. In the following theorem we show that we can block encode a linear combination of these Hamiltonians using the above sub-routines.
\begin{theorem}
Let $H=\sum_{i=1}^Mw_iH_i$ be the sum of $M$ Hamiltonians, and each of them is expressed as sum of unitaries as : $H_i=\sum_{j=1}^{M_i}h_{ij}U_{ij}$ such that $\lambda_i=\sum_j|h_{ij}|$, $w_i>0$. Each of the summand Hamiltonian is block-encoded using the sub-routines defined in Equations \ref{app:eqn:prepi} and \ref{app:eqn:seli}. Then, we can have a $(\nconst,\lceil\log_2(M)\rceil,0)$- block encoding of $\frac{H}{\nconst}$, where $\nconst=\sum_{i=1}^Mw_i\lambda_i$, using the ancilla preparation sub-routine ($\prep$) defined in Equation \ref{app:eqn:divPrep} and the unitary selection sub-routine ($\sel$) defined in Equation \ref{app:eqn:divSel}.
\begin{enumerate}
    \item The PREP sub-routine has an implementation cost of $\mathcal{C}_{\prep}=\sum_{i=1}^M\mathcal{C}_{\prep_i}+\mathcal{C}_{w}$, where $\mathcal{C}_{\prep_i}$ is the number of gates to implement $\prep_i$, and $\mathcal{C}_w$ is the cost of preparing the state $\sum_{i=1}^M\sqrt{\frac{w_i\lambda_i}{\nconst}}\ket{i}$.

    \item The $\sel$ sub-routine can be implemented with a set of multi-controlled-X gates - \\
    $\{M_i\text{ pairs of }C^{\log_2M_i+1}X\text{ gates }:i=1,\ldots,M\}$,$M$ pairs of $C^{\log M}X$ gates and $\sum_{i=1}^MM_i$ single-controlled unitaries - $\{cU_{ij}: j=1,\ldots,M_i; i=1,\ldots,M\}$. 
\end{enumerate}
 \label{app:thm:blockEncodeDivConq}
\end{theorem}

\begin{proof}
The ancilla preparation and unitary selection sub-routines for the block encoding of $H/\nconst$ have been defined in Equations \ref{app:eqn:divPrep} and \ref{app:eqn:divSel}.
\begin{eqnarray}
 \prep\ket{0}^{\log M+\sum_i\log M_i}&=&\left(\sum_{i=1}^M\sqrt{\frac{w_i\lambda_i}{\mathcal{\nconst}}}\ket{i}\right)\otimes\bigotimes_{i=1}^M\prep_i   \nonumber\\
 &=&\left(\sum_{i=1}^M\sqrt{\frac{w_i\lambda_i}{\mathcal{\nconst}}}\ket{i}\right)\otimes\bigotimes_{i=1}^M\left(\sum_{j=1}^{M_i}\sqrt{\frac{h_{ij}}{\lambda_i}}\ket{j}\right)   \nonumber \\
 \sel&=&\sum_{i=1}^M\left(\ket{i}\bra{i}\otimes\bigotimes_{k=1}^{i-1}\id\otimes\sel_i\otimes\bigotimes_{k=i+1}^M\id\right)  \nonumber \\
 &=&\sum_{i=1}^M\left(\ket{i}\bra{i}\otimes\bigotimes_{k=1}^{i-1}\id\otimes\left(\sum_{j=1}^{M_i}\ket{j}\bra{j}\otimes U_{ij}\right)\otimes\bigotimes_{k=i+1}^M\id\right)   \nonumber
\end{eqnarray}
Thus
\begin{eqnarray}
 &&\sel\cdot\prep\ket{0}\ket{\psi}  \nonumber \\
 &=&\sum_{i=1}^M\sum_{j_i=1}^{M_i}\left(\sqrt{\frac{w_i\lambda_i}{\nconst}}\ket{i}\otimes\bigotimes_{k=1}^{i-1}\left(\sum_{j_k=1}^{M_k}\sqrt{\frac{h_{kj}}{\lambda_k}}\ket{j_k}\right)\otimes\sqrt{\frac{h_{ij_i}}{\lambda_i}}\ket{j_i}\otimes\bigotimes_{k=i+1}^{M}\left(\sum_{j_k=1}^{M_k}\sqrt{\frac{h_{kj}}{\lambda_k}}\ket{j_k}\right)\right)U_{ij_i}\ket{\psi}  \nonumber
\end{eqnarray}
and 
\begin{eqnarray}
 \bra{0}\prep^{\dagger}=\left(\prep\ket{0}\right)^{\dagger}=\left(\sum_{i=1}^M\sqrt{\frac{w_i\lambda_i}{\mathcal{\nconst}}}\bra{i}\right)\otimes\bigotimes_{i=1}^M\left(\sum_{j=1}^{M_i}\sqrt{\frac{h_{ij}}{\lambda_i}}\bra{j}\right)    \nonumber
\end{eqnarray}
and hence
\begin{eqnarray}
 &&\bra{0}\prep^{\dagger}\cdot\sel\cdot\prep\ket{0}\ket{\psi} \nonumber \\
 &=&\sum_{i=1}^M\sum_{j_i=1}^{M_i}\frac{w_i\lambda_i}{\nconst}\frac{h_{ij_i}}{\lambda_i}U_{ij_i}\ket{\psi}+\ket{\Psi^{\perp}}=\left(\frac{1}{\nconst}\sum_{i=1}^Mw_iH_i\right)\ket{\psi}+\ket{\Psi^{\perp}} \nonumber
\end{eqnarray}
Thus we have a block encoding of $\frac{H}{\nconst}$.

It is quite clear that the cost of implementing $\prep$ is as stated in the statement of the lemma. 
So now we describe the implementation of $\sel$, which can be written as follows. 
\begin{eqnarray}
 \sel:\ket{i}\ket{0,k_1}_1\ldots\ket{1,j}_i\ldots\ket{0,k_M}_M\ket{\psi}\mapsto\ket{i}\ket{0,k_1}_1\ldots\ket{1,j}_i\ldots\ket{0,k_M}_MU_{ij}\ket{\psi}  \nonumber
\end{eqnarray}
In the above we have represented each set of ancillae qubits in the $M+1$ subspaces of $\prep$ as a separate register. We allot one ancilla qubit, initialized to 0, for each $\prep_i$ register. If state of the first register containing $\log M$ qubits is $\ket{i}$ then the $i^{th}$ register corresponding to $\prep_i$ is selected by flipping this ancilla to $1$. We require $M$ (compute-uncompute) pairs of $C^{\log_2M}X$ gates and $M$ ancillae to make this selection. Now if the state of $\prep_i$-regitser is $\ket{j}$, then we select the $j^{th}$ unitary in the LCU decomposition of $H_i$ i.e. $U_{ij}$. To select unitaries of the $i^{th}$ Hamiltonian $H_i$, we require $M_i$ pairs of $C^{\log_2M_i+1}X$. Each of these flip another ancilla corresponding to each unitary in the LCU decomposition. The unitaries are implemented controlled on this ancilla. This explains the implementation cost of the $\sel$ sub-routine.
\end{proof}

\paragraph{Advantages : } Now we explain that we can have a decrease in gate complexity if we follow this divide and conquer approach, instead of block encoding $H$ as sum of $M'=\sum_{i=1}^MM_i$ unitaries. 
\begin{eqnarray}
 H=\sum_{i=1}^Mw_iH_i=\sum_{i=1}^M\sum_{j=1}^{M_i}w_ih_{ij}U_{ij}   \nonumber
\end{eqnarray}
We have a $\prep'$ sub-routine, acting on $\log_2M'$ ancillae qubits, whose states select a particular unitary in the decomposition. A superposition of these basis states with weights $w_ih_{ij}$ can be obtained by using approximately $\log_2M'=\log_2(\sum_iM_i)$ H, $2M'+3\log_2M'-7=2\sum_iM_i+3\log_2(\sum_iM_i)-7$ CNOT and $2M'-2=2\sum_iM_i-2$ rotation gates \cite{2021_APPS, 2016_NDW, 2011_PB}. In the $\sel'$ sub-routine we have $M'$ unitaries, each controlled on $\log_2M'$ qubits. Each of them, in turn can be implemented with a (compute-uncompute) pair of $C^{\log_2M'}X$ and one controlled unitary. Decomposing the multi-controlled-NOT in terms of Clifford+T \cite{2017_HLZetal, 2018_G}, we see that we require at most $M'(4\log_2M'-4)$ T, $M'(4\log_2M'-3)$ CNOT. The use of logical AND gadgets reduces the gate complexity in the uncomputation part.

Now let us use the divide-and-conquer method described in Theorem \ref{app:thm:blockEncodeDivConq}. For the $\prep$ sub-routine we require $\log_2M+\sum_i\log_2M_i=\log(M\prod_iM_i)$ H, $2(M+\sum_iM_i)+3(\log_2M+\sum_i\log_2M_i)-7(M+1)=2(M+\sum_iM_i)+3\log(M\prod_iM_i)-7(M+1)$ CNOT and $2(M+\sum_iM_i)-2(M+1)=2\sum_iM_i-2$ rotation gates. Comparing with the above estimate we see that we require same number of rotation gates, more H gates and the difference in CNOT-count is
\begin{eqnarray}
&&2(M+\sum_iM_i)+3\log(M\prod_iM_i)-7(M+1)-2\sum_iM_i-3\log(\sum_iM_i)+7   \nonumber \\
&=&2M+3\log\left(\frac{M\prod_iM_i}{\sum_iM_i}\right)-7M=3\log\left(\frac{M\prod_iM_i}{\sum_iM_i}\right)-5M
\end{eqnarray}
which can be less than 0 for certain values of $M$ and $M_i$. For the $\sel$ sub-routine we require, for each $i$, $M_i$ pairs of $C^{\log_2M_i+1}X$, $M$ pairs of $C^{\log_2M}X$. Decomposing these \cite{2017_HLZetal, 2018_G}, we require $\sum_iM_i(4\log (M_i+1)-4)+M(4\log M-4)$ T, $\sum_iM_i(4\log (M_i+1)-3)+M(4\log M-3)$ CNOT. Thus the difference in T-gate count estimate is
\begin{eqnarray}
&&\sum_iM_i(4\log (M_i+1)-4)+M(4\log M-4)-(4\log(\sum_iM_i)-4)(\sum_iM_i)  \nonumber \\
&=&4\sum_iM_i\log\left(\frac{M_i+1}{\sum_jM_j}\right)+4M\log M-4M   \nonumber
\end{eqnarray}
which is less than 0 in most cases. Similarly we can show that the difference in CNOT count estimate is
\begin{eqnarray}
 4\sum_iM_i\log\left(\frac{M_i+1}{\sum_jM_j}\right)+4M\log M-3M \nonumber
\end{eqnarray}
which is again less than 0 in most cases. We use same number of controlled unitaries in both the approaches. Thus, using the divide-and-conquer technique (Theorem \ref{app:thm:blockEncodeDivConq}) it is possible to reduce the implementation cost in terms of gate count, especially the T-gate and CNOT gate.

\paragraph{Block encoding of Hamiltonians with same ancilla preparation sub-routine : } Suppose, in Theorem \ref{app:thm:blockEncodeDivConq} each $\prep_i$ are the same, which can occur if the LCU decomposition of each $H_i$ has the same weights. We note that the unitaries in the decomposition can be different. Then, in the $\prep$ sub-routine of Equation \ref{app:eqn:divPrep}, it is sufficient to keep only one copy of $\prep_i$. 
\begin{eqnarray}
    \prep\ket{0}^{\log M+\log M_i}=\left(\sum_{i=1}^M\sqrt{\frac{w_i\lambda_i}{\nconst}}\ket{i}\right)\otimes\prep_i
\end{eqnarray}
In the special case when all $H_i$ are same, but acting on disjoint subspaces then the first $\log M$ qubits need to be in equal superposition. This has been explained in Section \ref{subsec:divConqBlock}, as it is more pertinent for our paper.

\paragraph{Block encoding of product of Hamiltonians : } Suppose we have $H_p=\prod_{i=1}^MH_i$, where each $H_i$ can be block encoded with sub-routines described in Equations \ref{app:eqn:prepi}-\ref{app:eqn:prepiSeli}. Let $\nconst'=\prod_{i=1}^M\lambda_i$. Then we can block encode $H_p/\nconst'$ using the following sub-routines.
\begin{eqnarray}
 \prep_p\ket{0}^{\sum_i\log_2M_i}&=&\bigotimes_{i=1}^M\left(\sum_{j_i=1}^{M_i}\sqrt{\frac{h_{ij_i}}{\lambda_i}}\ket{j_i}\right)  \label{app:eqn:prepProd}   \\
 &=&\sum_{j_1=1}^{M_1}\sum_{j_2=1}^{M_2}\ldots\sum_{j_M=1}^{M_M}\left(\sqrt\frac{\prod_{k=1}^Mh_{kj_k}}{\nconst'}\bigotimes_{k=1}^M\ket{j_k}\right)   \nonumber \\
 \sel_p&=&\sum_{j_1=1}^{M_1}\sum_{j_2=1}^{M_2}\ldots\sum_{j_M=1}^{M_M}\left(\bigotimes_{k=1}^M\ket{j_k}\bra{j_k}\otimes\prod_{k=1}^MU_{kj_k}\right)   \label{app:eqn:selProd}
\end{eqnarray}
It follows that
\begin{eqnarray}
&& \bra{0}\prep_p^{\dagger}\cdot\sel_p\cdot\prep_p\ket{0}\ket{\psi}   \nonumber \\
&=&\frac{1}{\nconst'}\sum_{j_1=1}^{M_1}\sum_{j_2=1}^{M_2}\ldots\sum_{j_M=1}^{M_M}\prod_{k=1}^Mh_{kj_k}U_{kj_k}\ket{\psi}+\ket{\Psi^{\perp}}=\left(\frac{1}{\nconst}\prod_{i=1}^MH_i\ket{\psi}\right)+\ket{\Psi^{\perp}}
\end{eqnarray}
and it is also easy to see that the total implementation cost is the sum of the cost of implementing the block encoding of each $H_i$. Hence we have the following theorem.
\begin{theorem}
 If $H_p=\prod_{i=1}^MH_i$ is a product of $M$ Hamiltonians, such that $H_i$ can be block encoded with the sub-routine defined in Equations \ref{app:eqn:prepi}-\ref{app:eqn:seli}. Then we can block encode $H_p/\nconst'$ with the $\prep_p$ and $\sel_p$ sub-routines defined in Equations \ref{app:eqn:prepProd}-\ref{app:eqn:selProd}. 
 \begin{enumerate}
     \item The $\prep_p$ sub-routine has an implementation cost of $\sum_{i=1}^M\mathcal{C}_{\prep_i}$, where $\mathcal{C}_{\prep_i}$ is the implementation cost of $\prep_i$.

     \item The $\sel_p$ sub-routine has an implementation cost of $\sum_{i=1}^M\mathcal{C}_{\sel_i}$, where $\mathcal{C}_{\sel_i}$ is the implementation cost of $\sel_i$.
 \end{enumerate}
 \label{app:thm:blockProdDiv}
\end{theorem}

\paragraph{Advantages : } Now let us compare with the procedure where we block encode $H_p$ by expressing it as sum of $M''=\prod_{i=1}^MM_i$ unitaries. We can have an ancilla preparation sub-routine with $\log_2M''$ ancillae and for arbitrary weights we require $\log_2M''$ H, $2M''+3\log_2M''-7$ CNOT and $2M''-2$ rotation gates for preparing the weighted superposition. Using Theorem \ref{app:thm:blockProdDiv} we require the same number of H gates but the number of CNOT gates required is at most $2\sum_iM_i+3\sum_i\log_2M_i-\sum_i7$ and the number of rotation gates required is at most $2\sum_iM_i-\sum_i2$, which is much less. The difference in the number of CNOT gates is
\begin{eqnarray}
&& 2\sum_iM_i+3\sum_i\log M_i-7M-2\prod_iM_i-3\log(\prod_iM_i)+7    \nonumber \\
&=&2\left(\sum_iM_i-\prod_iM_i\right)-7(M-1)<0   \nonumber
\end{eqnarray}
while the difference in the number of rotation gates is
\begin{eqnarray}
 2\sum_iM_i-2M-2\prod_iM_i+2=2\left(\sum_iM_i-\prod_iM_i\right)-2(M-1)<0.   \nonumber
\end{eqnarray}
Without using Theorem \ref{app:thm:blockProdDiv}, for unitary selection sub-routine we require $M''$ unitaries, each of which is controlled on $\log_2M''$ ancillae. So we require $M''$ pairs of $C^{\log_2M''}X$ gates and $M''$ controlled unitaries. Decomposing the multi-controlled-Xs, we require $M''(4\log_2M''-4)$ T and $M''(4\log_2M''-3)$ CNOT. Using Theorem \ref{app:thm:blockProdDiv} we require, for each $i$, $M_i$ pairs of $C^{\log_2M_i}X$ gates and $M_i$ controlled unitaries. Thus in total we require $\sum_iM_i(4\log_2M_i-4)$ T and $\sum_iM_i(4\log_2M_i-3)$ CNOT. The difference in T-gate-count estimate is
\begin{eqnarray}
&& 4\sum_iM_i\log M_i-4(\prod_iM_i)\log\left(\prod_iM_i\right)-4\sum_iM_i+4\prod_iM_i \nonumber  \\
&=&4\sum_iM_i\log M_i-4(\prod_iM_i)\sum_i\log M_i-4(\sum_iM_i-\prod_iM_i)   \nonumber \\
&\leq&4\sum_i\left(M_i-\prod_jM_j\right)\log M_i-4\sum_i\left(M_i-\prod_jM_j\right) \nonumber\\
&=&4\sum_i\left(M_i-\prod_jM_j\right)\log_2\frac{M_i}{2},
\end{eqnarray}
and the difference in CNOT-gate-count estimate is
\begin{eqnarray}
 4\sum_i\left(M_i-\prod_iM_i\right)\log M_i-3\left(\sum_iM_i-\prod_iM_i\right) \nonumber
\end{eqnarray}
both of which are less than 0 in most cases. Clearly we get much less gate count, using the divide-and-conquer approach (Theorem \ref{app:thm:blockProdDiv}).

\section{Synthesizing group of multi-controlled-X gates : split-and-merge}
\label{app:CX}
 
There often arises situations where we need to select and implement something. For example, in many simulation algorithms \cite{2015_BCCKS, 2017_LC, 2019_LC, 2019_GSLW, 2007_BACS, 2012_CW} we need to selectively implement all the unitaries appearing in a LCU decomposition of a Hamiltonian. Let $M$ be the number of unitaries and for simplicity, we assume that $M$ is a power of $2$. Usually we allot $\log_2 M$ ancillae, the state of which selects an unitary. Thus we require $M$ unitaries, each controlled on $\log_2 M$ qubits. Each of these multi-controlled unitary can be implemented with a (compute-uncompute) pair of $C^{\log_2 M}X$ gates and a single-controlled unitary. Such sets can also appear in other applications like in \cite{2023_MWZ, 2005_MVBS, 2022_TAS, 2021_SP, 2018_BGBetal, 2008_GLM} and so the technique we develop here can also be useful in these cases. In our case the size of this set is $M$ and the number of T gates required, following the construction in \cite{2017_HLZetal, 2018_G}, is $\mathcal{T}_1=M(4\log_2M-4)$, while the number of CNOT required is $M(4\log_2M-3)$. The use of logical AND gadgets eliminates the need to use any T gate for the uncomputation part.

 We can reduce the number of gates  by 'splitting' the control and 'merging' the resulting logic. The basic intuition is as follows. Each unitary is associated with a $\log_2M$-bit binary string, corresponding to a basis state of the $\log_2M$ control qubits. Suppose we split the control qubits into two sets, each of length $\frac{\log_2 M}{2}$ and associate each unitary with a pair of binary strings of length $\frac{\log_2 M}{2}$. For each set, we use $M^{1/2}$ number of $C^{\frac{\log_2M}{2}}X$ gates to select $M^{1/2}$ basis states by flipping $M^{1/2}$ ancillae. Then we use $M^{1/2}\cdot M^{1/2}$ number of $C^2X$ gates to select a basis state from each set and associate it to a unitary. Thus we require $M^{1/2}C^{\frac{\log_2M}{2}}X+MC^2X$ pairs of gates (compute and uncompute) and hence the number of T gates required is at most $\sqrt{M}\left(\frac{4\log_2M}{2}-4\right)+M\cdot 4$ \cite{2017_HLZetal}.  This constitutes a savings of a log factor in the complexity.  The difference in the cost from the case without splitting is
\begin{eqnarray}
&&    M(4\log_2M-8)-\sqrt{M}(4\log_2M-8)- 4M
=4\sqrt{M}(\sqrt{M}-1)(\log_2M-2)>0,   \nonumber
\end{eqnarray}
when $M\geq 4$. Similarly we can show that we require fewer CNOT gates at the price of extra ancillae. The number of extra ancillae required is at most $\sqrt{M}+\sqrt{M}=2\sqrt{M}$.
This technique can be generalized to the case where the controls are split into multiple sets, as stated in the following theorem.
\begin{theorem}
Consider the unitary $U = \sum_{j=0}^{M-1} \ketbra{j}{j} \otimes U_j$ for unitary operators $U_j$ that can be implemented controllably. We assume $M$ is a power of 2 for simplicity.
Suppose we have $\log_2M$ qubits and $M$ (compute-uncompute) pairs of $C^{\log_2M}X$ gates for selecting the $M$ basis states. Let $r_1,\ldots,r_n\geq 1$ be positive fractions such that $\sum_{i=1}^n\frac{1}{r_i}=1$ and $\frac{\log_2M}{r_i}$ are integers. Then, $U$ can be implemented with a circuit with $$\sum_{i=1}^nM^{\frac{1}{r_i}}C^{\frac{\log_2M}{r_i}}X + MC^nX$$ 
(compute-uncompute) pairs of gates, $M$ applications of controlled $U_j$ and at most $\sum_{i=1}^nM^{\frac{1}{r_i}}$ ancillae. 
\label{app:thm:CX}
\end{theorem}
\begin{proof}
We split the $\log_2M$ qubits into $n$ sets, such that the $i^{th}$ set has $\frac{\log_2M}{r_i}$ qubits. Using $M^{\frac{1}{r_i}}$ number of $C^{\frac{\log M}{r_i}}X$ gates we select from a set of $2^{\frac{\log_2M}{r_i}}=M^{\frac{1}{r_i}}$ extra ancilla qubits, each corresponding to a basis state of the qubits in this set. That is, each multi-controlled-X gate has a target on one of these extra ancilla qubits, which gets selected (i.e. state flips) if the control qubits are in a certain basis state.
We can use $\left(\prod_{i=1}^nM^{\frac{1}{r_i}}\right)=M$ number of $C^nX$ gates, such that each has one control in an ancilla qubit of each of the $n$ sets, in order to select the basis states of the $\log_2M$ qubits.
\end{proof}
A very simple illustration has been given in Figure \ref{ckt:CX}, where $M=8$. 
\begin{figure}
\centering
\begin{subfigure}[b]{0.48\textwidth}
    \Qcircuit @C=0.25em @R=0.25em{
    \lstick{q_1} & \qw &\ctrlo{1}&\ctrlo{1}&\ctrlo{1}&\ctrlo{1} &\ctrl{1}&\ctrl{1}&\ctrl{1}&\ctrl{1} &\qw \\
    \lstick{q_2} & \qw &\ctrlo{1}&\ctrlo{1}&\ctrl{1}&\ctrl{1} &\ctrlo{1}&\ctrlo{1}&\ctrl{1}&\ctrl{1} &\qw \\
    \lstick{q_3} & \qw &\ctrlo{1}&\ctrl{1}&\ctrlo{1}&\ctrl{1} &\ctrlo{1}&\ctrl{1}&\ctrlo{1}&\ctrl{1} &\qw \\
    \lstick{} &\qw{/} &\gate{U_1}&\gate{U_2}&\gate{U_3}&\gate{U_4} &\gate{U_5}&\gate{U_6}&\gate{U_7}&\gate{U_8} &\qw
    }
    \caption{}
\end{subfigure}
\hfill
\begin{subfigure}[b]{0.48\textwidth}
\centering
  \Qcircuit @C=0.25em @R=0.25em{
\lstick{q_1} & \qw &\ctrlo{1}&\ctrlo{1}&\ctrl{1}&\ctrl{1} &\qw&\qw&\qw&\qw &\qw&\qw&\qw&\qw &\qw \\
\lstick{q_2} & \qw &\ctrlo{2}&\ctrl{3}&\ctrlo{4}&\ctrl{5} &\qw&\qw&\qw&\qw &\qw&\qw&\qw&\qw &\qw \\
\lstick{q_3} & \qw &\qw&\qw&\qw&\qw &\ctrlo{1}&\ctrl{1}&\ctrlo{2}&\ctrl{2} &\ctrlo{3}&\ctrl{3}&\ctrlo{4}&\ctrl{4} &\qw \\
\lstick{a_1} & \qw &\targ&\qw&\qw&\qw &\ctrl{4}&\ctrl{4}&\qw&\qw &\qw&\qw&\qw&\qw &\qw \\
\lstick{a_2} & \qw &\qw&\targ&\qw&\qw &\qw&\qw&\ctrl{3}&\ctrl{3} &\qw&\qw&\qw&\qw &\qw \\
\lstick{a_3} & \qw &\qw&\qw&\targ&\qw &\qw&\qw&\qw&\qw &\ctrl{2}&\ctrl{2}&\qw&\qw &\qw \\
\lstick{a_4} & \qw &\qw&\qw&\qw&\targ &\qw&\qw&\qw&\qw &\qw&\qw&\ctrl{1}&\ctrl{1} &\qw \\
\lstick{} & \qw{/} &\qw&\qw&\qw&\qw &\gate{U_1}&\gate{U_2}&\gate{U_3}&\gate{U_4} &\gate{U_5}&\gate{U_6}&\gate{U_7}&\gate{U_8} &\qw 
  }
  \caption{}
\end{subfigure}
  \caption{(a) A SELECT circuit, consisting of 8 unitaries and 3 qubits. (b) An implementation of the same circuit using the split-and-merge technique (Theorem \ref{app:thm:CX}).}
  \label{ckt:CX}
 \end{figure}
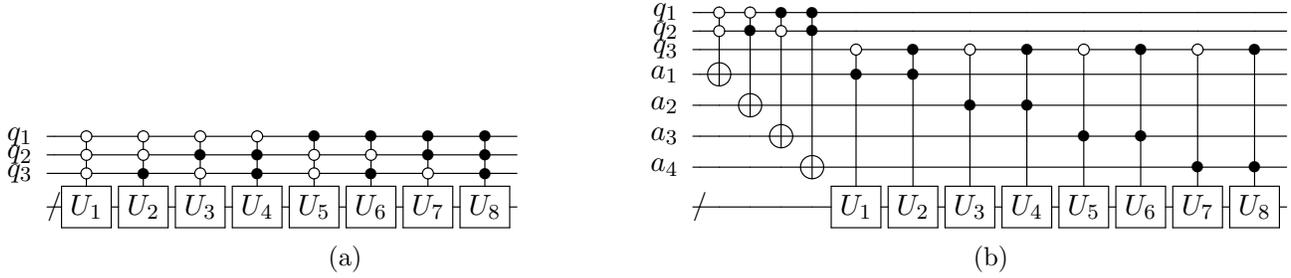 
 
\section{Required results on norm and commutator of matrices}
\label{app:norm}

In this section we give some results on norm and commutator of matrices, which we use repeatedly throughout our paper. The spectral norm of a matrix $A$, denoted by $\|A\|$, is its largest singular value.
\begin{theorem}[\cite{2005_L2}]
 Let $A\in\real^{m\times n}$ have a singular value decomposition $U_A\Sigma_AV_A^T$ and let $B\in\real^{p\times q}$ have a singular value decomposition $U_B\Sigma_BV_B^T$. Then
 $$
 (U_A\otimes U_B)(\Sigma_A\otimes\Sigma_B)(V_A^T\otimes V_B^T)
 $$
 yields a singular value decomposition of $A\otimes B$ (after a simple reordering of the diagonal elements of $\Sigma_A\otimes\Sigma_B$ and the corresponding right and left singular vectors).
 \label{thm:2005_L2}
\end{theorem}
Thus we can say that $\|A\otimes B\|= \|A\|\|B\|$. Also we know that it is an operator norm, and thus satisfies the scaling property $\|aA\|=|a|\|A\|$, the submultiplicative property $\|AB\|\leq\|A\|\|B\|$, and the triangle inequality $\|A+B\|\leq\|A\|+\|B\|$. If $A$ is unitary then $\|A\|=1$.

Let us define the adjoint operator $\ad_x:y\rightarrow [x,y]$. 

\begin{lemma}
 Let $X_j=\sum_{i_j=1}^{m_j}A_{i_j}^{(j)}$, for $j=1,\ldots,p$, where $A_{i_j}^{(j)}$ are elements from the same ring. Then,
 $$
 \ad_{X_p}\ad_{X_{p-1}}\ldots\ad_{X_3}\ad_{X_2}X_1=\sum_{i_p=1}^{m_p}\sum_{i_{p-1}=1}^{m_{p-1}}\cdots\sum_{i_2=1}^{m_2}\sum_{i_1=1}^{m_1}\ad_{A_{i_p}^{(p)}}\ad_{A_{i_{p-1}}^{(p-1)}}\ldots\ad_{A_{i_3}^{(3)}}\ad_{A_{i_2}^{(2)}}A_{i_1}^{(1)} .
 $$
 \label{lem:comSum}
\end{lemma}
\begin{proof}
 We prove the lemma by induction. First we consider the following base case.
 \begin{eqnarray}
  \ad_{X_2}X_1&=&\left[\sum_{i_2=1}^{m_2}A_{i_2}^{(2)}, \sum_{i_1=1}^{m_1}A_{i_1}^{(1)}\right]  \nonumber \\
  &=&\left(\sum_{i_2=1}^{m_2}A_{i_2}^{(2)}\right)\left(\sum_{i_1=1}^{m_1}A_{i_1}^{(1)}\right)-\left(\sum_{i_1=1}^{m_1}A_{i_1}^{(1)}\right)\left(\sum_{i_2=1}^{m_2}A_{i_2}^{(2)}\right)  \nonumber \\
  &=&\sum_{i_2=1}^{m_2}\sum_{i_1=1}^{m_1}\left[A_{i_2}^{(2)},A_{i_1}^{(1)}\right]=\sum_{i_2=1}^{m_2}\sum_{i_1=1}^{m_1}\ad_{A_{i_2}^{(2)}}A_{i_1}^{(1)}  \label{eqn:comSum0}
 \end{eqnarray}
Assume the result holds for the nested commutators between $X_1,\ldots,X_{p-1}$. That is,
\begin{eqnarray}
 \ad_{X_{p-1}}\ldots\ad_{X_3}\ad_{X_2}X_1&=&\sum_{i_{p-1}=1}^{m_{p-1}}\cdots\sum_{i_2=1}^{m_2}\sum_{i_1=1}^{m_1}\ad_{A_{i_{p-1}}^{(p-1)}}\ldots\ad_{A_{i_3}^{(3)}}\ad_{A_{i_2}^{(2)}}A_{i_1}^{(1)} \nonumber \\
&=&\sum_{i_{p-1}=1}^{m_{p-1}}\cdots\sum_{i_2=1}^{m_2}\sum_{i_1=1}^{m_1}\left[A_{i_{p-1}}^{(p-1)},\left[\cdots \left[A_{i_2}^{(2)},A_{i_1}^{(i_1)}\right]\cdots \right]\right]\label{eqn:comSum1}
\end{eqnarray}
Then, using the above equation we have 
\begin{eqnarray}
\ad_{X_p} \left(\ad_{X_{p-1}}\ldots\ad_{X_3}\ad_{X_2}X_1\right)&=&\left[\sum_{i_p=1}^{m_p}A_{i_p}^{(p)},\sum_{i_{p-1}=1}^{m_{p-1}}\cdots\sum_{i_2=1}^{m_2}\sum_{i_1=1}^{m_1}\left[A_{i_{p-1}}^{(p-1)},\left[\cdots \left[A_{i_2}^{(2)},A_{i_1}^{(i_1)}\right]\cdots \right]\right] \right]   \nonumber \\
&=&\sum_{i_p=1}^{m_p}\cdots\sum_{i_2=1}^{m_2}\sum_{i_1=1}^{m_1}\left[A_{i_p}^{(p)},\left[\cdots \left[A_{i_2}^{(2)},A_{i_1}^{(i_1)}\right]\cdots \right]\right] \qquad [\text{Equation }\ref{eqn:comSum0}]  \nonumber
\end{eqnarray}
and thus the lemma is proved.
\end{proof}

\begin{lemma}
 $$
   \left[\otimes_{i=1}^nA_i,\otimes_{i=1}^nB_i\right]=\sum_{k=1}^n\left(\bigotimes_{i=1}^{k-1}B_iA_i\otimes [A_k,B_k]\otimes\bigotimes_{i=k+1}^nA_iB_i\right)
 $$
 \label{lem:comTensor}
\end{lemma}
\begin{proof}
 We prove this by induction. 
 
 \emph{Base case : } Let $n=2$.
 \begin{eqnarray}
  [A_1\otimes A_2, B_1\otimes B_2]&=&(A_1\otimes A_2)(B_1\otimes B_2)-(B_1\otimes B_2)(A_1\otimes A_2)  \nonumber \\
  &=&(A_1B_1\otimes A_2B_2)-(B_1A_1\otimes B_2A_2) \nonumber \\
  &=&A_1B_1\otimes A_2B_2-B_1A_1\otimes A_2B_2+B_1A_1\otimes A_2B_2-B_1A_1\otimes B_2A_2  \nonumber \\
  &=&(A_1B_1-B_1A_1)\otimes A_2B_2+B_1A_1\otimes (A_2B_2-B_2A_2)  \nonumber   \\
  &=&[A_1,B_1]\otimes A_2B_2+B_1A_1\otimes [A_2,B_2]  \nonumber
 \end{eqnarray}
Now assume the given equality holds for $n=m$. That is, 
\begin{eqnarray}
 \left[\otimes_{i=1}^mA_i,\otimes_{i=1}^mB_i\right]=\sum_{k=1}^m\left(\bigotimes_{i=1}^{k-1}B_iA_i\otimes [A_k,B_k]\otimes\bigotimes_{i=k+1}^mA_iB_i\right) 
\end{eqnarray}
We show it holds for $n=m+1$ and hence the Lemma is proved.
\begin{eqnarray}
&& \left[\otimes_{i=1}^{m+1}A_i,\otimes_{i=1}^{m+1}B_i\right]=\otimes_{i=1}^{m+1}A_iB_i-\otimes_{i=1}^{m+1}B_iA_i \nonumber \\
 &=&A_1B_1\otimes\bigotimes_{i=2}^{m+1}A_iB_i-B_1A_1\otimes\bigotimes_{i=2}^{m+1}A_iB_i+B_1A_1\otimes\bigotimes_{i=2}^{m+1}A_iB_i-B_1A_1\otimes\bigotimes_{i=2}^{m+1}B_iA_i  \nonumber \\
 &=&[A_1,B_1-B_1A_1]\otimes\bigotimes_{i=2}^{m+1}A_iB_i+B_1A_1\otimes \left(\otimes_{i=2}^{m+1}A_iB_i-\otimes_{i=2}^{m+1}B_iA_i\right) \nonumber \\
 &=&[A_1,B_1]\otimes\bigotimes_{i=2}^{m+1}A_iB_i+B_1A_1\otimes\left(\left(\sum_{k=2}^{m+1}\bigotimes_{i=2}^{k-1}B_iA_i\otimes [A_k,B_k]\otimes\bigotimes_{i=k+1}^{m+1}A_iB_i\right)\right)   \nonumber   \\
&=& \sum_{k=1}^{m+1}\left(\bigotimes_{i=1}^{k-1}B_iA_i\otimes [A_k,B_k]\otimes\bigotimes_{i=k+1}^{m+1}A_iB_i\right)  \nonumber
\end{eqnarray}
\end{proof}

\begin{fact}
 For $p+1$ matrices $A_1,A_2,\ldots, A_{p+1}$ we have
 \begin{eqnarray}
  \ad_{A_{p+1}}\ad_{A_p}\ldots\ad_{A_2}A_1\leq 2^p\|A_{p+1}\|\|A_p\|\ldots\|A_2\|\|A_1\| \nonumber
 \end{eqnarray}
\label{fact:nestComNorm}
\end{fact}

\section{LCU decomposition of operators}
\label{app:lcu}

Efficient decomposition of operators as linear combination of unitaries is an important step in many simulation algorithms like \cite{2015_BCCKS, 2017_LC, 2019_LC, 2019_GSLW}. For a Hamiltonian $H$, if $H=\sum_jh_jU_j$ is a decomposition into unitaries $U_j$, then we denote $\sum_j|h_j|$ as the $\ell_1$ norm of the decomposition, which serves as an upper bound on the spectral norm of $H$. This factor determines the complexity of LCU based simulation algorithms. The number of unitaries in the decomposition determine the gate and ancillae complexity. 

We begin this section by describing some general results that pertain to the decomposition of diagonal matrices over the complex field, as sum of unitaries - identity and signature matrices. A matrix whose diagonal elements are 0 or 1, is referred to as diagonal binary matrix and a matrix whose diagonal elements are 1 and -1, is referred to as a signature matrix (unitary). Existing simulation methods do employ such decompositions, but we use a slightly more formal presentation for ease of reference. For integer diagonal matrices we give decomposition with exponentially less number of unitaries. 

As an example for how the idea of the signature matrix decomposition works let us first consider the matrix $M_{01}$ which is a diagonal binary matrix and without loss of generality, assume that the basis is chosen such that the zero entries are sorted before the 1 entries.
\begin{equation}
 M_{01}=\diag(0,\ldots,0,1,\ldots,1)=\frac{1}{2}\left(\diag(1,\ldots,1,1,\ldots,1)-\diag(1,\ldots,1,-1,\ldots,-1)\right)
\label{eq:binary_diag_mat}
\end{equation}
This shows that by subtracting an appropriate pattern of positive and negative numbers any binary diagonal matrix can be formed by a difference between two signature matrices. In general, however, we will wish to deal with diagonal matrices that are non-binary.

\begin{lemma}
Let $M$ is a $N\times N$ diagonal matrix, such that there are $N'$ distinct diagonal elements - $m_1< m_2<\ldots< m_{N'}=m_{max}$. (The subscripts do not indicate the position of the element along the diagonal.) Then $M$ can be written as $M=c_0\id+\sum_{i=1}^{N'-1}c_iD_i$, where $D_i$ are signature matrices and $\sum_{i=0}^{N'-1}|c_i|=|m_1|+\sum_{i=2}^{N'}|m_i-m_{i-1}|$. Specifically, if every $m_i\geq 0$ then the sum is $m_{max}$.
 \label{lem:M}
\end{lemma}
\begin{proof}
 Let $B_{m_i}$ is the diagonal binary matrix, derived from $M$, such that $B_{m_i}[j,j]=1$ if $M[j,j]\geq m_i$, else it is $0$.
\begin{eqnarray}
 M&=&m_1B_{m_1}+(m_2-m_1)B_{m_2-m_1}+(m_3-m_2)B_{m_3-m_2}+
\cdots +(m_{N'}-m_{N'-1})B_{m_{N'}-m_{N'-1}}   \nonumber
\end{eqnarray}
Each of the diagonal binary matrix can be decomposed as sum of identity and a signature matrix, as shown before in Equation \ref{eq:binary_diag_mat}. In this way we decompose $M$ as sum of $\id$ and $N'-1$ signature matrices. 
\begin{eqnarray}
&& |c_0|+\sum_{i=1}^{N'-1}|c_i|\nonumber \\
&=&\left(|m_1|+\frac{1}{2}\sum_{i=2}^{N'}|m_i-m_{i-1}| \right)+\frac{1}{2}\sum_{i=2}^{N'}|m_i-m_{i-1}|   \nonumber \\
&=&|m_1|+\sum_{i=2}^{N'}|m_i-m_{i-1}|  \nonumber \\
&=&m_{N'}\qquad\text{if each } m_i>0  \nonumber
\end{eqnarray}
\end{proof}
\begin{lemma}
 Let $M_I=\diag(m_1,\ldots,m_N)$ is a $N\times N$ matrix, whose each entry is a non-negative integer and $m_{max}=\max_im_i$. Then $M=c_0\id+\sum_{i=1}^{N'}c_iD_i$, where $D_i$ are signature matrices and $N'\leq\lceil\log_2(m_{\max}+1)\rceil=\zeta$. Also, $\sum_{i=0}^{N'}|c_i|\leq 2^{\zeta}-1$.
 \label{lem:MI}
\end{lemma}
\begin{proof}
 The number of qubits required to implement $m_{max}$ is $\lceil\log_2(m_{max}+1)\rceil=\zeta$. We represent each diagonal integer in the binary representation. Thus the $i^{th}$ element is as follows.
 \begin{eqnarray}
  m_i=b_0^{(i)}2^0+b_1^{(i)}2^1+b_2^{(i)}2^2+\cdots+b_{\zeta-1}^{(i)}2^{\zeta-1}  \nonumber
 \end{eqnarray}
Thus $M_I$ can be written as follows.
\begin{eqnarray}
 M_I=2^0\diag(b_0^{(1)},b_0^{(2)},\ldots,b_0^{(N)})+2^1\diag(b_1^{(1)},b_1^{(2)},\ldots,b_1^{(N)})+\cdots+2^{\zeta-1}\diag(b_{\zeta-1}^{(1)},b_{\zeta-1}^{(2)},\ldots,b_{\zeta-1}^{(N)})    \nonumber
\end{eqnarray}
We can use Equation \ref{eq:binary_diag_mat} to decompose each diagonal binary matrix as a sum of $\id$ and one signature matrix. This proves the first decomposition part of the lemma. For the second part, we have the following.
\begin{eqnarray}
 \sum_i|c_i|=\frac{1}{2}\sum_{i=0}^{\zeta-1}2^i+\frac{1}{2}\sum_{i=0}^{\zeta-1}2^i=2^{\zeta}-1 \nonumber
\end{eqnarray}
\end{proof}

If a diagonal matrix has both positive and negative real values then we can have the following decompositions. These can be especially useful to reduce the $\ell_1$ norm of the coefficients, if the negative entries are quite large. 
First, we give a decomposition for the case when the entries are not necessarily integers. 
\begin{lemma}
Let $M$ is a $N\times N$ diagonal real matrix, such that there are $N'$ distinct non-zero positive elements and $N''$ distinct (non-zero) negative elements. Then $M$ can be written as $M=c_0\id+\sum_{i=1}^{N'+N''}c_iD_i$, where $D_i$ are signature matrices and $\sum_{i=0}^{N'+N''}|c_i|=m_{max}$, where $m_{max}$ is the largest positive diagonal entry.
 \label{lem:M-}
\end{lemma}
\begin{proof}
 We split $M$ as difference of two matrices - $M=M'-M''$, where $M'[i,i]=M[i,i]$, if $M[i,i]>0$, else $M'[i,i]=0$. While, $M''[i,i]=|M[i,i]|$, if $M[i,i]<0$, else $M''[i,i]=0$. So, both $M'$ and $M''$ are positive diagonal matrices. Let us order the elements (including 0) of $M'$ and $M''$, respectively, as follows - $m_0'<m_1'<\cdots<m_{N'}'$ and $m_0''<m_1''<\cdots<m_{N''}''$. We note that the subscripts here do not reflect the position of the element along the diagonal. For both the matrices $m_0'=m_0''=0$.
 We can use Lemma \ref{lem:M} and decompose the matrices as sum of identity and signature matrices,
 \begin{eqnarray}
  M'&=&\left(m_0'+\frac{1}{2}\sum_{i=1}^{N'}(m_i'-m_{i-1}')\right)\id+\frac{1}{2}\sum_{i=1}^{N'}(m_i'-m_{i-1}')D_i' \nonumber \\
  M''&=&\left(m_0''+\frac{1}{2}\sum_{i=1}^{N''}(m_i''-m_{i-1}'')\right)\id+\frac{1}{2}\sum_{i=1}^{N''}(m_i''-m_{i-1}'')D_i'' \nonumber 
 \end{eqnarray}
where $D_i', D_i''$ are signature matrices. Therefore,
\begin{eqnarray}
 M&=&\left(\frac{1}{2}\sum_{i=1}^{N'}(m_i'-m_{i-1}')-\frac{1}{2}\sum_{i=1}^{N''}(m_i''-m_{i-1}'')\right)\id+\frac{1}{2}\sum_{i=1}^{N'}(m_i'-m_{i-1}')D_i'-\frac{1}{2}\sum_{i=1}^{N''}(m_i''-m_{i-1}'')D_i''  \nonumber
\end{eqnarray}
since $m_0'=m_0''=0$ and hence the sum of the absolute value of the coefficients is 
\begin{eqnarray}
&& \left(\frac{1}{2}\sum_{i=1}^{N'}(m_i'-m_{i-1}')-\frac{1}{2}\sum_{i=1}^{N''}(m_i''-m_{i-1}'')\right)+\frac{1}{2}\sum_{i=1}^{N'}(m_i'-m_{i-1}')+\frac{1}{2}\sum_{i=1}^{N''}(m_i''-m_{i-1}'')  \nonumber \\
&=& \sum_{i=1}^{N'}(m_i'-m_{i-1}')=m_{N'}=m_{max}.
\end{eqnarray}

\end{proof}
We can have similar result when the entries are both positive and negative integers, including 0.
\begin{lemma}
Let $M_I$ is a $N\times N$ diagonal integer matrix, which has $N'$ positive integers whose maximum value is $m_{max}'$ and $N''$ negative integers such that $m_{max}''=\max_i \{|M_I[i,i]|:M_I[i,i]<0\}$. Then, $M=c_0\id+\sum_{i=1}^{N'+N''}c_iD_i$, where $D_i$ are signature matrices and $N'\leq\lceil\log_2(m_{max}'+1)\rceil=\zeta'$ and $N''\leq\lceil\log_2(m_{max}''+1)\rceil=\zeta''$. Also, $\sum_{i=0}^{N'+N''}|c_i|\leq 2^{\zeta'}-1$.
 \label{lem:MI-}
\end{lemma}

Now we consider some particular matrices that are relevant for the operators that appear in this paper. We give some decompositions as sum of some fundamental quantum gates. 

\paragraph{I. LCU decomposition of $A$ and $A^2$ : } First consider the operator $A$ described as 
\begin{eqnarray}
    A&=&\frac{1}{i\Delta}\mathcal{F}\left(\frac{2\pi i}{d}C'\right)\mathcal{F}^{\dagger}\qquad\text{where}\quad C'=\diag(0,1,2,\ldots,d-1)
\end{eqnarray}

\begin{lemma}
 Let $U=\id_{(\zeta)}\otimes\ldots\otimes\id_{(\ell+1)}\otimes\Z_{(\ell)}\otimes\id_{(\ell-1)}\otimes\ldots\otimes\id_{(1)}$ is a tensor product of $\zeta$ single-qubit unitaries where $\Z$ is applied on qubit $\ell$ and $\id$ on the rest. Then $U$ is a diagonal matrix of the following form.
 \begin{eqnarray}
  U_{j,j}&=& 1\qquad\text{if }\qquad j=2^{\ell}k,2^{\ell}k+1,\ldots 2^{\ell}k+2^{\ell-1}-1    \nonumber \\
  &=&-1\qquad\text{if }\qquad j=2^{\ell}k+2^{\ell-1},2^{\ell}k+2^{\ell-1}+1,\ldots 2^{\ell}k+2^{\ell}-1    \nonumber \\
  &\qquad\qquad&\qquad\text{where }\qquad k=0,1,\ldots,2^{\zeta-\ell}-1   \nonumber
 \end{eqnarray}
 \label{lem:signZ}
\end{lemma}
\begin{proof}
 $\id_{(\ell-1)}\otimes\ldots\otimes\id_{(1)}=\id_{2^{\ell-1}}$ is a $2^{\ell-1}\times 2^{\ell-1}$ identity matrix. 
 
 For any matrix $M$, $\Z\otimes M$ gives a block diagonal matrix with $+M$ in the first or upper block and $-M$ in the second or lower block. Thus, $\Z_{(\ell)}\otimes\left(\id_{(\ell-1)}\otimes\ldots\otimes\id_{(1)}\right)=\Z\otimes\id_{2^{\ell-1}}$ is a diagonal matrix which $+\id_{2^{\ell-1}}$ in the first block and $-\id_{2^{\ell-1}}$ in the second block. This implies that it has $+1$ in the first $2^{\ell-1}$ diagonal entries and $-1$ in the remaining $2^{\ell-1}$ entries.
 
 $\id_{(\zeta)}\otimes\ldots\otimes\id_{(\ell+1)}=\id_{2^{\zeta-\ell}}$ is a $2^{\zeta-\ell}\times 2^{\zeta-\ell}$ identity matrix.
 
 For any matrix $M'$, $\id_{2^{\zeta-\ell}}\otimes M'$ is a block diagonal matrix with $M'$ embedded along the diagonal. Thus $U=\id_{2^{\zeta-\ell}}\otimes\left(\Z\otimes\id_{2^{\ell-1}}\right)$ is a $2^{\zeta}\times 2^{\zeta}$ matrix with $\left(\Z\otimes\id_{2^{\ell-1}}\right)_{2^{\ell}\times 2^{\ell}}$ repeating $2^{\zeta-\ell}$ times along the diagonal. This explains the range of the index $k$ in the statement of the lemma. Since in each block the first $2^{\ell-1}$ entries are $+1$ and the remaining ones are $-1$, we get the above-mentioned range of the index $j$.
\end{proof}

\begin{lemma}
 Let $C'=\diag(0,1,2,3,\ldots,d-1)$. Then, assuming $d=2^{\zeta}$, 
 \begin{eqnarray}
  C'=\frac{2^{\zeta}-1}{2}\id-\frac{1}{2}\sum_{i=0}^{\zeta-1}2^i\Z_{(i+1)}\qquad\text{where}\qquad \Z_{(i+1)}=\id_{(\zeta)}\otimes\ldots\otimes\id_{(i+2)}\otimes\Z_{(i+1)}\otimes\id_{(i)}\otimes\ldots\otimes\id_{(1)}   \nonumber
 \end{eqnarray}
\label{lem:C'}
\end{lemma}
\begin{proof}
 Using Lemma \ref{lem:MI} we can decompose $C'$ as follows.
 \begin{eqnarray}
  C'=\sum_{i=0}^{\zeta-1}2^i\diag\left(b_i^{(0)},b_i^{(1)},\ldots,b_i^{(d-1)}\right)   \nonumber
 \end{eqnarray}
where $b_i^{(j)}$ is the $i^{th}$ bit occurring in the binary expansion of $j$. Using Equation \ref{eq:binary_diag_mat} we can further decompose each $\diag\left(b_i^{(0)},b_i^{(1)},\ldots,b_i^{(d-1)}\right)$ as sum of identity and a signature matrix. Now if we follow the binary decomposition of consecutive integers then we see that $\diag\left(b_0^{(0)},b_0^{(1)},\ldots,b_0^{(d-1)}\right)$ has alternative 0 and 1, which leads to a signature matrix with alternate $+1$ and $-1$. $\diag\left(b_1^{(0)},b_1^{(1)},\ldots,b_1^{(d-1)}\right)$ has two 0, the next two 1 and so on. This yields a signature matrix with two $+1$, then next two $-1$ and so on. Generalizing $\diag\left(b_i^{(0)},b_i^{(1)},\ldots,b_i^{(d-1)}\right)$ has the first $2^{i+1}$ entries 0, the next $2^{i+1}$ entries are 1 and so on and we get the corresponding pattern in the signature matrix. Thus using Lemma \ref{lem:signZ} we can write
\begin{eqnarray}
 \diag\left(b_i^{(0)},b_i^{(1)},\ldots,b_i^{(d-1)}\right)=\frac{1}{2}\left(\bigotimes_{j=1}^{\zeta}\id_{(j)}-Z_{(i+1)}\right) \nonumber
\end{eqnarray}
and hence the lemma follows.
\end{proof}
As a corollary we can have the following LCU decomposition of $A$.
\begin{corollary}
 Let $d=2^{\zeta}$. Then we can write 
 \begin{eqnarray}
  A=\frac{2\pi}{d\Delta}\left(\frac{2^{\zeta}-1}{2}\id-\mathcal{F}\left(\sum_{i=0}^{\zeta-1}2^{i-1}\Z_{(i+1)}\right)\mathcal{F}^{\dagger}\right) \nonumber
 \end{eqnarray}
\label{cor:A_lcu}
\end{corollary}
It can be easily shown that this dcomposition not only has the same number of unitaries, but also the same $\ell_1$ norm, compared to the decomposition obtained by using Lemma \ref{lem:MI}. The additional advantage is the fact that in this case we specified the unitaries in terms of fundamental gates. The following result also follows from the previous corollary.
\begin{corollary}
Let $d=2^{\zeta}$. Then we can write
\begin{eqnarray}
    A^2=  \frac{\pi^2}{d^2\Delta^2}\mathcal{F}\left(\left(c_0^2+\sum_{i=0}^{\zeta-1}2^{2i}\right)\id-2c_0\sum_{i=0}^{\zeta-1}2^i\Z_{(i+1)}+2\sum_{i=0}^{\zeta-2}\sum_{j=i+1}^{\zeta-1}2^{i+j}\Z_{(i+1)}\Z_{(j+1)}\right)\mathcal{F}^{\dagger}  \nonumber
\end{eqnarray}
    \label{cor:A2_lcu}
\end{corollary}
It can be shown that compared to Lemma \ref{lem:MI}, this decomposition has the same $\ell_1$ norm, that is,
\begin{eqnarray}
&&\frac{\pi^2}{d^2\Delta^2}\left( c_0^2+\sum_{i=0}^{\zeta-1}2^{2i}+2c_0\sum_{i=0}^{\zeta-1}2^i+2\sum_{i=0}^{\zeta-2}2^i\sum_{j=i+1}^{\zeta-1}2^j \right) \nonumber \\
&=&\frac{\pi^2}{d^2\Delta^2}\left(c_0^2+\frac{4^{\zeta}-1}{3}+2c_0(2^{\zeta}-1)+2\sum_{i=0}^{\zeta-2}2^{i+i+1}\sum_{j=0}^{(\zeta-1)-(i+1)}2^j \right) \nonumber \\
&=&\frac{\pi^2}{d^2\Delta^2}\left(c_0^2+\frac{d^2-1}{3}+2c_0(d-1)+2\sum_{i=0}^{\zeta-3}2^{i+i+1}(2^{\zeta-(i+1)}-1)+2\cdot 2^{(\zeta-2)+(\zeta-1)}\right)   \nonumber \\
&=&\frac{\pi^2}{d^2\Delta^2}\left(c_0^2+\frac{d^2-1}{3}+2c_0(d-1)+2^{\zeta+1}\sum_{i=0}^{\zeta-3}2^i+2^{2\zeta-2}-4\sum_{i=0}^{\zeta-3}4^i\right) \nonumber \\
&=&\frac{\pi^2}{d^2\Delta^2}\left(c_0^2+\frac{d^2-1}{3}+2c_0(d-1)+2^{\zeta+1}(2^{\zeta-2}-1)+\frac{4^{\zeta}}{4}-4\frac{4^{\zeta-2}-1}{3}\right)  \nonumber \\
&=&\frac{\pi^2}{d^2\Delta^2}\left(d^2\left(3+\frac{1}{3}+\frac{1}{2}+\frac{1}{4}-\frac{1}{12}\right)-8d+4\right)=\frac{4\pi^2(d-1)^2}{d^2\Delta^2}\leq\frac{4\pi^2}{\Delta^2}   \nonumber
\end{eqnarray}
but the number of unitaries, that is,
\begin{eqnarray}
 1+(\zeta-1)+\left((\zeta-1)+(\zeta-2)+\cdots+1\right)=\frac{\zeta(\zeta+1)}{2}=\frac{\log_2^2d+\log_2d}{2}
\end{eqnarray}
is slightly more. But the advantage is the fact, that we have a decomposition in terms of very fundamental quantum gates. 

\paragraph{II. LCU decomposition of $\nabla$ and $\nabla^2$ : } We approximate these differential operators with $(2a+1)$-point central difference formula, by which we get decompositions as sum of adders. In general the approximation error is given by the following expression,
\begin{lemma}[\cite{2005_L, 2017_KWBA}]
 $$
 \nabla_{\mu}^2\psi(x)=\frac{1}{h^2}\sum_{k=-a}^ad_{2a+1,k}\psi(x+kh\hat{e}_{\mu})+\mathcal{R}_{2a+1}
 $$
 where $\hat{e}_{\mu}$ is the unit vector along the $\mu^{th}$ component of $x$, $(x+kh\hat{e}_{\mu})$ is evaluated modulo the grid length $L$, $\mathcal{R}_{2a+1}\in O(h^{2a-1})$ and
 $$
    d_{2a+1,k\neq 0}=\frac{2(-1)^{a+k+1}(a!)^2}{(a+k)!(a-k)!k^2}\qquad d_{2a+1,k=0}=-\sum_{k=-a,k\neq 0}^ad_{2a+1,k}.
 $$
 \label{lem:lcuNabla2}
\end{lemma}
More specific bounds can be made on the truncation error for these central difference formulas through the use of Taylor's remainder theorem.  Also the bounds on the 1-norm of the coefficients for the adder decomposition is  summarized below.
\begin{lemma}[Theorem 7 and Lemma 6 in \cite{2017_KWBA}]
 Let $\psi(x)\in\cmplx^{2a+1}$ on $x\in\real$ for $a\in\intg_{+}$. Then the error in the $(2a+1)$-point centered difference formula for the second derivative of $\psi(x)$ evaluated on a uniform mesh with spacing $h$ is at most
 $$
 \left|\mathcal{R}_{2a+1}\right|\leq \frac{\pi^{3/2}}{9}e^{2a[1-\ln 2]}h^{2a-1}\max_x\left|\psi^{(2a+1)}(x)\right|.
 $$
 Also, the sum of the norms of the coefficients is bounded above as follows.
 $$
    \sum_{k=-a,k\neq 0}^a\left|d_{2a+1,k}\right|\leq \frac{2\pi^2}{3}.
 $$
 \label{lem:d}
\end{lemma}
Thus $\nabla^2$ is approximated by a sum of $2a+1$ adders and the $\ell_1$ norm of the coefficients is at most $\frac{4\pi^2}{3h^2}$.

Next we need a similar expression for the gradient so that we can understand how to block encode the result as a function of the number of points used in the decomposition.  The first result stems from earlier work by Li \cite{2005_L} which gives a high order derivative expression using centered differences.
\begin{lemma}[\cite{2005_L}]
 $$
    \nabla_{\mu}\psi(x)=\frac{1}{h}\sum_{k=-a}^ad_{2a+1,k}'\psi(x+kh\hat{e}_{\mu})+\mathcal{R}_{2a+1}'
 $$
 where $\hat{e}_{\mu}$ is the unit vector along the $\mu^{th}$ component of $x$, $(x_{\mu}+kh\hat{e}_{\mu})$ is evaluated modulo the grid length $L$, $|\mathcal{R}'_{2a+1}| \in O(h^{2a})$ and
 $$
    d_{2a+1,k}'=\frac{(-1)^{k+1}(a!)^2}{j(a-k)!(a+k)!}\qquad d_{2a+1,0}'=0.
 $$
 \label{lem:lcuNabla}
\end{lemma}
Next there we need to bound the one-norm of this formula, the bound for which is given below.
\begin{lemma}
The sum of the norms of the coefficients in the $(2a+1)$-point centered finite difference formula is bounded above as follows.
 \begin{eqnarray}
    \sum_{k=-a,k\neq 0}^a\left|d_{2a+1,k}'\right|&\leq& 2\ln a+\gamma\qquad\text{where } \gamma\approx0.577\text{ is the Euler-Mascheroni constant.}     \nonumber \\
    &\leq&\ln 2a^2 \qquad\text{when }a\geq\sqrt{e}\approx 1.4  \nonumber
 \end{eqnarray}
 \label{lem:d'}
\end{lemma}
\begin{proof}
 \begin{eqnarray}
  \sum_{k=-a,k\neq 0}^a\left|d_{2a+1,k}'\right|&=& \sum_{k=-a,k\neq 0}^a\frac{(a!)^2}{|k|(a-k)!(a+k)!} \nonumber \\
  &\leq&\sum_{k=-a,k\neq 0}^a\frac{1}{|k|}\qquad\text{for } |k|\leq a\text{ when } a\geq 1    \nonumber   \\
  &=&   2\sum_{k=1}^a\frac{1}{|k|}=2\ln a+\gamma\qquad \gamma\approx0.577\text{ is the Euler-Mascheroni constant.}   \nonumber
 \end{eqnarray}
\end{proof}
Finally, for completeness we prove a specific truncation bound on the finite difference approximation to the gradient.
\begin{lemma}
 Let $\psi(x)\in\cmplx^{2a+1}$ on $x\in\real$ for $a\in\intg_{+}$. Then the error in the $(2a+1)$-point centered difference formula for the first derivative of $\psi(x)$ evaluated on a uniform mesh with spacing $h$ is at most
 $$
   \left|\mathcal{R}_{2a+1}'\right|\leq \frac{(2\ln a+\gamma)}{6\sqrt{\pi}}e^{2a[1-\ln 2]}h^{2a+1}\max_x\left|\psi^{(2a+1)}(x)\right|.
 $$
 \label{lem:R'}
\end{lemma}
\begin{proof}
 Using Corollary 2.1 in \cite{2005_L} and the triangle inequality we have the following.
 \begin{eqnarray}
  \left|\mathcal{R}_{2a+1}'\right|&\leq& \frac{h^{2a+1}}{(2a+1)!}\max_x\left|\psi^{(2a+1)}(x)\right|\sum_{k=-a,\neq 0}^a|d_{2a+1,k}||k|^{2a+1}  \nonumber \\
  &\leq& \frac{h^{2a+1}a^{2a+1}}{(2a+1)!} \max_x\left|\psi^{(2a+1)}(x)\right|\sum_{k=-a,\neq 0}^a|d_{2a+1,k}|  \nonumber \\
  &<& \frac{(2\ln a+\gamma)h^{2a+1}a^{2a+1}}{(2a+1)!} \max_x\left|\psi^{(2a+1)}(x)\right| \qquad [\text{Using Lemma } \ref{lem:d'}] \label{eqn:R'temp}
 \end{eqnarray}
Using Stirling's approximation and the fact that $a\geq 1$ we have the following.
\begin{eqnarray}
 \frac{a^{2a+1}}{(2a+1)!}\leq \frac{\sqrt{a}e^{2a[1-\ln 2]}}{2(2a+1)\sqrt{\pi}}\leq\frac{e^{2a[1-\ln 2]}}{6\sqrt{\pi}}  \nonumber
\end{eqnarray}
The Lemma is proved by substituting the above inequality into Equation \ref{eqn:R'temp}.
\end{proof}
Thus $\nabla$ can be written as sum of $2a$ unitaries, which are adders, and the $\ell_1$ norm of the coefficients is at most $2\ln a+\gamma\leq \ln (2a^2)$. 

As a final note, we see here that the accuracy of the discrete derivatives considered increases as we increase the number of points used in the formula provided that the underlying wave function is sufficiently smooth.  For our purposes, we will not discuss in detail the specific value of $a$ that is optimal and will assume that it is a constant.  This is because it is in general difficult to provide bounds on the values of the higher-order derivatives of the wave function as a function of the evolution time.  While high order formulas in principle can be valuable to address accuracy concerns, we need such guarantees in order to understand the optimal order to take for a particular evolution.  This is especially relevant since the initial state is not necessarily in $C^\infty$ and thus the asymptotic advantages may disappear for classes of functions that are not sufficiently smooth.  For these reasons, we leave detailed discussion of the truncation error to subsequent work and focus on the case where $a$ is a constant.

\paragraph{III. LCU decomposition of $E^2$ : } In the electric link basis, $E_{\ell,\mu}^2=\sum_{\epsilon=-\Lambda}^{\Lambda-1}\epsilon^2\ket{\epsilon}\bra{\epsilon}_{\ell,\mu} $, is a diagonal positive integer matrix and so we can use Lemma \ref{lem:MI} to express it as a linear combination of at most $1+\lceil\log_2(\Lambda^2+1)\rceil\approx 2\log_2\Lambda$ unitaries and the $\ell_1$ norm of the coefficients is at most $\Lambda^2$. Alternatively, $E^2$ can be expressed as linear combination of slightly more number of Z operators, but with the same $\ell_1$ norm \cite{2020_SLSW}.
\begin{eqnarray}
 E^2=\frac{1}{6}\left(2^{2\zeta-1}+1\right)\id+\sum_{j=0}^{\zeta-1}2^{j-1}\Z_j+\sum_{j=0}^{\zeta-2}\sum_{k>j}^{\zeta-1}2^{j+k-1}\Z_j\Z_k, \quad\text{ where }\quad\zeta=\log_2\Lambda
 \label{eqn:E2_lcu_0}
\end{eqnarray}

\paragraph{IV. LCU decomposition of $U$ : } We know that $U=\sum_{\epsilon=-\Lambda}^{\Lambda-1}\ket{\epsilon+1}\bra{\epsilon}=\exp(i\Delta A)=\mathcal{F}C\mathcal{F}^{\dagger}$, where C is the Sylvester's "clock" matrix defined as
\begin{equation}
    C = \begin{pmatrix}
            1 & 0 & 0        &\cdots & 0 \\
            0 & \omega & 0   &\cdots & 0 \\
            0 & 0 & \omega^2 &\cdots & 0 \\
            \vdots & \vdots & \vdots & \ddots & \vdots \\
            0 & 0 & 0 &\cdots & \omega^{d-1} \\
        \end{pmatrix}   \qquad [\omega = e^{2\pi i/d},\quad d=\text{dimension of }C]    \nonumber
\end{equation}
\begin{lemma}
Let $R(2^k)=\begin{bmatrix}1 & 0 \\ 0 & \omega^{2^k} \end{bmatrix}$, be a rotation gate. Then,
\begin{eqnarray}
 R(2^x) \otimes\cdots\otimes R(2^1)\otimes R(2^0)=\bigotimes_{k=0}^xR(2^k)=\diag(1,\omega,\ldots,\omega^{2^{x+1}-1}).
\end{eqnarray}
\label{lem:C}
\end{lemma}
\begin{proof}
 We prove by induction. For the base case we have 
 \begin{eqnarray}
  R(2^1)\otimes R(2^0)=\diag(1,\omega^2)\otimes\diag(1,\omega)=\diag(1,\omega,\omega^2,\omega^3).   \nonumber
 \end{eqnarray}
Let the lemma holds when $k=m-1$, i.e.
\begin{eqnarray}
 \bigotimes_{k=0}^{m-1}R(2^k)=\diag(1,\omega,\ldots,\omega^{2^{m}-1}). \nonumber
\end{eqnarray}
Then
\begin{eqnarray}
 \bigotimes_{k=0}^mR(2^k)&=&R(2^m)\otimes\left(\bigotimes_{k=0}^{m-1}R(2^k)\right)=\diag(1,\omega^{2^m})\otimes\diag(1,\omega,\ldots,\omega^{2^m-1}) \nonumber \\
 &=&\diag(1,\omega,\ldots,\omega^{2^m-1},\omega^{2^m},\omega^{2^m+1},\ldots,\omega^{2^{m+1}-1}) \nonumber
\end{eqnarray}
and the lemma follows.
\end{proof}
Setting $d=2^{x+1}$, we can have the following decomposition of $U$, and hence the $\ell_1$ norm of $U$ is 1. 
\begin{corollary}
Let $R(2^k)=\begin{bmatrix}1 & 0 \\ 0 & \omega^{2^k} \end{bmatrix}$ be a rotation gate. Then,
$ U=\mathcal{F}\left(\bigotimes_{k=0}^{\log_2d-1}R(2^k)\right)\mathcal{F}^{\dagger} $.
    \label{cor:U_lcu}
\end{corollary}

\paragraph{V. LCU decomposition of the fragment Hamiltonians : }
We can decompose the fragment Hamiltonians appearing in this paper (Equation \ref{eqn:HPF}), as sum of unitaries, using the LCU decomposition of the above operators and Lemma \ref{lem:M}. We briefly describe the decompositions in the points below.
\begin{enumerate}
    \item $\boldsymbol{H_s}=-\frac{1}{c}\sum_{j=1}^{\eta}\sum_{q=1}^N\sum_{\mu\neq\nu\neq\xi=1}^3\sigma_{j,\mu}\otimes\id\otimes\left(\nabla_{\nu}A_{q,\xi}-\nabla_{\xi}A_{q,\nu}\right)$ : Using the LCU decompositions of $\nabla$ (Lemma \ref{lem:lcuNabla}) and $A$ (Corollary \ref{cor:A_lcu}), we can write it as sum of $2\cdot3\eta\cdot N(2a+1)(\lceil\log_2d\rceil+1)\approx 12\eta Na\log_2d$ unitaries and using triangle inequality (Appendix \ref{app:norm}) we get the $\ell_1$ norm as $\leq\frac{1}{c}3\eta N\cdot 2\cdot\frac{\ln (2a^2)}{h}\cdot\frac{2\pi}{\Delta}=\frac{12\pi\eta N\ln (2a^2)}{ch\Delta}$.

    \item $\boldsymbol{H_{V_{ee}}}=\frac{1}{\Delta}\sum_{k<j}^{\eta}\sum_{\q,\vect{r}}^N\left(\id\otimes\frac{1}{\|\q-\vect{r}\|_2}\ket{\q}\bra{\q}_k\ket{\vect{r}}\bra{\vect{r}}_j\otimes\id\right)$ :  Here $\sum_{\q,\vect{r}}^N  \ket{\q}\bra{\q}_k\otimes\ket{\vect{r}}\bra{\vect{r}}_j $ is a diagonal operator.  Using Lemma \ref{lem:M} we can express 
$\sum_{\q,\vect{r}}\frac{1}{\|\q-\vect{r}\|_2}\ket{\q}\bra{\q}_i\otimes\ket{\vect{r}}\bra{\vect{r}}_j$ as sum of $\id$ and $N-1$ signature matrices and the $\ell_1$ norm of the deocmposition is $\left(\frac{1}{\|\q-\vect{r}\|_2}\right)_{\max}=\frac{1}{\Delta}$. Hence $H_{V_{ee}}$ can be decomposed as sum of at most $\frac{\eta(\eta-1)N}{2}$ unitaries with $\ell_1$ norm at most $\frac{\eta(\eta-1)}{2\Delta^2}$.  Since $\Delta = \frac{\Omega^{1/3}}{N^{1/3}}$, so the dependence on $N$ is sort of implicit in the norm.

    \item $\boldsymbol{H_{V_{ne}}}= -\frac{1}{\Delta}\sum_{j}^{\eta}\sum_{\kappa}^K\sum_{\q}^N\left(\id\otimes\frac{Z_{\kappa}}{\|\q-\vect{R}_{\kappa}\|_2}\ket{\q}\bra{\q}_j\otimes\id\right) : $ Here $\sum_{\q}^N\ket{\q}\bra{\q}$ is a diagonal operator.
 The number of distinct values of $\frac{Z_{\kappa}}{\|\q-\vect{R}_{\kappa}\|_2}$, for a particular value of $\kappa$, is at most $N$. So, using Lemma \ref{lem:M} we can decompose $\sum_q\frac{Z_{\kappa}}{\|\q-\vect{R}_{\kappa}\|_2}\ket{\q}\bra{\q}$ as sum of $\id$ and at most $N-1$ signature matrices, and the $\ell_1$ norm is at most $Z_{\kappa}/d_{min}$, where $d_{min}=\left(\|\q-\vect{R}_{\kappa}\|_2\right)_{min}=\Delta$. Hence $H_{V_{ne}}$ can be decomposed as sum of at most $\eta KN$ unitaries and the $\ell_1$ norm is at most $\frac{\eta Z_{sum}}{\Delta^2}$, where $Z_{sum}=\sum_{\kappa=1}^K|Z_{\kappa}|$. Here too, the dependence of the norm on $N$ is implied due to $\Delta$.

 \item $\boldsymbol{H_{f1}} = \sum_{q}^N\sum_{\mu}^3\id\otimes\id\otimes\frac{1}{2}E_{q,\mu}^2 : $ Since $E_{q,\mu}^2$ is a diagonal positive integer matrix, so we can use Lemma \ref{lem:MI}.  Thus $H_{f1}$ can be decomposed as sum of at most $3N\left(1+\lceil\log_2\left(\Lambda^2+1\right)\rceil\right)\approx 6N\log_2\Lambda$ unitaries with $\ell_1$ norm at most $\frac{3N\Lambda^2}{2}$.

 \item $\boldsymbol{H_{f2}} = -\sum_{q}^N\sum_{\mu\neq\nu}^3\id\otimes\id\otimes\left(W_{q,\mu,\nu}^2+h.c.\right) : $ Using Corollary \ref{cor:U_lcu} and the definition of $W^2$ (Equation \ref{eqn:W}), we can decompose $H_{f2}$ as sum of $6N$ unitaries, with $\ell_1$ norm at most $6N$.

 \item $\boldsymbol{H_{1\pi}} = -\frac{1}{2}\sum_{j}^{\eta}\sum_{q}^N\sum_{\mu}^3\id\otimes\nabla_{j,\mu}^2\otimes\id : $ Using Lemma \ref{lem:lcuNabla2} we can express it as sum of at most $6a\eta N$ unitaries and the $\ell_1$ norm is at most $6\eta N\frac{4\pi^2}{3h^2}=\frac{8\pi^2\eta N}{h^2}$.

 \item $\boldsymbol{H_{2\pi}} = \frac{1}{c}\sum_{j}^{\eta}\sum_{q}^N\sum_{\mu}^3\id\otimes (i\nabla_{j,\mu})\otimes A_{q,\mu} : $ 
 Using Lemma \ref{lem:lcuNabla} and C \ref{cor:A_lcu} we have a decomposition with sum of at most $6a\eta N\log_2d$ unitaries and using Lemma \ref{lem:d'} the $\ell_1$ norm is at most $\frac{12\pi\eta N\ln (2a^2)}{ch\Delta}$.

 \item $\boldsymbol{H_{3\pi}} = \frac{1}{2c^2}\sum_{j}^{\eta}\sum_q^N\sum_{\mu}^3\id\otimes\id\otimes A_{q,\mu}^2 : $ We use Lemma \ref{lem:MI} to decompose $\diag\left(0^2,1^2,2^2,\ldots,(d-1)^2\right)$, and hence $A^2$, as sum of at most $\lceil\log_2(d-1)^2+1\rceil+1\lessapprox 2\log_2d$ unitaries. Thus $H_{3\pi}$ can be exressed as sum of at most $6\eta N\log_2d$ unitaries and the $\ell_1$ norm is at most $3\eta N\frac{4\pi^2}{c^2\Delta^2}=\frac{12\pi^2\eta N}{c^2\Delta^2}$.
 \end{enumerate}
In Table \ref{tab:norm} we have summarized the number of unitaries in the LCU decomposition of various operators and Hamiltonians and also mentioned an upper bound on the $\ell_1$ norm. 
 
\section{Circuit decompositions for simulation circuits}
\label{sec:method}

The following four sub-sections describe the algorithms to simulate the exponential of the four Hamiltonian terms that comprise our Hamiltonian, depicted as leaves in Figure \ref{fig:divNconq}, thus proving Lemma \ref{lem:g1}-\ref{lem:g32}. As mentioned before, we describe the circuits in terms of Clifford+T and (controlled)-rotation gates.  

\subsection{Algorithm to simulate $e^{-iH_{12}\tau_1}=e^{i\tau_1 \sum_{q=1}^N\sum_{\mu\neq\nu=1}^3\id\otimes\id\otimes W_{q,\mu,\nu}^2  }$}
\label{subsec:H12}

In this section we prove the complexity of simulating $e^{-iH_{12}\tau_1}$ using qubitization, thus proving Lemma \ref{lem:g1}. We know that $H_{12}=H_{f2}$ (Equation \ref{eqn:Hf}) which corresponds to the plaquette terms in the dynamics of the system. In Corollary \ref{cor:U_lcu} we show that the raising operator $U_{q,\mu}$ (defined in Equation \ref{eq:raising_link_op}).
\begin{eqnarray}
 U_{q,\mu}=\mathcal{F}_{q,\mu}\left(\bigotimes_{k=1}^{\log_2d-1}R_z(\theta_k)\right)_{q,\mu}\mathcal{F}_{q,\mu}^{\dagger}\qquad\text{where }\theta_k=\frac{2\pi}{d}2^k\text{ and }\mathcal{F}\text{ is Fourier Transform}.
\end{eqnarray}
This shows that an individual $U_{q,\mu}$ can be implemented using $\log_2(d)$ single qubit rotations.  The plaquette operator $W_{q,\mu,\nu}$ can be implemented using $4$ such terms and thus $W_{q,\mu,\nu}^2$ (Equation \ref{eqn:W}) can be implemented by a layer of at most $4\log_2d$ parallel rotations, conjugated by Fourier transformation. The ancilla preparation sub-routine does the following.
\begin{eqnarray}
 \prep_{f2}\ket{0}^*=\left(\frac{1}{\sqrt{N}}\sum_{q=1}^N\ket{q}\right)\otimes\left(\frac{1}{\sqrt{6}}\sum_{\mu\neq\nu=1}^3\sum_{k=0}^1\ket{\mu}\ket{\nu}\ket{k}\right)
 \label{H12:prep}
\end{eqnarray}
First we have the $\log_2N$-qubit electric link index register that stores the $N$ electric link indices in equal superposition. Next we have the 4+1-qubit spin index register. The first 4 qubits stores the value of $\mu, \nu$. The last 1 qubit indicates if we apply h.c. If $\mu=\nu$ or $\mu,\nu>3$, then we discard. Throughout this paper, by 'discarding' we mean unfollowing a computation path. This is indicated by an ancilla qubit, which when set to $\ket{1}$, we only apply $\id$. All the registers are in equal superposition and so we require $5+\log_2N$ H gates. Comparing the constraints on $\mu,\nu$ takes $O(1)$ extra gates and ancillae. 

The unitary selection sub-routine does the following.
\begin{eqnarray}
 &&\sel_{f2}:\ket{q}\ket{\mu,\nu,0}\left(\bigotimes_{q=1}^N\bigotimes_{\mu'=1}^3\ket{f_e=0}\right)\ket{\phi} \nonumber \\
 &\mapsto&\ket{q}\ket{\mu,\nu,0}\left(\ket{1}_{q,\mu}\ket{1}_{q+1_{\mu},\nu}\ket{1}_{q+1_{\nu},\mu}\ket{1}_{q,\nu}\right)\mathcal{F}\left(\bigotimes_{k=1}^{\log_2d-1}R_z(\theta_k)\right)_{q,\mu}\left(\bigotimes_{k=1}^{\log_2d-1}R_z(\theta_k)\right)_{q+1_{\mu},\nu}   \nonumber \\
&& \left(\bigotimes_{k=1}^{\log_2d-1}R_z(-\theta_k)\right)_{q+1_{\nu},\mu} \left(\bigotimes_{k=1}^{\log_2d-1}R_z(-\theta_k)\right)_{q,\nu}\mathcal{F}^{\dagger} \ket{\phi}
\label{H12:sel}
\end{eqnarray}
Throughout this paper, for an operator $U$ and subspace indexed by some letter $q$, we write $\left(U\right)_q$ to imply identity acts on the remaining subspaces. The above selection operator can also be expressed as
\begin{eqnarray}
    \sel_{f2}&=&\sum_{q=1}^N\sum_{\mu\neq\nu=1}^3\sum_{k=0}^1\ket{q,\mu,\nu,0}\bra{q,\mu,\nu,0}\otimes U_{q,\mu}U_{q+1_{\mu},\nu}U_{q+1_{\nu},\mu}U_{q,\nu} \nonumber \\
    &&+\sum_{q=1}^N\sum_{\mu\neq\nu=1}^3\sum_{k=0}^1\ket{q,\mu,\nu,1}\bra{q,\mu,\nu,1}\otimes U_{q,\nu}^{\dagger}U_{q+1_{\nu},\mu}^{\dagger}U_{q+1_{\mu},\nu}^{\dagger}U_{q,\mu}^{\dagger},  \nonumber
\end{eqnarray}
by which we can conveniently prove that
\begin{eqnarray}
    \bra{0}\prep_{f2}^{\dagger}\cdot\sel_{f2}\cdot\prep_{f2}\ket{0}=\frac{H_{f2}}{6N},  \nonumber
\end{eqnarray}
providing a $(6N,.,0)$-block encoding if $H_{f2}$ and we also observe that $\|H_{f2}\|\leq 6N$, from Table \ref{tab:norm}. In the following sections we have preferred the format of Equation \ref{H12:sel} because for convenience in explaining the transformations of the states of the registers, including some ancillae qubits. This also helps in explaining the number of controls in some multi-controlled gates. 

In each of the $3N$ subspaces we allocate one ancilla $f_e$, initialized to $\ket{0}$ for selection. We use $N$ number of $C^{\log_2N}X$-gates, controlled on the link index register to select a link subspace. We use $3N$ number of $C^5X$ gates to select spin subspaces. The 4 controls are on the spin register and the last one is controlled on the target qubits of the $C^{\log_2N}X$ gates. In short, $\ket{f_e}$ is flipped to $\ket{1}$ if both the link register state $\ket{q}$ and spin register state $\ket{\mu}$ match. The remaining operations are all controlled on the state of $f_e$. We use the same set of gates at the end of the operations for uncomputing. We use the optimization technique of Theorem \ref{thm:CX} to synthesize the $NC^{\log_2N}X$ gates. If we split the ancillae into $M_{f2}$ sets, such that the $i^{th}$ set contains $\frac{\log_2N}{r_i}$ qubits, then the total number of pairs of multi-controlled-X gates we require is
\begin{eqnarray}
    \sum_{i=1}^{M_{f2}}N^{\frac{1}{r_i}}C^{\frac{\log_2N}{r_i}}X+NC^{M_{f2}}X+3NC^5X.
    \label{H12:CN}
\end{eqnarray}
In each of the $3N$ subspaces we apply rotation gates controlled on $\ket{f_e=1}$ and the third qubit in spin register. If the latter is $\ket{0}$ then the angles are as shown in Equation \ref{H12:sel}, else we apply the h.c. i.e. negative of these angles. Thus in each subspace we have two multi-controlled rotation gates. These rotation gates are conjugated by $\log_2d$-qubits QFT. If we use AQFT \cite{2020_NSM} then we incur a T-gate cost of about $2\cdot 3N\left( 8(\log_2d)(\log_2\log_2d-2)+1.2\log_2^2(\log_2d)\right)$ and a H-gate cost of about $2\cdot 3N\log_2d$. 

Thus for one block encoding of $\frac{H_{f2}}{6N}$ number of controlled rotation gates required is
\begin{eqnarray}
    \g_{1}^r=6N, \label{eqn:g1r}
\end{eqnarray}
number of T-gates required is 
\begin{eqnarray}
    \g_{1}^t\leq 48N(\log_2d)(\log_2\log_2d)+7.2N\log_2^2(\log_2d)+4\sum_{i=1}^{M_{f2}}N^{\frac{1}{r_i}}\frac{\log_2N}{r_i}+4NM_{f2}+48N,  \label{eqn:g1t}
\end{eqnarray}
while the number of CNOT gates required is
\begin{eqnarray}
    \g_{1}^c\leq 4\sum_{i=1}^{M_{f2}}N^{\frac{1}{r_i}}\frac{\log_2N}{r_i}+4NM_{f2}+51N. \label{eqn:g1c}  
\end{eqnarray}

Counting the H gates the total number of gates required for the block encoding of $\frac{H_{12}}{6N}$ is,
\begin{eqnarray}
    \g_{1}'&\leq& 6N\log_2d+\log_2N+5+105N+8\sum_{i=1}^{M_{f2}}N^{\frac{1}{r_i}}\frac{\log_2N}{r_i}+8NM_{f2}  \nonumber \\ 
    &\in& O\left(N\log_2d\right)\qquad \text{when }\frac{1}{r_i}\leq\frac{1}{2},
    \label{eqn:g1'}
\end{eqnarray}
assuming $M_{f2}$ is a constant. Using Corollary 60 of \cite{2019_GSLW} we need
\begin{eqnarray}
 O\left(N\tau_1+\frac{\log(1/\delta_{12})}{\log\log(1/\delta_{12})}\right)
\end{eqnarray}
calls to the $\sel_{f2}$ and $\prep_{f2}$ oracles that define the block encoding of $\frac{H_{f2}}{6N}$ in order to implement an $\delta_{1}$-precise  block encoding of $e^{-iH_{f2}\tau_1}$. Thus number of gates required  for simulating $e^{-iH_{f2}\tau_1}$ is as follows.
\begin{eqnarray}
 \g_{1}\in O\left(N^2\tau_1\log d+\frac{\log(1/\delta_{12})}{\log\log(1/\delta_{12})}N\log d\right)
 \label{eqn:g1}
\end{eqnarray}

\subsubsection{Algorithm to simulate $e^{-iH_{21}\tau_2}=e^{-i\tau_2\sum_{q=1}^N\sum_{\mu=1}^3\id\otimes\id\otimes\frac{1}{2}E_{q,\mu}^2}$}
\label{subsec:H21}

We know that $H_{21}=H_{f1}$ (Equation \ref{eqn:Hf}) and since $[E_{\ell,\mu}^2,E_{q,\nu}^2]=0$ if $\ell\neq q$, so $
 e^{-iH_{f1}\tau_2}=\prod_{q=1}^N\prod_{\mu=1}^3\id\otimes\id\otimes e^{-i\frac{1}{2}E_{q,\mu}^2\tau_2}$. If $\zeta=1+\log_2\Lambda$, then $E^2$ can be written as sum of Z-operators, as shown below, and we simulate $e^{-iH_{f1}\tau_2}$ by Trotterization, as done in \cite{2020_SLSW}.
\begin{eqnarray}
 E^2=\frac{1}{6}\left(2^{2\zeta-1}+1\right)\id+\sum_{j=0}^{\zeta-1}2^{j-1}\Z_j+\sum_{j=0}^{\zeta-2}\sum_{k>j}^{\zeta-1}2^{j+k-1}\Z_j\Z_k. \label{eqn:E2_lcu}
\end{eqnarray}
\begin{lemma}[Lemma 2 in \cite{2020_SLSW}]
 There exists a circuit that implements $e^{-iE^2\tau_2}$ on $\zeta$ qubits exactly, up to an (efficiently computable) global phase, using $\frac{(\zeta+2)(\zeta-1)}{2}$ CNOT operations and $\frac{\zeta(\zeta+1)}{2}$ single-qubit rotations. Here $\zeta=1+\log_2\Lambda$.
 \label{lem:E2020_SLSW}
\end{lemma}
Since $E^2$ is expressed as sum of Pauli operators so we can use the algorithm in \cite{2023_MWZ} to optimize the rotation gates at the cost of increasing a few Toffoli gates, which has T-count 7 \cite{2021_MM} or 4 \cite{2013_J}, if using classical measurements. Then, the number of (controlled) rotations is equal to the number of non-zero distinct eigenvalues (ignoring sign) of $E^2$, which is $\Lambda$. This can give less rotation gates till $\zeta=4$. Since Toffolis are exactly implementable, so it can even lead to less umber of T-gates, specially for low synthesis error regime. The number of CNOT gates can be further optimized using algorithms like \cite{2022_GHLetal, 2018_AAM}, with or without using connectivity constraints. 

Since the product terms of $e^{-iH_{f1}\tau_2}$ mutually commute, so there is no error introduced in the simulation of this Hamiltonian and the total number of rotation gates for one Trotter step is 
\begin{eqnarray}
\g_{2}^r\leq 3N\frac{(1+\log_2\Lambda)(2+\log_2\Lambda)}{2},
\label{eqn:g2r}
\end{eqnarray}
while the number of CNOTs is
\begin{eqnarray}
\g_{2}^c\leq 3N\frac{\log_2\Lambda(3+\log_2\Lambda)}{2}.
\label{eqn:g2c}
\end{eqnarray}
Thus number of gates required to implement $e^{-iH_{f1}\tau_2}$ is
\begin{eqnarray}
 \g_{2}=\g_{2}^r+\g_{2}^c\leq 3N\left(\log_2^2\Lambda+3\log_2\Lambda+1\right)\in O\left(N\log_2^2\Lambda\right).
 \label{eqn:g2}
\end{eqnarray}

Alternatively, we can use qubitization to simulate $e^{-iH_{21}\tau_2}$. The Algorithm-II described in Section \ref{subsec:totalQubit} applies qubitization on the entire Hamiltonian $\hat{H}_{PF}$ and for this we require to block encode $H_{21}$ and that is the main motivation for explaining this here. 

The ancilla preparation sub-routine is defined as follows.
\begin{eqnarray}
    \prep_{f1}\ket{0}^*=\left(\frac{1}{\sqrt{N}}\sum_{q=1}^N\ket{q}\right)\otimes\left(\frac{1}{\sqrt{3}}\sum_{\mu=1}^3\ket{\mu}\right)\otimes\left(\sum_{k=1}^{\frac{\log^2(2\Lambda)+\log(2\Lambda)+2}{2}}\sqrt{\frac{w_k}{\sum_kw_k}}\ket{k}\right).
\end{eqnarray}
Here $w_k$ are the weight of the unitaries in the LCU decomposition of $E^2$, as given in Equation \ref{eqn:E2_lcu}.
In the first $\log_2N$-qubit electric link index register we store the $N$ electric link indices in equal superposition using $\log_2N$ H gates. In the next 2-qubit spin index register we store the values of $\mu$, using 2 H gates. Since $E^2$ is a sum of $\frac{\log_2^2(2\Lambda)+\log_2(2\Lambda)+2}{2}=\frac{\log_2^2\Lambda+3\log_2\Lambda+4}{2}$ unitaries (Equation \ref{eqn:E2_lcu}), so the last register of $\log_2\left(\log_2^2\Lambda+3\log_2\Lambda+4\right)-1$ qubits stores the indices of the unitaries in a superposition, weighted according to Equation \ref{eqn:E2_lcu}. To obtain proper weighting of the basis states we can use any arbitrary state preparation algorithm, for example, \cite{2008_PMH, 2016_NDW, 2021_APPS}, so we require at most $\log_2\left(\log_2^2\Lambda+3\log_2\Lambda+4\right)$ H, $\log_2^2\Lambda+3\log_2\Lambda+3\log_2\left(\log_2^2\Lambda+3\log_2\Lambda+4\right)$ CNOT and $\log_2^2\Lambda+3\log_2\Lambda+2$ rotation gates. 

The unitary selection sub-routine does the following.
\begin{eqnarray}
    \sel_{f1}:\ket{q}\ket{\mu}\ket{k}\left(\bigotimes_{q=1}^N\bigotimes_{\mu'=1}^3\ket{f_e=0}_{q,\mu'}\right)\ket{\phi}\mapsto\ket{q}\ket{\mu}\ket{k}\left(\ket{1} \right)_{q,\mu}\left(E_k^2\right)_{q,\mu}\ket{\phi}
\end{eqnarray}
and it follows in a straightforward manner that
\begin{eqnarray}
    \bra{0}\prep_{f1}^{\dagger}\cdot\sel_{f1}\cdot\prep_{f1}\ket{0}=\frac{H_{f1}}{3N\Lambda^2/2},  \nonumber
\end{eqnarray}
and here too we keep in mind that $\|H_{f1}\|\leq\frac{3N\Lambda^2}{2}$, from Table \ref{tab:norm}.
In each of the $3N$ subspaces we allocate an ancilla $f_e$, initialized to $\ket{0}$ for selection. Using $N$ number of $C^{\log_2N}X$ gates controlled on the link index register, a link subspace is selected. We use $3N$ number of $C^3X$ gates to select spin subspaces. The 2 controls are on the spin register and the last one is controlled on the target qubits of the $C^{\log_2N}X$ gates. In short, 
$\ket{f_e}$ is flipped to $\ket{1}$ if both the link register state $\ket{q}$ and spin register state $\ket{\mu}$ match. The remaining operations are all controlled on the state of $f_e$. Same set of gates are used at the end of the operations for uncomputing. We use the optimization technique of Theorem \ref{thm:CX} to synthesize the $NC^{\log_2N}X$  gates. If we split the ancillae into $M_{f1}$ sets, such that the $i^{th}$ set has $\frac{\log_2N}{r_i}$ qubits, then the total number of multi-controlled-X gates we require is 
\begin{eqnarray}
    \sum_{i=1}^{M_{f1}}N^{\frac{1}{r_i}}C^{\frac{\log_2N}{r_i}}X+NC^{M_{f1}}X+3NC^3X.
\end{eqnarray}
In each of the $3N$ subspaces we use $\frac{\log_2^2\Lambda+3\log_2\Lambda+4}{2}$ number of $C^{\log_2\frac{\log^2\Lambda+3\log\Lambda+4}{2}}X$ gates to select and apply the unitaries in the decomposition of $E^2$ i.e. $\log_2^2(2\Lambda)$ CZ gates. Using the optimization technique of Theorem \ref{thm:CX}, the number of multi-controlled-X gates we require is,
\begin{eqnarray}
    \sum_{i=1}^{M_{f1}'}\left(\frac{\log^2\Lambda+3\log\Lambda+4}{2}\right)^{\frac{1}{r_i}}C^{\frac{\log(\log^2\Lambda+3\log\Lambda+4)-1}{r_i}}X+\frac{\log^2\Lambda+3\log\Lambda+4}{2}C^{M_{f1}'}X,
\end{eqnarray}
where $M_{f1}'$ is a constant. Thus for one block encoding of $H_{f1}$ number of rotation gates required is
\begin{eqnarray}
    \g_2^{r'}&\leq& \log_2^2\Lambda+3\log_2\Lambda+2 ;
\end{eqnarray}
the number of T gates required is
\begin{eqnarray}
    \g_{2}^{t'}&\leq& 4\sum_{i=1}^{M_{f1}}N^{\frac{1}{r_i}}\frac{\log_2N}{r_i}+4NM_{f1}+12N\sum_{i=1}^{M_{f1}'}\left(\frac{\log_2^2\Lambda+3\log_2\Lambda+4}{2}\right)^{\frac{1}{r_i}}\log\left(\log_2^2\Lambda+3\log_2\Lambda+4\right)    \nonumber \\
    &&+6NM_{f1}'\left(\log_2^2\Lambda+3\log_2\Lambda+4\right);
\end{eqnarray}
the number of CNOT gates required is asymptotically same as T gates; and so counting the H and CZ gates the total number of gates required for $(3N\Lambda^2/2,.,0)$-block encoding of $H_{21}$ is,
\begin{eqnarray}
    \g_2'\in O\left(N\log_2^2\Lambda\right)    ;
\end{eqnarray}
assuming $M_{f1}, M_{f1}'$ are constants. We have deliberately skipped the middle argument while specifying the block encoding constant of $H_{21}$. This argument basically denotes the number of extra ancilla we require in the $\prep_{f1}$ sub-routine and we do not require this for our gate complexity. So for simplicity and convenience, we have dropped it and we will do so henceforth. Sometimes, we will be even more crisp and simply say 'block-encoding of $\frac{H_{21}}{3N\Lambda^2/2}$'. So we require 
\begin{eqnarray}
    O\left(N\Lambda^2\tau_2+\frac{\log(1/\delta_{21})}{\log\log(1/\delta_{21})}\right) \nonumber
\end{eqnarray}
calls to the $\prep_{f1}$ and $\sel_{f1}$ oracles in order to implement an $\delta_{21}$-precise block encoding of $e^{-iH_{21}\tau_2}$.Thus the number of gates required for simulating $e^{-iH_{21}\tau_2}$ is
\begin{eqnarray}
    \g_{21}\in O\left(N^2\Lambda^2\log^2\Lambda\tau_2+\frac{\log(1/\delta_{21})}{\log\log(1/\delta_{21})}N\log^2\Lambda\right). \nonumber
\end{eqnarray}

\subsubsection{Algorithm to simulate $e^{-iH_{31}\tau_3}$}
\label{subsec:H31}

We block encode $H_{31}=H_s+H_{3\pi}$ in a recursive manner, using Theorem \ref{thm:blockEncodeDivConq} repeatedly. 
\begin{eqnarray}
 \text{Let }H_s^{j,q}&=&-\sum_{\mu\neq\nu\neq\xi=1}^3\left(\sigma_{j,\mu}\otimes\id\right)\otimes\left(\nabla_{\nu}A_{q,\xi}-\nabla_{\xi}A_{q,\nu}\right) \nonumber \\
 &=&\sum_{\mu\neq\nu\neq\xi=1}^3\left(\sigma_{j,\mu}\otimes\id\right)\otimes\left(\nabla_{\xi}A_{q,\nu}-\nabla_{\nu}A_{q,\xi}\right)   \nonumber \\
 \text{and } H_{3\pi}^{j,q}&=&\sum_{\mu=1}^3\id\otimes\id\otimes A_{q,\mu}^2, \nonumber \\
 \text{such that }\quad H_{31}^{j,q}&=&\frac{1}{c}H_s^{j,q}+\frac{1}{2c^2}H_{3\pi}^{j,q},\qquad H_{31}^j=\sum_{q=1}^NH_{31}^{j,q},\qquad H_{31}=\sum_{j=1}^{\eta}H_{31}^j.
\end{eqnarray}

\paragraph{Block encoding of $H_s^{j,q}$ : }
The ancillae preparation sub-routine, denoted by $\prep_s^{j,q}$ does the following.
\begin{eqnarray}
\prep_s^{j,q}\ket{0}^*
&=&\left(\frac{1}{\sqrt{6}}\sum_{\mu\neq\nu\neq\xi}^3\sum_{b=0}^1\ket{\mu}\ket{\nu}\ket{\xi}\ket{b}\right) \otimes\left(\sum_{k=-a}^{a}\sqrt{\frac{|d_{2a+1,k}'|}{\sum_k|d_{2a+1,k}'|}}\ket{k+a}\right)  \nonumber \\
&\otimes&\left(\sum_{k'=1}^{\log_2d}\sqrt{\frac{w_{k'}'}{\sum_{k'}w_{k'}'}}\ket{k'}\right)
\end{eqnarray}
 The first $(2\times 3+1)=7$-qubit spin index register stores directions or spins in equal superposition and we need $7$ H gates for this. If $\mu,\nu,\xi$ are not unequal or any of them is greater than 3 then we discard the computational path. The last qubit of this register selects between $\nabla_{\nu}A_{q,\xi}$ and $\nabla_{\xi}A_{q,\nu}$. The second and third registers with $\log_2(2a)$ and $\log_2\log_2d$ qubits, respectively, indicate which adder to apply or on which qubit Z-gate should be applied. These are unitaries obtained in the LCU decomposition of $\nabla$ (Lemma \ref{lem:lcuNabla}) and $A$ (Corollary \ref{cor:A_lcu}) in Section \ref{app:lcu}. To obtain proper weighting of the basis states we require at most $2a+\log_2d-4$ rotation gates, $2a+\log_2d+3\log_2(2a\log_2d)-14$ CNOT and $\log_2(2a\log_2d)$ H gates \cite{2016_NDW}. 

We denote the next sub-routine by $\sel_s^{j,q}$, which is described as follows.
\begin{eqnarray}
&& \sel_s^{j,q}:\ket{\mu,\nu,\xi,0}\ket{k''}\ket{k'}\ket{\phi} \nonumber \\
&\mapsto&\ket{\mu,\nu,\xi,0}\ket{k''}\ket{k'}\left(\sigma_{\mu}\otimes\id\right)_j\left(\nabla_{k''}\right)_{q,\nu}\left(A_{k'}\right)_{q,\xi}\ket{\phi}
\end{eqnarray}
Controlled on $\ket{\mu}$, we apply $\sigma_{\mu}$ on the spin subspace of the $j^{th}$ particle. Controlled on $\ket{\nu,\xi}$ we select spin-subspaces of the $q^{th}$ link register. This step require $O(1)$ gates. Controlled on $\ket{k''}$ and $\ket{k'}$ we apply the $k''^{th}$ and $k'^{th}$ unitary in the LCU decompositions of $\nabla$ and $A$, respectively. If the third qubit in the spin register is $\ket{1}$ then we apply $\nabla_{\xi}$ and $A_{\nu}$. All the unitaries in decomposition of $\nabla$ and $A$ act on $\log_2d$ qubits and they are controlled on $\log_2(2a)$ and $\log_2\log_2d$ qubits respectively. $A$ is a sum of $\log_2d$ Z gates, thus to select and implement these unitaries we require $\log_2d$ compute-uncompute pairs of $C^{\log_2\log_2d}X$ gates and $\log_2d$ CZ gates. Similarly we can use $2a$ pairs of $C^{\log_2(2a)}X$ gates and $2a$ single-controlled adders to select and implement the unitaries in the LCU decomposition of $\nabla$. Using Theorem \ref{thm:CX}, we find that the number of pairs of multi-controlled-X gates we require is 
\begin{eqnarray}
    \sum_{i=1}^{M_{s1}}(\log_2d)^{\frac{1}{r_i}}C^{\frac{\log_2\log_2d}{r_i}}X+\log_2dC^{M_{s1}}X+\sum_{i=1}^{M_{s2}}(2a)^{\frac{1}{r_i}}C^{\frac{\log_2(2a)}{r_i}}X+(2a)C^{M_{s2}}X
\end{eqnarray}
where we have split the $\log_2\log_2d$ control qubits for $A$ into $M_{s1}$ sets and the $\log_2(2a)$ control qubits for $\nabla$ into $M_{s2}$ sets. 

It follows that
\begin{eqnarray}
    \bra{0}\prep_{s}^{j,q\dagger}\cdot\sel_{s}^{j,q}\cdot\prep_{s}^{j,q}\ket{0}=\frac{H_{s}^{j,q}}{12\pi\ln 2a^2/h\Delta},  \nonumber
\end{eqnarray}
and thus we have a $(12\pi\ln 2a^2/h\Delta,.,0)$-block encoding of $H_s^{j,q}$.
  
\paragraph{Block encoding of $H_{3\pi}^{j,q}$ :} The first ancillae preparation sub-routine is described as follows.
\begin{eqnarray}
&& \prep_{3\pi}^{j,q}\ket{0}^{*}=
\left(\frac{1}{\sqrt{3}}\sum_{\mu'=1}^3\ket{\mu'}\right)\otimes\left(\sum_{k=1}^{\frac{\log^2d+\log d}{2}}\sqrt{\frac{w_k'}{\sum_{k}w_k'}}\ket{k}\right)
\end{eqnarray}
The first 2-qubit register is the spin index register. Since $A^2$ is a sum of $\frac{\log_2^2d+\log_2d}{2}$ unitaries (Table \ref{tab:norm}), so we prepare a $\log_2\left(\frac{\log_2^2d+\log_2d}{2}\right)$-qubit register in a superposition weighted according to the LCU decomposition of $A^2$ (Corollary \ref{cor:A2_lcu} in Section \ref{app:lcu}) and this can be done with $\log_2\frac{\log_2^2d+\log_2d}{2}$ H, $\log_2^2d+\log_2d+3\log_2\left(\frac{\log_2^2d+\log_2d}{2}\right)-7$ CNOT and $\log_2^2d+\log_2d-2$ rotation gates. 

The next sub-routine is described as follows.
\begin{eqnarray}
&& \sel_{3\pi}^{j,q}:\ket{\mu'}\ket{k}\ket{\phi}
\mapsto\ket{k}\left(A_k^2\right)_{q,\mu'}\ket{\phi}
\end{eqnarray}
To implement $A^2$, controlled on $\ket{k}$ register, we require $\frac{\log_2^2d+\log_2d}{2}$ pairs of $C^{\log_2\frac{\log^2d+\log d}{2}}X$ and $\log_2d+2\frac{(\log_2d-1)\log_2d}{2}=\log_2^2d$ CZ gates. For the latter, we have taken into account single Z gates and $ZZ$ operators, appearing in the LCU decomposition of $A^2$ (Corollary \ref{cor:A2_lcu}). Using Theorem \ref{thm:CX}, we find that the number of pairs of multi-controlled-X gates we require is 
\begin{eqnarray}
    \sum_{i=1}^{M_{3\pi}}(\frac{\log_2^2d+\log_2d}{2})^{\frac{1}{r_i}}C^{\frac{\log_2\frac{\log_2^2d+\log_2d}{2}}{r_i}}X+\frac{\log_2^2d+\log_2d}{2}C^{M_{3\pi}}X
\end{eqnarray}
where we have split the $\log_2\frac{\log_2^2d+\log_2d}{2}$ control qubits into $M_{3\pi}$ sets.

It follows that
\begin{eqnarray}
    \bra{0}\prep_{3\pi}^{j,q\dagger}\cdot\sel_{3\pi}^{j,q}\cdot\prep_{3\pi}^{j,q}\ket{0}=\frac{H_{3\pi}^{j,q}}{24\pi^2/\Delta^2},  \nonumber
\end{eqnarray}
and thus we have a $(24\pi^2/\Delta^2,.,0)$-block encoding of $H_{3\pi}^{j,q}$.

\paragraph{Block encoding of $H_{31}$ : }  We use Theorem \ref{thm:blockEncodeDivConq} repeatedly. First we block encode $H_{31}^{j,q}=\frac{1}{c}H_s^{j,q}+\frac{1}{2c^2}H_{3\pi}^{j,q}$ with $O(1)$ extra gate cost. Next we consider $H_{31}^{j}=\sum_{q=1}^NH_{31}^{j,q}$, where each of the summand Hamiltonians act on separate link registers. So we can prepare $\log_2N$ ancilla qubits in an equal superposition of all the link indices using $\log_2N$ H gates. Similarly for $H_{31}=\sum_{j=1}^{\eta}H_{31}^j$, we prepare $\log_2\eta$ qubits in an equal superposition of all the $\eta$ indices with $\log_2\eta$ H gates. Thus the overall ancilla preparation sub-routine is,
\begin{eqnarray}
    \prep_{31}\ket{0}^* &=&
    \left(\frac{1}{\sqrt{\eta}}\sum_{j=1}^{\eta}\ket{j}\right)\otimes\left(\frac{1}{\sqrt{N}}\sum_{q=1}^{N}\ket{q}\right) 
    \otimes\left(\sqrt{\frac{\lambda_s}{c\nconst}}\ket{0}+\sqrt{\frac{\lambda_{3\pi}}{2c^2\nconst}}\ket{1}\right)   \nonumber \\
    &&\otimes\prep_{s}^{j,q}\otimes\prep_{3\pi}^{j,q},
\end{eqnarray}
where $\lambda_s=\|H_s^{j,q}\|=\frac{12\pi\ln 2a^2}{h\Delta}$, $\lambda_{3\pi}=\|H_{3\pi}^{j,q}\|=\frac{24\pi^2}{\Delta^2}$ and $\nconst=\frac{\lambda_s}{c}+\frac{\lambda_{3\pi}}{2c^2}$. The overall unitary selection sub-routine is as follows. 
\begin{eqnarray}
    \sel_{31}:\ket{j,q,0}\ket{\mu,\nu,\xi,b,k'',k'}\ket{\mu',k}\ket{\phi}&\mapsto&\ket{j,q,0}\ket{\mu',k}\sel_s^{j,q}\left(\ket{\mu,\nu,\xi,b,k'',k'}\ket{\phi}\right)   \nonumber \\
    \sel_{31}:\ket{j,q,1}\ket{\mu,\nu,\xi,b,k'',k'}\ket{\mu',k}\ket{\phi}&\mapsto&\ket{j,q,1}\ket{\mu,\nu,\xi,b,k'',k'}\sel_{3\pi}^{j,q}\left(\ket{\mu',k}\ket{\phi}\right)   \nonumber
\end{eqnarray}
It is straightforward to check that
\begin{eqnarray}
    \bra{0}\prep_{31}^{\dagger}\cdot\sel_{31}\cdot\prep_{31}\ket{0}=\frac{H_{31}}{\eta N\nconst},  \nonumber
\end{eqnarray}
where $\eta N\nconst=\frac{12\pi\eta N\ln 2a^2}{ch\Delta}+\frac{12\pi^2\eta N}{c^2\Delta^2}$, which is also the sum of the norms of the Hamiltonians $H_s$ and $H_{3\pi}$ in Table \ref{tab:norm}. Thus we have a $(\eta N\nconst,.,0)$-block encoding of $H_{31}$.

Using $\eta$ pairs of $C^{\log_2\eta}X$ gates we select a particle register by flipping a qubit initialized to $\ket{0}$. Also, using $N$ pairs of $C^{\log_2N}X$ gates we select a link register by flipping another qubit.  Thus using Theorem \ref{thm:CX} we find that the number of pairs of multi-controlled-X gates we require is 
\begin{eqnarray}
    \sum_{i=1}^{M_{31}}N^{\frac{1}{r_i}}C^{\frac{\log_2N}{r_i}}X+NC^{M_{31}}X+\sum_{i=1}^{M_{31}'}\eta^{\frac{1}{r_i}}C^{\frac{\log_2\eta}{r_i}}X+\eta C^{M_{31}'}X
\end{eqnarray}
where we have split the $\log_2N$ and $\log_2\eta$ control qubits into $M_{31}$ and $M_{31}'$ sets, respectively. In this case the unitaries in the decomposition of $\nabla, A, A^2$ have 3 controls and they are applied on each of the $3N$ link subspace. $\sigma$ with 3 controls, are applied on each of the $\eta$ particle subspace.
Overall, we require $3N\cdot 2a=6aN$ number of 3-qubit-controlled $\log_2d$-qubit adders, which can be decomposed as $6aN$ 1-qubit-controlled adders and $12aN$ Toffoli pairs. Using the construction in \cite{2018_G}, we require $6aN\cdot 4(\log_2d-1)$ controlled-T, $6aN\cdot (5\log_2d-4)$ controlled-CNOT to implement the controlled adders. There are other constructions of adders, for example, \cite{2000_D, 2017_RG, 2004_CDKM} and usually there are trade-offs between these constructions. We have taken the estimates from \cite{2018_G}, because of better bound on T-gate cost. We also require $3N\cdot (\log_2d+\log_2^2d)$ Z gates (for $A, A^2$), each controlled on 3 qubits. Each of these can be decomposed as CZ and 2 Toffoli pairs. Also, we require $3\eta$ Paulis, each controlled on 3-qubits.  

Hence, for block encoding of $\frac{H_{31}}{\eta N\nconst}$ the number of controlled rotations required is
\begin{eqnarray}
    \g_{31}^r\leq 2a+2\log_2d+\log_2^2d,
\end{eqnarray}
the number of T-gates required is
\begin{eqnarray}
    \g_{31}^t&\leq& 4\sum_{i=1}^{M_{31}}N^{\frac{1}{r_i}}\frac{\log_2N}{r_i}+4NM_{31}+4\sum_{i=1}^{M_{31}'}\eta^{\frac{1}{r_i}}\frac{\log_2\eta}{r_i}+4\eta M_{31}  \nonumber \\ 
    &&+12N\sum_{i=1}^{M_{s1}}(\log_2d)^{\frac{1}{r_i}}\frac{\log_2\log_2d}{r_i}+12N\log_2dM_{s1}+12N\sum_{i=1}^{M_{s2}}(2a)^{\frac{1}{r_i}}\frac{\log_2(2a)}{r_i}+24aNM_{s2}    \nonumber \\
    &&+12N\sum_{i=1}^{M_{3\pi}}\left(\frac{\log_2^2d+\log_2d}{2}\right)^{\frac{1}{r_i}}\frac{\log_2\frac{\log_2^2d+\log_2d}{2}}{r_i}+6N\left(\log_2^2d+\log_2d\right)M_{3\pi}    \nonumber \\
    &&+24aN\log_2d-24aN+24aN+12N(\log_2d+\log_2^2d)+6aN(5\log_2d-4)
\end{eqnarray}
while the number of CNOT gates required is some constant times $\g_{31}^t$. Counting the rotation, H, CZ and other gates, the total number of gates required for the block encoding of $\frac{H_{31}}{\eta N\nconst}$ is 
\begin{eqnarray}
    \g_{31}'&\in& O\left(\eta+N(a+\log_2d)\log_2d\right), \label{eqn:g31'}
\end{eqnarray}
assuming $M_{s1}, M_{s2}, M_{3\pi}, M_{31}, M_{31}'$ are constants and each $\frac{1}{r_i}\leq\frac{1}{2}$. From Table \ref{tab:norm}, 
\begin{eqnarray}
    \|H_{31}\|&\leq&\frac{12\pi^2\eta N}{c^2\Delta^2}+\frac{12\pi\eta N\ln (2a^2)}{ch\Delta}=\frac{12\pi\eta N}{c\Delta^2}\left(\frac{\pi}{c}+\frac{\Delta\ln (2a^2)}{h}\right)  \nonumber \\
 &\leq&\frac{K_{31}\eta N\ln (2a^2)}{\Delta^2} \qquad [K_{31}=\text{constant}],    
 \end{eqnarray}
 where we assumed that $h\leq K_h\Delta$, for some constant $K_h$. So we need
\begin{eqnarray}
   R_{31}\in O\left(\frac{\eta N\ln(2a^2)}{\Delta^2}\tau_3+\frac{\log(1/\delta_{31})}{\log\log(1/\delta_{31})}\right)
\end{eqnarray}
calls to the block encoding of $\frac{H_{31}}{\eta N\nconst}$ in order to implement an $\delta_{31}$-precise block encoding of $e^{-iH_{31}\tau_3}$ \cite{2019_GSLW}. Thus the number of gates required for simulating $e^{-iH_{31}\tau_3}$ is 
\begin{eqnarray}
    \g_{31}&\in& O(R_{31}\cdot\g_{31}') \nonumber \\
    &\in&O\left(\frac{\eta^2N\ln (2a^2)}{\Delta^2}\tau_3+\frac{\eta N^2\ln (2a^2)\log d}{\Delta^2}(a+\log d)\tau_3\right. \nonumber \\
    &&\left.+\frac{\log (1/\delta_{31})}{\log\log (1/\delta_{31})}\left(\eta+N(a+\log d)\log d\right)\right)
    \label{eqn:g31}
\end{eqnarray}

\subsubsection{Algorithm to simulate $e^{-iH_{32}\tau_3}$}
\label{subsec:H32}

We know that $H_{32}=H_V+H_{1\pi}+H_{2\pi}$ and here also we use Theorem \ref{thm:blockEncodeDivConq} to block encode in a recursive manner. We define the following.
\begin{eqnarray}
   && H_{1\pi}^{j,q,\mu}=-\id\otimes\nabla_{j,\mu}^2\otimes\id,\qquad H_{2\pi}^{j,q,\mu}=\id\otimes\left(i\nabla_{j,\mu}\right)\otimes A_{q,\mu}  \nonumber \\
   &&H_{12\pi}^{j,q,\mu}=\frac{1}{2}H_{1\pi}^{j,q,\mu}+\frac{1}{c}H_{2\pi}^{j,q,\mu},\qquad H_{12\pi}=\sum_{j=1}^{\eta}\sum_{q=1}^N\sum_{\mu=1}^3H_{12\pi}^{j,q,\mu}   \nonumber
\end{eqnarray}

\paragraph{Block encoding of $H_{12\pi}$ : } As in the case of $H_{31}$, we first block encode $H_{1\pi}^{j,q,\mu}$ and $H_{2\pi}^{j,q,\mu}$ separately using the ancillae preparation sub-routines $\prep_{1\pi}^{j,q,\mu}$ and $\prep_{2\pi}^{j,q,\mu}$ respectively, followed by the unitary selection sub-routines $\sel_{1\pi}^{j,q,\mu}$ and $\sel_{2\pi}^{j,q,\mu}$ respectively. Then we block encode $H_{12\pi}^{j,q,\mu}$ and $H_{12\pi}$, as discussed in Theorem \ref{thm:blockEncodeDivConq}. Whenever the same Hamiltonian is applied on disjoint spaces we apply the optimization described in Remark \ref{remark:divConqBlock}. Thus our overall ancillae preparation sub-routine is as follows.
\begin{eqnarray}
    \prep_{12\pi}\ket{0}^* 
 &=&\left(\frac{1}{\sqrt{\eta}}\sum_{j=1}^{\eta}\ket{j}\right)\otimes\left(\frac{1}{\sqrt{N}}\sum_{q=1}^N\ket{q}\right)\otimes\left(\frac{1}{\sqrt{3}}\sum_{\mu=1}^3\ket{\mu}\right) \nonumber \\
  &&\otimes\left(\sqrt{\frac{\lambda_1}{2\nconst'}}\ket{0}+\sqrt{\frac{\lambda_2}{c\nconst'}}\ket{1}\right)  
   \otimes\prep_{1\pi}^{j,q,\mu}\otimes\prep_{2\pi}^{j,q,\mu}, 
\end{eqnarray}
  where $\lambda_1=\|2H_{1\pi}\|,\lambda_2=\|cH_{2\pi}\|,\nconst'=\frac{\lambda_1}{2}+\frac{\lambda_2}{c}=\frac{8\pi^2\eta N}{h^2}+\frac{12\pi\eta N\ln 2a^2}{ch\Delta}$ and
  \begin{eqnarray}
   \prep_{1\pi}^{j,q,\mu}\ket{0}^* &=& \left(\sum_{k=-a}^a\sqrt{\frac{|d_{2a+1,k}|}{\sum_{k}|d_{2a+1,k}|}}\ket{k+a}\right);  \\
\prep_{2\pi}^{j,q,\mu}\ket{0}^* &=& \left(\sum_{k_1=-a}^a\sqrt{\frac{|d_{2a+1,k_1}''|}{\sum_{k_1}|d_{2a+1,k_1}''|}}\ket{k_1+a}\right)\otimes\left(\sum_{k_2=1}^{\log_2d}\sqrt{\frac{w_{k_2}}{\sum_{k_2}w_{k_2}}}\ket{k_2}\right) 
\end{eqnarray}
We use $\log_2\eta$, $\log_2N$ and 2 H gates to prepare an equal superposition of $\eta$ particle indices, $N$ link indices and $3$ spins in the first, second and third register respectively. In the fourth register we require 2 rotations. $\prep_{1\pi}^{j,q,\mu}$ acts on the $\approx\log_2(2a)$-qubit fifth register where we store the indices of the adders in the decomposition of $\nabla^2$ (Lemma \ref{lem:lcuNabla2}) with appropriate weights. This can be done using $\log_2(2a)$ H, $4a+3\log_2(2a)-7$ CNOT and $4a-2$ rotation gates. $\prep_{2\pi}^{j,q,\mu}$ acts on the last two registers. The second last one has $\log_2(2a)$ qubits and stores the indices of the adders in the LCU decomposition of $\nabla$ (Lemma \ref{lem:lcuNabla}). We observe that we work with $i\nabla$ because it is Hermitian and this factor is adjusted in the weights. The last register has $\log_2\log_2d$ qubits and stores the indices of the Z gates occurring in the LCU decomposition of $A$ (Corollary \ref{cor:A_lcu}). To prepare these superposition we require $\log_2(2a)+\log_2\log_2d=\log_2(2a\log_2d)$ H, $(4a+3\log_2(2a)-7)+(2\log_2d+3\log_2\log_2d-7)=4a+2\log_2d+3\log_2(2a\log_2d)-14$ CNOT and $(4a-2)+(2\log_2d-2)=4a+2\log_2d-4$ rotation gates. 

The overall unitary selection sub-routine is as follows.
\begin{eqnarray}
&&\sel_{1\pi}^{j,q,\mu}:\ket{k'}\ket{\phi}
\mapsto\ket{k'}\left(\id\otimes\nabla_{k'}^2\right)_{j,\mu}\ket{\phi}    \label{sel:1pijqmu} \\
&&\sel_{2\pi}^{j,q,\mu}:\ket{k_1'}\ket{k_2'}\ket{\phi}\mapsto\ket{k_1'}\ket{k_2'}\left(\nabla_{k_1'}\right)_{j,\mu}\left(A_{k_2'}\right)_{q,\mu}\ket{\phi} \label{sel:2pijqmu} \\
   && \sel_{12\pi}:\ket{j,q,\mu,0}\ket{k'}\ket{k_1',k_2'}\ket{\phi}   
   \mapsto\ket{j,q,\mu,0}\ket{k_1',k_2'}\sel_{1\pi}^{j,q,\mu}\left(\ket{k'}\ket{\phi}\right)    \nonumber \\
  && \sel_{12\pi}:\ket{j,q,\mu,1}\ket{k'}\ket{k_1',k_2'}\ket{\phi} 
  \mapsto\ket{j,q,\mu,1}\ket{k'}\sel_{2\pi}^{j,q,\mu}\left(\ket{k_1',k_2'}\ket{\phi}\right)  \nonumber 
\end{eqnarray}
Using $\eta$ pairs of $C^{\log_2\eta}X$ gates and 3 pairs of $C^2X$ gates we select a particle-spin register by flipping a qubit initialized to $\ket{0}$. Using $N$ pairs of $C^{\log_2N}X$ gates we select a link subspace by flipping another qubit. Due to $H_{1\pi}$, in each of the $3\eta$ registers we apply controlled $\nabla^2$ operator (Equation \ref{sel:1pijqmu}), which is a sum of $\approx 2a$ adders, each acting on $\log_2N^{1/3}=\frac{1}{3}\log_2N$ qubits, controlled on $\log_2(2a)$ qubits. Thus we require $\approx 3\eta\cdot 2a$ pairs of $C^{\log_2(2a)}$ X-gates, $3\eta\cdot 2a$ controlled adders. Due to $H_{2\pi}$, we apply controlled $\nabla$ and $A$ operators (Equation \ref{sel:2pijqmu}) in each of the $3\eta$ and $3N$ particle and link registers respectively. $\nabla$ is a sum of $2a$ adders, each acting on $\frac{1}{3}\log_2N$ qubits, controlled on $\log_2(2a)$ qubits. $A$ is a sum of $\log_2d$ Z gates, each controlled on $\log_2d\log_2d$ qubits. So, here we require $3\eta\cdot 2a$ pairs of $C^{\log_2(2a)}X$ gates, $3\eta\cdot 2a$ controlled adders, $3N\cdot\log_2d$ pairs of $C^{\log_2\log_2d}$X gates and $3N\cdot\log_2d$ CZ gates. We can implement the controlled adders using $12\eta a\cdot 4(\frac{1}{3}\log_2N-1)$ controlled-T, $12\eta a\cdot (\frac{5}{3}\log_2N-4)$ controlled-CNOT.

Using Theorem \ref{thm:CX} we find that the  number of pairs of multi-controlled-X gates we required is
\begin{eqnarray}
    &&\sum_{i=1}^{M_1}N^{\frac{1}{r_i}}C^{\frac{\log_2N}{r_i}}X+NC^{M_1}X+\sum_{i=1}^{M_2}\eta^{\frac{1}{r_i}}C^{\frac{\log_2\eta}{r_i}}X+\eta C^{M_2}X+3C^2X    \nonumber \\
    &+&3\eta \left(\sum_{i=1}^{M_3}(2a)^{\frac{1}{r_i}}C^{\frac{\log_2(2a)}{r_i}}X+2aC^{M_3}X\right)+3N\left(\sum_{i=1}^{M_4}(\log_2d)^{\frac{1}{r_i}}C^{\frac{\log_2\log_2d}{r_i}}X+\log_2dC^{M_4}X\right) \nonumber.
\end{eqnarray}
It can be verified in a straightforward manner that
\begin{eqnarray}
    \bra{0}\prep_{12\pi}^{\dagger}\cdot\sel_{12\pi}\cdot\prep_{12\pi}\ket{0}=\frac{H_{12\pi}}{\nconst'},  \nonumber
\end{eqnarray}
where $\nconst'=\frac{8\pi^2\eta N}{h^2}+\frac{12\pi\eta N\ln 2a^2}{ch\Delta}$, which is also the sum of the norms of $H_{1\pi}$ and $H_{2\pi}$ (Table \ref{tab:norm}). 

So, for block encoding of $\frac{H_{12\pi}}{\nconst'}$ the number of rotation gates required is
\begin{eqnarray}
    \g_{12\pi}^r\leq 8a+2\log_2d,
\end{eqnarray}
the number of T-gates required is
\begin{eqnarray}
    \g_{12\pi}^t&\leq& 4\sum_{i=1}^{M_1}N^{\frac{1}{r_i}}\frac{\log_2N}{r_i}+4NM_1+4\sum_{i=1}^{M_2}\eta^{\frac{1}{r_i}}\frac{\log_2\eta}{r_i}+4\eta M_2    \nonumber \\
    &&+12\eta\left(\sum_{i=1}^{M_3}(2a)^{\frac{1}{r_i}}\frac{\log_2(2a)}{r_i}+2aM_3\right)+12N\left(\sum_{i=1}^{M_4}(\log_2d)^{\frac{1}{r_i}}\frac{\log_2\log_2d}{r_i}+\log_2dM_4\right)    \nonumber \\
    &&+16\eta a\log_2N
\end{eqnarray}
while the number of CNOT gates is a constant times $\g_{12\pi}^c$, and hence the total number of gates is
\begin{eqnarray}
    \g_{12\pi}'\in O\left(\eta a\log_2N+N\log_2d\right)
\end{eqnarray}
where we assumed that each $\frac{1}{r_i}\leq\frac{1}{2}$ and $M_1,M_2,M_3,M_4$ are constants.

\paragraph{Block encoding of $H_V$ : } We know that $H_V=H_{V_{ee}}+H_{V_{ne}}$ and we block encode it, following the approach taken in \cite{2019_BBMN, 2021_SBWetal}, with some modifications and incorporating the optimizations in Theorem \ref{thm:CX}. The ancilla preparation sub-routine is as follows.
\begin{eqnarray}
    \prep_V\ket{1}\ket{0}^*&\propto&\ket{0}\sum_{i<j}^{\eta}\sum_{v_x,v_y,v_z=-N^{1/3}}^{N^{1/3}}\frac{1}{\|\vect{v}\|_2}\ket{i}\ket{j}\ket{v_x,v_y,v_z}    \nonumber \\
 &&-\ket{1}\sum_{i=1}^{\eta}\sum_{\kappa=1}^K\sum_{v_x,v_y,v_z=-N^{1/3}}^{N^{1/3}}\frac{\sqrt{Z_{\kappa}}}{\|\vect{v}\|_2}\ket{i}\ket{\kappa}\ket{v_x,v_y,v_z}
\end{eqnarray}
We apply a H gate on the first ancilla, initialized to $\ket{1}$. The resulting state $\frac{1}{\sqrt{2}}\left(\ket{0}-\ket{1}\right)$ is used to select between the two Hamiltonians - $\ket{0}$ for $H_{V_{ee}}$ and $\ket{1}$ for $H_{V_{ne}}$. Also, the -1 phase of $H_{V_{ne}}$ is taken care of at this stage. Next we have a $\log_2\eta$-qubit register, where we store the particle indices in equal superposition using $\log_2\eta$ H gates. The next register is also of $\log_2\eta$-qubits (assuming $K\leq\eta$). If the first particle index register is $\ket{0}$, then we prepare the second register in equal superposition over particle indices and this requires $\log_2\eta$ H gates. We impose the constraint $i\geq j$ by flagging a qubit, in which case we discard the computational path. If the first qubit is $\ket{1}$ then we prepare the second register in a superposition over $\ket{\kappa}$ (positions of neutrons), weighted by nuclear $\sqrt{Z_{\kappa}}$, the nuclear charge. This is given by a classical database with complexity $O(K)$. We can use the QROM and subsampling strategies, discussed in \cite{2018_BGBetal}. We assume that $K\leq \eta$. For a material, in practice, there will be a limited number of nuclear charges with nuclei in a regular array , so this complexity instead will be $O(\log_2K)$. We follow the state preparation procedure, described in \cite{2019_BBMN}, to prepare $\sum_{v_x,v_y,v_z=-N^{1/3}}^{N^{1/3}}\frac{1}{\|v\|_2}\ket{\vect{v}}$. This has been described in Appendix \ref{app:statePrep}. The overall complexity obtained is $O\left(\log_2N\log_2\frac{N}{\delta'}+\log_2\eta\right)$, where $\delta'$ is an upper bound on the tolerable error for the block encoding of $H_V$. If a full classical database for the nuclei is required, then the complexity will have an additional factor of $O\left(K\log\frac{1}{\delta''}\right)$, where $\delta''$ is the relative precision with which the positions of the nuclei are specified.

The unitary selection sub-routine is described as follows.
\begin{eqnarray}
&& \sel_V:\ket{0}\ket{i}\ket{j}\ket{\vect{v}}\ket{\vect{q}_1,\ldots \vect{q}_i,\ldots \vect{q}_j,\ldots \vect{q}_{\eta}}\ket{0}  \nonumber\\
&&\mapsto\ket{0}\ket{i}\ket{j}\ket{\vect{v}}\ket{\vect{q}_1,\ldots \vect{q}_i,\ldots \vect{q}_j,\ldots \vect{q}_{\eta}}\ket{\vect{q}_i-\vect{q}_j} \nonumber \\
&& \sel_V:\ket{1}\ket{i}\ket{\kappa}\ket{\vect{v}}\ket{\vect{q}_1,\ldots \vect{q}_i,\ldots \vect{q}_{\eta}}\ket{0}  \nonumber\\
&&\mapsto\ket{1}\ket{i}\ket{\kappa}{\vect{v}}\ket{\vect{q}_1,\ldots \vect{q}_i,\ldots  \vect{q}_{\eta}}\ket{\vect{R}_{\kappa}-\vect{q}_i}
\end{eqnarray}
If the first qubit is $\ket{0}$ we discard if $\vect{q}_i-\vect{q}_j\neq \vect{v}$, i.e. we flag this state as failure and perform identity along this computational path. Since for each pair of $\vect{q}_i, \vect{q}_j$, only one value of $\vect{v}$ survives, so the probability distribution is unaffected. We can use $\eta$ pairs of $C^{\log_2\eta}X$ gates to select the particle registers. It takes $O(\log_2 N)$ gates for comparing and calculating difference of the position co-ordinates.

If the first register is $\ket{1}$ then we do the following. We use a classical database to access $\vect{R}_{\kappa}$ and this has complexity $O\left(K\right)$. With $\eta$ pairs of $C^{\log_2\eta}X$ gates, we select the particle, controlled on the particle index register. We take the difference $\vect{R}_{\kappa}-\vect{q}_i$ and discard the computational path if it is not equal to $\vect{v}$. This step has complexity $O(\log_2N)$.

Thus we obtain a block encoding of $\frac{H_V}{\lambda_V}$, where $\lambda_V=\|H_V\|=\frac{\eta(\eta-1)}{2\Delta^2}+\frac{\eta Z_{sum}}{\Delta^2}$ and incorporating the optimizations of Theorem \ref{thm:CX}, the total number of gates required is
\begin{eqnarray}
 \g_V'\in O\left(\eta+\log_2N\log_2\frac{N}{\delta'}+K\log_2\frac{1}{\delta''}\right).
\end{eqnarray}

\paragraph{Block encoding of $H_{32}$ : } Since $H_V$ has a probabilistic ancilla preparation sub-routine, we can block encode $H_{32}=H_V+H_{12\pi}$ using the procedure described in \cite{2021_SBWetal}, by repeating the $\prep_V$ sub-routine constant number of times. This does not change the asymptotic gate complexity. Thus total number of gates required to encode $\frac{H_{32}}{\lambda_{32}}$, where $\lambda_{32}=\|H_V\|+\|H_{12\pi}\|=\frac{\eta(\eta-1)}{2\Delta^2}+\frac{\eta Z_{sum}}{\Delta^2}+\frac{8\pi^2\eta N}{h^2}+\frac{12\pi\eta N\ln 2a^2}{ch\Delta}$.
\begin{eqnarray}
\g_{32}' &\in& O\left(\eta a\log_2N+N\log_2d+\log_2N\log_2\frac{N}{\delta'}+K\log_2\frac{1}{\delta''}\right) \label{eqn:g32'}
\end{eqnarray}
Now, we can bound the sum of the $\ell_1$ norm of $H_V$, $H_{1\pi}$ and $H_{2\pi}$, (and hence $\lambda_{32}$) as follows.
\begin{eqnarray}
&& \frac{12\pi\eta N\ln (2a^2)}{ch\Delta}+\frac{8\pi^2\eta N}{h^2}+\frac{\eta(\eta-1)}{2\Delta^2}+\frac{\eta Z_{sum}}{\Delta^2}   \nonumber \\
&\leq&\frac{\eta N}{\Delta^2}\left(\frac{12\pi\ln (2a^2)}{c}\cdot\frac{\Delta}{h} + \frac{8\pi\Delta^2}{h^2}+\left(\frac{\eta+2Z_{sum}}{2N}\right)\right)  \nonumber \\
&\lessapprox& K_{32}\frac{\eta N}{\Delta^2}\left(1+\frac{\eta_s}{N}\right),\nonumber
\end{eqnarray}
where $K_{32}$ is a constant and $\eta_s=\eta+2Z_{sum}$, $Z_{sum}=\sum_{\kappa=1}^K|Z_{\kappa}|$. Thus to obtain an $\delta_{32}$-precise implementation of $e^{-i(H_V+H_{1\pi}+H_{2\pi})\tau_3}$ we need to repeat the block encoding of the Hamiltonian
\begin{eqnarray}
R_{32}\in O\left(\frac{\eta N}{\Delta^2}\left(1+\frac{\eta_s}{N}\right)\tau_3+\frac{\log(1/\delta_{32})}{\log\log(1/\delta_{32})}\right) \label{eqn:r32}
\end{eqnarray}
times and hence the gate complexity is
\begin{eqnarray}
 \g_{32}&\in& O(R_{32}\cdot \g_{32}') \nonumber \\
 &\in&O\left(\frac{\eta N}{\Delta^2}\left(1+\frac{\eta_s}{N}\right)\tau_3\left(\eta a\log N+N\log d+\log N\log\frac{N}{\delta'}+K\log\frac{1}{\delta'}\right)   \right. \nonumber \\
 &&\left.+\frac{\log(1/\delta_{32})}{\log\log(1/\delta_{32})} \left(\eta a\log N+N\log d+\log N\log\frac{N}{\delta'}+K\log\frac{1}{\delta'}\right) \right).    \label{eqn:g32}
\end{eqnarray}

\section{Trotter error and commutators}
\label{app:comm}

Let $H=\sum_{\gamma=1}^{\Gamma}H_{\gamma}$ be a time-independent operator and the evolution generated by $H$ is $e^{-it\sum_{\gamma=1}^{\Gamma}H_{\gamma}}$. Such evolutions can be approximated by product of exponentials, using product formulas like the first-order Lie-Trotter formula
\begin{eqnarray}
 \mathscr{S}_1(t)=e^{tH_{\Gamma}}\ldots e^{tH_1}
\end{eqnarray}
and higher-order Suzuki formulas \cite{1991_S} defined recursively via
\begin{eqnarray}
 \mathscr{S}_2(t)&=&e^{\frac{t}{2}H_1}\ldots e^{\frac{t}{2}H_{\Gamma}}e^{\frac{t}{2}H_{\Gamma}}\ldots e^{\frac{t}{2}H_1}, \nonumber \\
 \mathscr{S}_{2k}(t)&=&\mathscr{S}_{2k-2}^2(u_kt)\mathscr{S}_{2k-2}((1-4u_k)t)\mathscr{S}_{2k-2}^2(u_kt)
\end{eqnarray}
where $u_k=\frac{1}{4-4^{\frac{1}{2k-1}}}$. Quite a few bounds on the Trotter error has been derived before~\cite{1990_HdR, 2010_WBHS, 2015_WHWetal}, but we use the one in \cite{2021_CSTetal}, which shows the dependence on nested commutators. Specifically, the authors show that for a $p^{th}$ order Trotter-Suzuki formula, $\mathscr{S}_p(t)=e^{-itH}+\mathscr{A}(t)$, where
\begin{eqnarray}
    \|\mathscr{A}(t)\|\in O\left(\widetilde{\alpha}_{comm}t^{p+1}\right)
    \label{eqn:trotterErr},
\end{eqnarray}
if $H_{\gamma}$ are Hermitian. Also, in the above
\begin{eqnarray}
    \widetilde{\alpha}_{comm}=\sum_{\gamma_1,\gamma_2,\ldots,\gamma_{p+1}=1}^{\Gamma}\|[H_{\gamma_{p+1}},\ldots[H_{\gamma_2},H_{\gamma_1}]]\|.
    \label{eqn:alpha}
\end{eqnarray}

\begin{lemma}
Consider the following sum of nested commutator, obtained from distinct Hamiltonians from the set $\{H_1,\ldots,H_k\}$. Let $H_1',H_2',\ldots,H_{p'}'$ are $p'+1$ Hamiltonians that may or may not belong to the set.
 \begin{eqnarray}
   H_{nest}=\sum_{\gamma_1,\ldots,\gamma_{p-p'}=1}^{k}\left[H_{\gamma_{p-p'}},[H_{\gamma_{p-p'-1}},[\ldots [H_{\gamma_1},[H_{p'+1}',[\ldots[H_3',[H_2',H_1']]\ldots]] ]\ldots]\right] \nonumber
 \end{eqnarray}
 Then, 
\begin{eqnarray}
 \|H_{nest}\|\leq 2^{p-(p'+1)}\|[H_{p'+1}',[\ldots[H_3',[H_2',H_1']]\ldots]]\|\left(\sum_{i=1}^k\|H_i\|\right)^{p-p'} . \nonumber
\end{eqnarray}
\label{app:lem:nestComSum}
\end{lemma}

\begin{proof}
 Consider one group of summands as follows. Among $H_{\gamma_1},\ldots,H_{\gamma_{p-p'}}$, the number of occurrences of $H_1$ is $0\leq i_1\leq p-p'$, number of occurrences of $H_2$ is $0\leq i_2\leq p-p'-i_1$, number of occurrences of $H_3$ is $0\leq i_3\leq p-p'-i_1-i_2$ and so on, that is, number of occurrences of $H_k$ is $0\leq i_k\leq p-p'-i_1-i_2-\cdots -i_{k-1}$. Using Fact \ref{fact:nestComNorm} we can upper bound the norm of the sum of this group of summands as follows.
 \begin{eqnarray}
 && \|[H_{p'+1}',[\ldots[H_3',[H_2',H_1']]\ldots]]\|2^{p-(p'+1)}\binom{p-p'}{i_1}\|H_1\|^{i_1}\binom{p-p'-i_1}{i_2}\|H_2\|^{i_2}\ldots \nonumber \\
 &&\ldots\binom{p-p'-i_1-\cdots-i_{k-2}}{i_{k-1}}\|H_{k-1}\|^{i_{k-1}}\|H_k\|^{p-p'-\cdots-i_{k-1}} \nonumber
 \end{eqnarray}
Thus the total sum can be upper bounded as follows.
\begin{eqnarray}
&& \|H_{nest}\| \nonumber \\
&\leq&2^{p-(p'+1)}\|[H_{p'+1}',[\ldots[H_3',[H_2',H_1']]\ldots]]\|    \nonumber \\
&&\sum_{i_1=0}^{p-p'}\sum_{i_2=0}^{p-p'-i_1}\ldots\sum_{i_{k-1}=0}^{p-p'-\cdots-i_{k-2}}\binom{p-p'}{i_1}\ldots\binom{p-p'-\cdots-i_{k-2}}{i_{k-1}}\|H_1\|^{i_1}\ldots \|H_k\|^{p-p'-\cdots-i_{k-1}}  \nonumber \\
&=& 2^{p-(p'+1)}\|[H_{p'+1}',[\ldots[H_3',[H_2',H_1']]\ldots]]\|\left(\sum_{i=1}^k\|H_i\|\right)^{p-p'} \nonumber
\end{eqnarray}
\end{proof}
Thus we immediately prove Lemma~\ref{lem:alpha_comm}, that we are re-stating here again, for completeness.
\begin{lemma}
Let $H=\sum_{\gamma=1}^{\Gamma}H_{\gamma}$ and $\widetilde{\alpha}_{comm}=\sum_{\gamma_1,\gamma_2,\ldots,\gamma_{p+1}=1}^{\Gamma}\|[H_{\gamma_{p+1}},\ldots[H_{\gamma_2},H_{\gamma_1}]]\|$. Then for any integer $1\leq p'\leq p$,
\begin{eqnarray}
    \widetilde{\alpha}_{comm}\leq 2^{p-(p'+1)}\sum_{\gamma_{i_1},\gamma_{i_2},\ldots,\gamma_{i_{p'+1}}} \|[H_{\gamma_{p'+1}},[\ldots[H_{\gamma_3},[H_{\gamma_2},H_{\gamma_1}]]\ldots]]\|  \left(\sum_{\gamma=1}^{\Gamma}\|H_{\gamma}\|\right)^{p-p'}.    \nonumber
\end{eqnarray}
    \label{app:lem:alpha_comm}
\end{lemma}
Ideally, we would want to compute tight bounds for the nested commutators, preferably by exploiting some structure or properties of the Hamiltonians. Using norms do not always give tight bounds. But suppose we can compute such tight bounds till level $p'$ of nesting, while we want a bound till level $p$. Then these results can be very useful because we can look at them as if we have combined the tighter analysis till level $p'$ with a less tight analysis for the rest of the levels of nesting. In fact, we observe that in Lemma \ref{app:lem:alpha_comm} we actually group the terms with the same inner commutators till level $p'$ and then apply Lemma \ref{app:lem:nestComSum} to bound the sum of each such group, absolutely independent of the other groups. So, we can vary $p'$ for each group and appropriately vary the groupings in order to apply Lemma \ref{app:lem:nestComSum}. This can make the bound tighter for many applications. In this sense we have some flexibility. 

We have explained before how we compute the $\ell_1$ norm. Now we explain how we calculate the innermost pair-wise commutators. For all the pairs the bounds have been derived by expanding the commutators using Lemma \ref{lem:comSum} and \ref{lem:comTensor} and then using triangle inequality (Appendix \ref{app:norm}). For $[H_{\pi}, H_{V_{ee}}]$ and $[H_{\pi}, H_{V_{ne}}]$, we use the following additional lemma. So we first explain these two cases.

Similar to the particle configuration considered in this paper, let $S$ be a 3-D cubic lattice that whose each side is of length $L$. Each side has $N^{1/3}$ points and thus the inter-point spacing is $\Delta=\frac{L}{N^{1/3}}$. So we can say that the cube is sub-divided into $N$ unit-cells, each of length $\Delta$. For any two points $q=(q_x,q_y,q_z),r=(r_x,r_y.r_z)$ in the lattice, let the distance between them is denoted by $d_{qr}=\|q-r\|_2=\sqrt{(q_x-r_x)^2+(q_y-r_y)^2+(q_z-r_z)^2}$.
\begin{lemma}
Assume that $S$ consists of at least $N=1$ unit cells and has side length $\Delta>0$, then
 $$
  \sum_{q\neq r}\frac{1}{d_{qr}}\leq\frac{2\cdot N^{5/3}}{\Delta}.
 $$
 \label{lem:dqr}
\end{lemma}
\begin{proof}
 Consider the points on a 2-D lattice i.e. when $q_z-r_z=0$. We assume, that all points within the lattice have positive coordinates and one corner is $(0,0,.)$. Let us fix  $q$ to be this corner. We  ignore the third co-ordinate, because it is not relevant in 2-D plane. 
 
 We divide the points of the lattice in different sets and then compute the sum of the inverse of the distances of $q$ from the points within each set. First, consider the points within the square with length $\Delta$, cornered at $q$. We include these points in set $\mathcal{S}_1$. Apart from $q$, there are $2\times 1$ points at distance $\Delta\sqrt{1^2+0^2}$ from $q$ and $1$ point at distance $\Delta\sqrt{1^2+1^2}$ from $q$. The sum of the inverses of these distances is
 \begin{eqnarray}
 2\frac{1}{\Delta\sqrt{1^2+0^2}}+\frac{1}{\Delta\sqrt{1^2+1^2}}. \nonumber
 \end{eqnarray}
Next, we build the set $\mathcal{S}_2$, that includes all points within a square of sides $2\Delta$, cornered at $q$, but not those in $\mathcal{S}_0$. There are $2\times 2$ points at distances $\Delta\sqrt{2^2+0^2},\Delta\sqrt{2^2+1^2}$ and it is straightforward to see that for each distance there are two points, one translated in the X and the other in the Y-direction. There is 1 point at the corner, which is at distance $\Delta \sqrt{2^2+2^2}$. The sum of the inverses of these distances is
 \begin{eqnarray}
  \frac{1}{\Delta}\left[2\left(\frac{1}{\sqrt{2^2+0^2}}+\frac{1}{\sqrt{2^2+1^2}}\right)+\frac{1}{\sqrt{2^2+2^2}}\right]. \nonumber
 \end{eqnarray}
 Similarly we consider the set $\mathcal{S}_3=\left\{\text{Points within a square of sides }3\Delta,\text{ cornered at }q\right\}\setminus\left(\mathcal{S}_2\bigcup\mathcal{S}_1\right)$. There are $2\times 3$ points,  at distances $\Delta\sqrt{3^2+0^2},\Delta\sqrt{3^2+1^2},\Delta\sqrt{3^2+2^2}$ from $q$; and one corner point at distance $\Delta\sqrt{3^2+3^2}$. The sum of the inverse is
 \begin{eqnarray}
    \frac{1}{\Delta}\left[2\left(\frac{1}{\sqrt{3^2+0^2}}+\frac{1}{\sqrt{3^2+1^2}}+\frac{1}{\sqrt{3^2+2^2}}\right)+\frac{1}{\sqrt{3^2+3^2}}\right]. \nonumber
 \end{eqnarray}
We go on in this way, till we build the last set 
$$
\mathcal{S}_{N^{1/3}}=\left\{\text{Points within a square of sides }N^{1/3}\Delta,\text{ cornered at }q\right\}\setminus\left(\bigcup_{i=1}^{N^{1/3}-1}\mathcal{S}_i\right).
$$
 There are $2\times N^{1/3}$ points, at distances \\
 $\Delta\sqrt{N^{2/3}+0^2},\Delta\sqrt{N^{2/3}+1^2},\Delta\sqrt{N^{2/3}+2^2},\ldots,\Delta\sqrt{N^{2/3}+(N^{1/3}-1)^2}$ from $q$; and one corner point at distance $\Delta\sqrt{N^{2/3}+N^{2/3}}$. The sum of the inverse of these distances is
\begin{eqnarray}
 \frac{1}{\Delta}\left[2\left(\frac{1}{\sqrt{N^{2/3}+0^2}}\cdots+\frac{1}{\sqrt{N^{2/3}+(N^{1/3}-1)^2}}\right)+\frac{1}{\sqrt{N^{2/3}+N^{2/3}}}\right]. \nonumber
\end{eqnarray}
We claim that 
\begin{eqnarray}
f(k)=2\sum_{i=0}^{k-1}\left(\frac{1}{\sqrt{k^2+i^2}}\right)+\frac{1}{\sqrt{k^2+k^2}} \leq 2\qquad\text{when } k>4.
\end{eqnarray}
This is because $f(k)$ is continuous and differentiable in $[1,N']$, where $N'$ is finite. Also,
\begin{eqnarray}
    f'(k)=-2\sum_{i=0}^{k-1}k(k^2+i^2)^{-3/2}-\frac{1}{\sqrt{2}k^2}<0 \nonumber
\end{eqnarray}
in this finite interval, and hence the function $f(k)$ is monotonically decreasing. Since $f(k)\leq 2$ when $k\geq 4$, so our claim follows.

Thus, for one particular 2-D plane the sum of the inverse of the distances is at most $\frac{2\cdot N^{1/3}}{\Delta}$.

Now a 3-D cubic lattice can be generated by translations of this 2-D square lattice along the z direction. Now as we translate along the Z-direction, $q_z-r_z >0$, so the distances from $q=(0,0,0)$ increase and hence the sum of the inverse of these distances can again be bounded by $2\cdot N^{1/3}/\Delta$. Since a cube is generated by $N^{1/3}$ such translations, so
\begin{eqnarray}
 \sum_r\frac{1}{d_{0r}} \leq\frac{2\cdot N^{2/3}}{\Delta}. \nonumber
\end{eqnarray}
Hence $\sum_{q\neq r}\sum_{r}\frac{1}{d_{qr}}\leq N\frac{2\cdot N^{2/3}}{\Delta}$ and the lemma follows.
\end{proof}

Given this result, we can turn our attention to bounding the commutators of all remaining terms in the Hamiltonian.  We proceed in the following to enumerate each possible commutator that can emerge in the error bound.  These bounds will be used in our Trotter error bound estimates.
\paragraph{I. $\boldsymbol{\|[H_{\pi},H_{V_{ee}}]\|}$ and $\boldsymbol{\|[H_{\pi},H_{V_{ne}}]\|}$ : }
We know that
\begin{eqnarray}
 H_{\pi}&=&\frac{1}{2}\sum_{j=1}^{\eta}\sum_{\mu=1}^3\sum_{q=1}^N\left(-\id\otimes\nabla_{j,\mu}^2\otimes\id+ i\frac{2}{c} \id\otimes\nabla_{j,\mu}\otimes A_{q,\mu}+\frac{1}{c^2}\id\otimes\id\otimes A_{q,\mu}^2\right) \nonumber \\
 H_{V_{ee}}&=&\frac{1}{\Delta}\sum_{j'<k=1}^{\eta}\sum_{x_1,x_2=1}^N\id\otimes\frac{1}{\|x_1-x_2\|_2}(\ket{x_1}\bra{x_1}_{j'}\otimes\ket{x_2}\bra{x_2}_k)\otimes\id   \nonumber \\
 H_{V_{ne}}&=& -\frac{1}{\Delta}\sum_{j'=1}^{\eta}\sum_{\kappa=1}^K\sum_{x=1}^N\id\otimes\frac{Z_{\kappa}}{\|x-R_{\kappa}\|_2}\ket{x}\bra{x}_j\otimes\id \nonumber
\end{eqnarray}
Let $\|x_1-x_2\|_2=d_{x_1x_2}$ and $\|x-R_{\kappa}\|_2=d_{x\kappa}$. Using Lemma \ref{lem:comSum} and \ref{lem:comTensor}, and the fact that $[A^2,\id]=0$, we get
\begin{eqnarray}
 &&[H_{\pi},H_{V_{ee}}] \nonumber \\
 &=&-\frac{1}{2\Delta}\sum_{\substack{j=j'\text{ or }k\\\mu,q,x_1,x_2}}\id\otimes\frac{1}{d_{x1x2}}[\nabla_{j,\mu}^2,\ket{x_1}\bra{x_1}_{j'}\ket{x_2}\bra{x_2}_k]\otimes\id \nonumber \\
 && +\frac{i2}{c\Delta}\sum_{\substack{j=j'\text{ or }k\\\mu,q,x_1,x_2}}\id\otimes\frac{1}{d_{x1x2}}[\nabla_{j,\mu},\ket{x_1}\bra{x_1}_{j'}\ket{x_2}\bra{x_2}_k]\otimes A_{q,\mu} \nonumber   
\end{eqnarray}
From Lemma \ref{lem:dqr} we know $\sum_{x1x2}\frac{1}{d_{x_1x_2}}\leq\frac{2N^{5/3}}{\Delta}$. The spectral norm of the commutator is bounded as follows.
\begin{eqnarray}
&& \|[H_{\pi},H_{V_{ee}}]\| \nonumber \\
&\leq& 2\left(\frac{1}{2\Delta}\frac{2\cdot 3\eta(\eta-1)N}{2}\|\nabla_j^2\|\sum_{x_1,x_2}\frac{1}{d_{x1x2}}+\frac{2}{c\Delta}\frac{2\cdot 3\eta(\eta-1)N}{2}\|\nabla_j\|\|A_{\mu}\|\sum_{x_1,x_2}\frac{1}{d_{x1x2}} \right)\nonumber \\
&\leq& \frac{3\eta(\eta-1)N}{\Delta}\frac{4\pi^2}{3h^2}\frac{N^{5/3}}{\Delta}+\frac{12\eta(\eta-1)N}{c\Delta}\frac{\ln (2a^2)}{h}\frac{2\pi}{\Delta}\frac{N^{5/3}}{\Delta}  \nonumber \\
&\leq&\frac{4\pi\eta(\eta-1)N^{8/3}}{h^2\Delta^2}\left(\pi+\frac{6h\ln (2a^2)}{c\Delta}\right) \nonumber      
\end{eqnarray}
With similar arguments we can prove
\begin{eqnarray}
    \|[H_{\pi},H_{V_{ne}}]\|\leq \frac{4\pi\eta N^{5/3}KZ_{max}}{h^2\Delta^2}\left(\pi+\frac{6h\ln (2a^2)}{c\Delta}\right),
\end{eqnarray}
where $(Z_{\kappa})_{max}=Z_{max}$. In this case we have $\sum_{\kappa}Z_{\kappa}\sum_x\frac{1}{d_{x\kappa}}\leq Z_{max}\sum_{\kappa,x}\frac{1}{d_{x\kappa}}$. Similar to Lemma \ref{lem:dqr} we can prove that $\sum_{x}\frac{1}{d_{\kappa r}}\leq\frac{N^{2/3}}{\Delta}$, for some fixed $\kappa$. Hence $\sum_{\kappa x}\frac{1}{d_{x\kappa}}\leq K\frac{N^{2/3}}{\Delta}$. Here we make another observation, which has been important in the groupings we make. If $H_{\pi}=H_{1\pi}+H_{2\pi}+H_{3\pi}$ as defined in Equation \ref{eqn:Hpi}, then $[H_{3\pi},H_{V_{ee}}]=[H_{3\pi},H_{V_{ne}}]=0$.

\paragraph{II. $\boldsymbol{\|[H_s,H_{\pi}]\|}$ : } We know that
\begin{eqnarray}
 H_s&=&-\frac{1}{c}\sum_{j=1}^{\eta}\sum_{q=1}^N\sum_{\mu\neq\nu\neq\xi=1}^3\sigma_{j,\mu}\otimes\id\otimes\left(\nabla_{\nu}A_{q,\xi}-\nabla_{\xi}A_{q,\nu}\right)  \nonumber 
\end{eqnarray}
Using Lemma \ref{lem:comSum} and \ref{lem:comTensor} and the facts that $[(\nabla_{\nu}A_{q,\xi}-\nabla_{\xi}A_{q,\nu}),A_{q',\mu'}]=[(\nabla_{\nu}A_{q,\xi}-\nabla_{\xi}A_{q,\nu}),A_{q',\mu'}^2]=$ if $q\neq q'$ and $\mu'\neq\nu,\xi$, we have
\begin{eqnarray}
    [H_s,H_{\pi}]&=&i\frac{1}{c^2}\sum_{\substack{j,j',q\\\mu\neq\nu\neq\xi\\\mu'=\nu\text{ or }\xi}}  \sigma_{j,\mu}\otimes\nabla_{j',\mu'}\otimes [(\nabla_{\nu}A_{q,\xi}-\nabla_{\xi}A_{q,\nu}),A_{q,\mu'}]      \nonumber \\
&&+\frac{1}{2c^3}\sum_{\substack{j,j',q\\\mu\neq\nu\neq\xi\\\mu'=\nu\text{ or }\xi}}\sigma_{j,\mu}\otimes\id\otimes [(\nabla_{\nu}A_{q,\xi}-\nabla_{\xi}A_{q,\nu}),A_{q,\mu'}^2]         \nonumber
\end{eqnarray}
and hence 
\begin{eqnarray}
 &&\|[H_s,H_{\pi}]\|   \nonumber \\
 &\leq& 2\left(\frac{1}{c^2}6N\eta^2\|\nabla_{j',\mu'}\|\|(\nabla_{\nu}A_{q,\xi}-\nabla_{\xi}A_{q,\nu})\|\|A_{q,\mu'}\|+\frac{1}{2c^3}6N\eta^2\|(\nabla_{\nu}A_{q,\xi}-\nabla_{\xi}A_{q,\nu})\|\|A_{q,\mu'}^2\| \right)  \nonumber \\
&\leq& \frac{6\eta^2 N}{c^2}\left(2\frac{2\ln a+\gamma}{h}\frac{4\pi(2\ln a+\gamma)}{h\Delta}\frac{2\pi}{\Delta}+\frac{1}{c}\frac{4\pi(2\ln a+\gamma)}{h\Delta}\frac{4\pi^2}{\Delta^2}\right)   \nonumber   \\
&=&\frac{96\pi^2\eta^2 N(2\ln a+\gamma)}{hc^2\Delta^2}\left(\frac{2\ln a+\gamma}{h}+\frac{\pi}{c\Delta}\right)\leq \frac{96\pi^2\eta^2 N\ln (2a^2)}{hc^2\Delta^2}\left(\frac{\ln (2a^2)}{h}+\frac{\pi}{c\Delta}\right).   \label{eqn:HspinHpi}
\end{eqnarray}

\paragraph{III. $\boldsymbol{\|[H_s,H_{V_{ee}}]\|}$ and $\boldsymbol{\|[H_s,H_{V_{ne}}]\|}$ : } Expanding, using Lemma \ref{lem:comSum} and \ref{lem:comTensor}, we find that both these commutators are 0.

\paragraph{IV. $\boldsymbol{\|[H_f,H_{V_{ee}}]\|}$ and $\boldsymbol{\|[H_f,H_{V_{ne}}]\|}$ : } We know that
\begin{eqnarray}
 H_{f1}&=&\frac{1}{2}\sum_{q'=1}^N\sum_{\mu'=1}^3\id\otimes\id\otimes E_{q',\mu'}^2 \nonumber \\
 H_{f2}&=&-\sum_{q'=1}^N\sum_{\mu'\neq\nu'=1}^3\id\otimes\id\otimes W_{q',\mu',\nu'}^2.   \nonumber 
\end{eqnarray}
Using Lemma \ref{lem:comSum} and \ref{lem:comTensor} we find both these commutators are 0.

\paragraph{V. $\boldsymbol{\|[H_{f1},H_{f2}]\|}$ : } 
For a tighter bound on the commutator, we consider the following definitions of $E_{q,\mu}^2$ and $U_{q,\mu}$, as given in Section \ref{sec:ham}. 
\begin{eqnarray}
    E_{q,\mu}^2&=&\sum_{\epsilon=-\Lambda}^{\Lambda-1}\epsilon^2\ket{\epsilon}\bra{\epsilon}_{q,\mu}     \\
    U_{q,\mu}&=&\sum_{\epsilon=-\Lambda}^{\Lambda-1}\ket{\epsilon+1}\bra{\epsilon}_{q,\mu}  
\end{eqnarray}
The commutator between these two operators is
\begin{eqnarray}
    \|[E_{q,\mu}^2,U_{q,\mu}]\|&=&\|E_{q,\mu}^2U_{q,\mu}-U_{q,\mu}E_{q,\mu}^2\| \nonumber \\
    &=&\|\sum_{\epsilon=-\Lambda}^{\Lambda-1}(\epsilon+1)^2\ket{\epsilon+1}\bra{\epsilon}-\sum_{\epsilon=-\Lambda}^{\Lambda-1}\epsilon^2\ket{\epsilon+1}\bra{\epsilon} \|   \nonumber \\
    &=&\|\sum_{\epsilon=-\Lambda}^{\Lambda-1}(2\epsilon+1)\ket{\epsilon+1}\bra{\epsilon}\|=2\Lambda-1
\end{eqnarray}
Now from its definition the plaquette operator $W_{q',\mu',\nu'}^2$ is the product of 4 such $U$ operators that act on the sides of a plaquette. So $E_{q,\mu}^2$ has a non-zero commutator with $W_{q',\mu',\nu'}^2$ if and only if the link $(q,\mu)$ is any one of the 4 links of this plaquette. Thus,
\begin{eqnarray}
    \|[H_{f1},H_{f2}]\|\leq 3N\cdot 2(2\Lambda-1)\leq 12N\Lambda
\end{eqnarray}

\paragraph{VI. $\boldsymbol{\|[H_f,H_{\pi}]\|}$ : } Using similar arguments as before to get the indices for non-zero commutators, we have
\begin{eqnarray}
&& [H_f,H_{\pi}]  \nonumber \\
&=&-\frac{i}{2c}\sum_{j,\mu,q}\id\otimes\nabla_{j,\mu}\otimes [E_{q,\mu}^2,A_{q,\mu}]+\frac{1}{2c^2}\sum_{j,\mu,q}\id\otimes\id\otimes [E_{q,\mu}^2,A_{q,\mu}^2]   \nonumber \\
&&-i\frac{2}{c}\sum_{\substack{j,\mu\neq\nu,q\\ \q'=q\text{ or }q+1\\\mu'=\mu\text{ or }\nu}}\id\otimes\nabla_{j,\mu'}\otimes [W_{q,\mu,\nu}^2,A_{q',\mu'}]-\frac{1}{c^2}\sum_{\substack{j,\mu\neq\nu,q\\ \q'=q\text{ or }q+1\\\mu'=\mu\text{ or }\nu}}\id\otimes\id\otimes [W_{q,\mu,\nu}^2,A_{q',\mu'}^2]   \nonumber
\end{eqnarray}
and so,
\begin{eqnarray}
 &&\|[H_f,H_{\pi}]\| \nonumber \\
 &\leq&2\left(\frac{3\eta N}{2c}\|\nabla_{j\mu}\|\|E_{q\mu}^2\|\|A_{q,\mu}\|+\frac{3\eta N}{2c^2}\|E_{q\mu}^2\|\|A_{q,\mu}^2\|+\frac{24\eta N}{c}\|\nabla_{j\mu'}\|\|W_{q,\mu,\nu}^2\|\|A_{q',\mu'}\|  \right. \nonumber\\
 &&\left. +\frac{12\eta N}{c^2}\|W_{q,\mu,\nu}^2\|\|A_{q',\mu'}^2\|  \right) \nonumber \\
 &\leq&\frac{3\eta N}{c}\left(\frac{2\ln a+\gamma}{h}\Lambda^2\frac{2\pi}{\Delta}+\frac{1}{c}\Lambda^2\frac{4\pi^2}{\Delta^2}+16\frac{2\ln a+\gamma}{h}2\frac{2\pi}{\Delta}+\frac{8}{c}2\cdot\frac{4\pi^2}{\Delta^2} \right)     \nonumber \\
 &=&\frac{6\pi\eta N\Lambda^2}{c\Delta}\left(\frac{2\ln a+\gamma}{h}+\frac{2\pi}{c\Delta}\right)+\frac{198\pi\eta N}{c\Delta}\left(\frac{2\ln a +\gamma}{h}+\frac{\pi}{c\Delta}\right)  \nonumber \\
 &\leq&\frac{6\pi\eta N}{c\Delta}\left(\left(\frac{\ln(2a^2)}{h}+\frac{2\pi}{c\Delta}\right)(\Lambda^2+33)-\frac{33\pi}{c\Delta}\right)
 \label{eqn:HfHpi}
\end{eqnarray}

\paragraph{VII. $\boldsymbol{\|[H_s,H_{f1}]\|}$ and $\boldsymbol{\|[H_s,H_{f2}]\|}$ : } We know that
\begin{eqnarray}
 H_s&=&-\frac{1}{c}\sum_{j=1}^{\eta}\sum_{q=1}^N\sum_{\mu\neq\nu\neq\xi=1}^3\sigma_{j,\mu}\otimes\id\otimes\left(\nabla_{\nu}A_{q,\xi}-\nabla_{\xi}A_{q,\nu}\right)  \nonumber \\
 H_{f1}&=&\frac{1}{2}\sum_{q'=1}^N\sum_{\mu'=1}^3\id\otimes\id\otimes E_{q',\mu'}^2 \nonumber \\
 H_{f2}&=&-\sum_{q'=1}^N\sum_{\mu'\neq\nu'=1}^3\id\otimes\id\otimes W_{q',\mu',\nu'}^2.   \nonumber 
\end{eqnarray}
$[(\nabla_{\nu}A_{q,\xi}-\nabla_{\xi}A_{q,\nu}),E_{q',\mu'}^2]=0$ if $q\neq q'$ and $\mu'\neq\xi,\nu$. Also, $[(\nabla_{\nu}A_{q,\xi}-\nabla_{\xi}A_{q,\nu}), W_{q',\mu',\nu'}]\neq 0$ if the link $(q,\xi)$ or $(q,\nu)$ is equal to any of the links $(q',\mu')$, $(q'+1_{\mu'},\nu')$, $(q'+1_{\nu'},\mu')$, $(q',\nu')$. This can happen if $\nu'$ (or $\mu'$) is either $\nu,\xi$ and the other one varies; and if $\nu'=\nu$ (say) then $(q,\nu)=(q',\nu')$ or $(q,\nu)=(q'+1_{\mu'},\nu')$. In the following equations we refer to the latter condition as $q=q'$ or $q'+1$, for brevity.
Using Lemma \ref{lem:comSum} and \ref{lem:comTensor} we have the following.
\begin{eqnarray}
[H_s,H_{f1}]&=& -\frac{1}{2c}\sum_{\substack{q,j\\\mu\neq\nu\neq\xi\\\mu'=\nu\text{ or }\xi}}\sigma_{j,\mu}\otimes\id\otimes [(\nabla_{\nu}A_{q,\xi}-\nabla_{\xi}A_{q,\nu}),E_{q,\mu'}^2] \nonumber \\
\left[H_s,H_{f2}\right]&=& \frac{1}{c}\sum_{\substack{q,j\\\mu\neq\nu\neq\xi,\mu'\\\ell=q,q+1}}\sigma_{j,\mu}\otimes\id\otimes [(\nabla_{\nu}A_{q,\xi}-\nabla_{\xi}A_{q,\nu}),W_{q,\mu'\nu}^2] \nonumber
 \end{eqnarray}
We have used the following facts.  Using the bounds in Table \ref{tab:norm} we have 
\begin{eqnarray}
\|[H_s,H_{f1}]\|&\leq&2\frac{1}{2c}6\eta N\|\nabla_{\nu}A_{q,\xi}-\nabla_{\xi}A_{q,\nu})\|\|E_{q\mu'}^2\|\leq \frac{24\pi\eta N\Lambda^2\ln (2a^2)}{ch\Delta}   \nonumber \\
\|[H_s,H_{f2}]\|&\leq&2\frac{1}{c}18\eta N\|\nabla_{\nu}A_{q,\xi}-\nabla_{\xi}A_{q,\nu})\|\|W_{q',\mu',\nu'}^2\|\leq \frac{288\pi\eta N\ln (2a^2)}{ch\Delta}  \nonumber 
\end{eqnarray}

In Table \ref{tab:comm} we summarize the bounds on all the necessary pair-wise commutators derived by us. 

\section{State preparation algorithm}
\label{app:statePrep}

In this section we describe an algorithm to prepare a state proportional to the following.
\begin{eqnarray}
 \sum_{\vec{v}\in G}\frac{1}{\sqrt{\|\vec{v}\|}}\ket{\vec{v}}\qquad\text{where }G=[-N^{1/3},N^{1/3}]^3\setminus\{0,0,0\}
\end{eqnarray}
We follow the algorithm in \cite{2021_SBWetal}, \cite{2019_BBMN} with appropriate changes to take care of the difference in weights. The approach is to use a hierarachy of nested boxes in $G$ indexed by $\mu$, each box is larger than the previous by a factor of 2. For each box $\mu$ we prepare a set of $\vec{v}$ values in that cube. We use 8 registers $\ket{\mu}\ket{v_x}\ket{v_y}\ket{v_z}\ket{m}\ket{0}$ to hold this state. These subsystems are used as follows.
\begin{enumerate}
 \item[(a)] $\ket{\mu}$ indexes the box used.
 \item[(b)] $\ket{v_x,v_y,v_z}$ are the 3 components of $\vec{v}$ given as signed integers.
 \item[(c)] $\ket{m}$ is an ancilla in an equal superposition used to give the correct amplitude via an inequality test.
 \item[(d)] $\ket{0}$ flags that the state preparation is successful. 
\end{enumerate}
There are 4 aspects due to which we have a failure probability.
\begin{enumerate}
 \item The preparation of $\mu$ can fail.
 \item The signed integers can be negative zero, which is not allowed.
 \item There is a failure probability associated to the test whether $\vec{v}$ is inside a certain box.
 \item An inequality test made during the preparation also introduces a probability of failing.
\end{enumerate}
Let $n_p=1+\log_2\left(N^{1/3}+1\right)$ is the number of qubits required to represent $v_x, v_y$ and $v_z$, i.e. each will give numbers from $-(2^{n_p-1}-1)$ to $2^{n_p-1}-1$. The state preparation procedure can be summarized in the following steps.

\textbf{Step I} : We prepare a superposition state
\begin{eqnarray}
 \sqrt{\frac{3}{4^{n+1}-16}}\sum_{\mu=2}^{n_p}2^{\mu}\ket{\mu}
\end{eqnarray}
which ensures that we obtain the correct weighting for each cube. We use a unary encoding for $\ket{\mu}$. We use a ladder of $n_p$ controlled-H gates. More detail can be found in \cite{2021_SBWetal}. Each H-gate can be implemented with 2 H, 2 T and 1 CNOT.

 \textbf{Step II} : Controlled by $\mu$, we apply H gate on $\mu$ of the qubits representing $v_x, v_y, v_z$ to represent the values from $-(2^{\mu-1}-1)$ to $2^{\mu-1}-1$. We require at most $3n_p$ controlled-H gates. We will flag a minus zero as a failure. This can be done by checking each $\ket{v_x}, \ket{v_y}$ and $\ket{v_z}$, whether the sign bit is 1 and the remaining bits are 0. This requires 3 pairs of compute-uncompute $C^{n_p}X$-gates. Decomposing these, we require $2(4n_p-8)$ T, $2(4n_p-7)$ CNOT and $n_p-1$ ancillae.

The total number of combinations before flagging the failure is $2^{3\mu}$ and so the squared amplitude is the inverse of this. The state at this stage is
\begin{eqnarray}
 \sqrt{\frac{3}{4^{n+1}-16}}\sum_{\mu=2}^n\sum_{v_x,v_y,v_z=-(2^{\mu-1}-1)}^{2^{\mu-1}-1}2^{-\mu/2}\ket{\mu}\ket{v_x,v_y,v_z}.
\end{eqnarray}

\textbf{Step III} : We test whether $|v_x|, |v_y|, |v_z|<2^{\mu-2}$. If they are, then the point is inside the box for the next lower value of $\mu$, and we flag failure on the last ancilla qubit. For $\mu=2$ this implies that we test whether $\vec{v}=\vec{0}$, which we exclude. This requries testing if all the three qubits for $v_x,v_y,v_z$ are 0. Since the qubits tested are dependent on $\mu$, so the complexity is $O(n_p)$. 

Let $B_{\mu}$ (for box $\mu$) is the set of $\vec{v}$ such that the absolute values of $v_x, v_y, v_z$ are less than $2^{\mu-1}$, but it is not the case that they are all less than $2^{\mu-2}$. That is,
\begin{eqnarray}
 B_{\mu}=\left\{\vec{v}:(0\leq |v_x|< 2^{\mu-1})\bigwedge (0\leq |v_y|< 2^{\mu-1})\bigwedge (0\leq |v_z|< 2^{\mu-1}) \right. \nonumber \\
 \left. \bigwedge\left((|v_x|\geq 2^{\mu-2})\vee (|v_y|\geq 2^{\mu-2})\vee (|v_z|< 2^{\mu-2})\right) \right\} \nonumber
\end{eqnarray}
The state, excluding the failures, at this stage is
\begin{eqnarray}
 \sqrt{\frac{3}{4^{n+1}-16}}\sum_{\mu=2}^n\sum_{\vec{v}\in B_{\mu}}\frac{1}{2^{\mu/2}}\ket{\mu}\ket{v_x,v_y,v_z}.
\end{eqnarray}

\textbf{Step IV} : We prepare an ancilla register in an equal superposition of $\ket{m}$ for $m=0$ to $M-1$, where $M$ is a power of 2 and is chosen to be large enough to provide a sufficiently accurate approximation of the overall state preparation. This can superposition can be done entirely with H gates. Then, we test the inequality
\begin{eqnarray}
 \frac{2^{\mu-2}}{\|\vec{v}\|} >\frac{m}{M}. \nonumber
\end{eqnarray}
The left hand side can be as large as 1 in this region, because we can have just one of $v_x, v_y, v_z$ as large as $2^{\mu-2}$, and the other two equal to 0. That is, we are the center of a face of the inner cube. To avoid costly divisions, we test the following equivalent inequality.
\begin{eqnarray}
 \left(2^{\mu-2}\cdot M\right)^2 > m^2 \left(v_x^2+v_y^2+v_z^2\right)   \nonumber
\end{eqnarray}
The number of values of $m$ satisfying the above inequality is $Q=\left\lceil\frac{M2^{\mu-2}}{\|\vec{v}\|}\right\rceil$. The resulting state, at this stage, ignoring the part that fails, is
\begin{eqnarray}
 \sqrt{\frac{3}{M(4^{n+1}-16)}}\sum_{\mu=2}^n\sum_{\vec{v}\in B_{\mu}}\sum_{m=0}^{Q-1}\frac{1}{2^{\mu/2}}\ket{\mu}\ket{v_x,v_y,v_z}\ket{m}
\end{eqnarray}
The square of the amplitude for each $\vec{v}$ will then be
\begin{eqnarray}
 \frac{3\lceil M2^{\mu-2}/\|\vec{v}\| \rceil}{M(4^{n+1}-16)2^{\mu}}\approx\frac{3}{4(4^{n+1}-16)}\frac{1}{\|\vec{v}\|}
\end{eqnarray}
and hence the amplitude for each $\vec{v}$ will be proportional to $1/\sqrt{\|\vec{v}\|}$. 

Now we consider the error in the state preparation due to the finite value of $M$. The relevant quantity is the sum of the errors in the squared amplitudes, as that gives the error in the weightings of the operations applied to the target state. That error is upper bounded by
\begin{eqnarray}
 \frac{3}{M(4^{n+1}-16)}\sum_{\mu=2}^n\sum_{\vec{v}\in B_{\mu}}\frac{1}{2^{\mu}}< \frac{3}{M(4^{n+1}-16)}\sum_{\mu=2}^n2^{2\mu}=\frac{1}{M}.
\end{eqnarray}
If $n_M=\lceil\log_2M\rceil$, then we require $O(n_p^2+n_p+n_Mn_p+n_M)$ gates for this step \cite{2021_SBWetal}. If we take $n_M=\log(1/\delta')$, for some $\delta'>0$, then the gate complexity for this state preparation procedure is in $O\left(\log N+\log\frac{1}{\delta'}+\log^2N+\log N\log\frac{1}{\delta'}\right)\in O\left(\log N\log\frac{N}{\delta'}\right)$.

\end{document}